\documentclass{article}

\usepackage[margin=1.2in]{geometry}
\usepackage[numbers]{natbib}
\usepackage{algorithm}
\usepackage{algorithmic}
\usepackage[colorlinks,linkcolor=magenta,citecolor=blue, pagebackref=true]{hyperref}

\title{Nonparametric extensions of randomized\\response for private confidence sets}
\author{Ian Waudby-Smith, Zhiwei Steven Wu, and Aaditya Ramdas\vspace{0.1in}\\
  Carnegie Mellon University \vspace{0.05in}\\
  \texttt{\{ianws,zstevenwu,aramdas\}@cmu.edu}}

\usepackage[utf8]{inputenc}
\usepackage{amsmath, amsthm, amssymb, amsfonts}
\usepackage{dsfont}
\usepackage{graphicx}
\usepackage[makeroom]{cancel}
\usepackage[shortlabels]{enumitem}
\usepackage{mathabx}
\usepackage{thmtools}
\usepackage{thm-restate}
\usepackage[capitalize]{cleveref}

\declaretheorem[name=Theorem]{theorem}

\declaretheorem[name=Lemma]{lemma}
\declaretheorem[name=Corollary]{corollary}

\theoremstyle{plain}
\theoremstyle{definition}
\newtheorem{definition}[theorem]{Definition}

\theoremstyle{remark}
\newtheorem{remark}[theorem]{Remark}

\newcommand{\PP}{\mathbb P}
\newcommand{\RR}{\mathbb R}
\newcommand{\EE}{\mathbb E}
\newcommand{\NN}{\mathbb N}

\newcommand{\Pcal}{\mathcal P}
\newcommand{\Qcal}{\mathcal Q}
\newcommand{\Fcal}{\mathcal F}
\newcommand{\Gcal}{\mathcal G}
\newcommand{\Zcal}{\mathcal Z}

\newcommand{\Xcal}{\mathcal X}

\newcommand{\Kcal}{\mathcal K}

\newcommand{\Hcal}{\mathcal H}

\newcommand{\seq}[4]{(#1_{#2})_{{#2}={#3}}^{#4}}
\newcommand{\infseq}[3]{\seq{#1}{#2}{#3}{\infty}}
\newcommand{\infseqt}[1]{(#1)_{t=1}^\infty}
\newcommand{\prodseq}[3]{\prod_{#2 = 1}^{#3} #1_{#2}}

\newcommand{\ci}{CI}

\newcommand{\cs}{CS}

\newcommand{\1}{\mathds{1}}

\usepackage{setspace}

\newcommand{\dd}{\mathrm{d}}

\usepackage{color-edits}
\addauthor{sw}{blue}
\addauthor{iws}{olive}

\newcommand{\unifNoise}{\mathcal U}
\newcommand{\lapNoise}{\mathcal L}
\newcommand{\thickcline}[1]{%
    \@thickcline #1\@nil%
}
\newcommand{\eps}{\varepsilon}

\newcommand{\NPRR}{\texttt{NPRR}}

\newcommand{\PMEB}{\mathrm{EB}}
\newcommand{\Hoeff}{\mathrm{H}}
\newcommand{\opt}{\mathrm{opt}}

\newcommand{\GK}{\mathrm{GK}}

\newcommand{\PcalNPRR}[2]{\Pcal_{#1}^{#2}}
\newcommand{\PcalNPRRstarn}{\PcalNPRR{\mu^\star}{n}}
\newcommand{\PcalNPRRstarinf}{\PcalNPRR{\mu^\star}{\infty}}

\newcommand{\QcalNPRR}[2]{\Qcal_{#1}^{#2}}
\newcommand{\QcalNPRRstarn}{\QcalNPRR{\mu^\star}{n}}
\newcommand{\QcalNPRRstarinf}{\QcalNPRR{\mu^\star}{\infty}}

\newcommand{\ceil}{\mathrm{ceil}}
\newcommand{\floor}{\mathrm{floor}}

\newcommand{\EB}{\mathrm{EB}}

\newcommand{\onefone}{{_1F_1}}

\newcommand{\lpci}{LPCI}
\newcommand{\lpcs}{LPCS}

\newcommand{\Var}{\mathrm{Var}}
\newcommand{\Cov}{\mathrm{Cov}}

\newcommand{\ciplotwidth}{0.55\textwidth}
\newcommand{\nprrwidth}{0.4\textwidth}

\begin{document}
\setcounter{tocdepth}{1}
\makeatletter
\renewcommand\tableofcontents{%
    \@starttoc{toc}%
}
\makeatother

\maketitle

\begin{abstract}
This work derives methods for performing nonparametric, nonasymptotic statistical inference for population means under the constraint of local differential privacy (LDP). Given bounded observations $(X_1, \dots, X_n)$ with mean $\mu^\star$ that are privatized into $(Z_1, \dots, Z_n)$, we present confidence intervals (CI) and time-uniform confidence sequences (CS) for $\mu^\star$ when only given access to the privatized data. To achieve this, we study a nonparametric and sequentially interactive generalization of Warner's famous ``randomized response'' mechanism, satisfying LDP for arbitrary bounded random variables, and then provide CIs and CSs for their means given access to the resulting privatized observations. For example, our results yield private analogues of Hoeffding's inequality in both fixed-time and time-uniform regimes. We extend these Hoeffding-type CSs to capture time-varying (non-stationary) means, and conclude by illustrating how these methods can be used to conduct private online A/B tests.
\end{abstract}


\tableofcontents

\newpage
\section{Introduction}
It is easier than ever for mobile apps and web browsers to collect massive amounts of sensitive data about individuals.
\textit{Differential privacy} (DP) provides a framework that leverages statistical noise to limit the risk of sensitive information disclosure \cite{dwork2006calibrating}. The goal of private data analysis is to extract meaningful population-level information from the data (whether in the form of machine learning model training, statistical inference, etc.) while preserving the privacy of individuals via DP. In particular, this paper will focus on statistical inference (e.g. confidence intervals and $p$-values) for population means under DP constraints.

As motivating examples, suppose a city wishes to survey households to calculate the approval rating of their mayor, or an IT company aims to understand whether a redesigned homepage will lead to the average user spending more time on it. Both problems can be framed as estimating the mean of some (potentially large) population, but it may be  infeasible to query every single household or all possible website users. Fortunately, a \textit{sample mean} can still be used to estimate the population mean with some degree of precision. For example, a city may randomly choose households to query, or the technology company may show 10\% of users the redesigned webpage at random. This is often referred to as ``A/B testing'', and we expand on this application under privacy constraints in Section~\ref{section:a/b-testing}. When making decisions, however,
it is crucial to both calculate sample means \textit{and} quantify the uncertainty in those estimates (e.g. using confidence intervals, reviewed in Section~\ref{section:background-cs-martingale}). However, calculating confidence intervals under local differential privacy constraints (defined in Section~\ref{section:background-localdp}) poses a unique statistical challenge, because these intervals must incorporate both the uncertainty introduced from random sampling \textit{and} from the privacy mechanism. This paper studies and provides a nonparametric solution to precisely this challenge.

\subsection{Background I:~Local Differential Privacy}\label{section:background-localdp}
There are two main models of privacy within the DP framework: \textit{central} and \textit{local} DP (LDP) \citep{dwork2006calibrating,kasiviswanathan2011can,dwork2014algorithmic}. The former involves a centralized data aggregator that is trusted with constructing privatized output from raw data, while the latter performs privatization at the ``local'' or ``personal'' level (e.g. on an individual's smartphone before leaving the device) so that trust need not be placed in any data collector. Both models have their advantages and disadvantages: LDP is a more restrictive model of privacy and thus in general requires more noise to be added. On the other hand, the stronger privacy guarantees that do not require a trusted central aggregator make LDP an attractive framework in practice. This paper deals exclusively with LDP.

Making our setup more precise, suppose $X_1, X_2, \dots$ is a (potentially infinite) sequence of $[0, 1]$-valued random variables. We could instead have assumed boundedness on any known interval $[a,b]$ since we can always translate and scale the interval to $[0, 1]$ via the transformation $x \mapsto (x-a)/(b-a)$. We will refer to $\infseq Xt1$ as the ``raw'' or ``sensitive'' data that are yet to be privatized. Following the notation of \citet{duchi2013local-FOCS} the privatized views $Z_1, Z_2, \dots$ of $X_1, X_2, \dots$, respectively are generated by a sequence of conditional distributions $Q_1, Q_2, \dots$ which we refer to as the \textit{privacy mechanism}. Throughout this paper, we will allow this privacy mechanism to be \textit{sequentially interactive}, meaning that the distribution $Q_i$ of $Z_i$ may depend on the past privatized observations $Z_1^{i-1}:= (Z_1,\dots,Z_{i-1})$ \citep{duchi2013local-FOCS}.
In other words, the privatized view $Z_i$ of $X_i$ has a conditional  distribution $Q_i(\cdot \mid X_i = x, Z_1^{i-1} = z_1^{i-1})$.
Following \citet{duchi2013local-FOCS,duchi2018minimax} we say that $Q_i$ satisfies $\varepsilon$-local differential privacy if for all $z_1, \dots, z_{i-1} \in [0, 1]$ and $x, \widetilde x \in [0, 1]$, the following likelihood ratio is uniformly bounded:
\begin{equation}
\small
\label{eq:seqInteractiveLikelihood}
    \sup_{z \in [0, 1]} \frac{q_i(z \mid X_i = x, Z_1^{i-1} = z_1^{i-1})}{q_i(z \mid X_i = \widetilde x, Z_1^{i-1} = z_1^{i-1})} \leq \exp \{ \varepsilon \},
\end{equation}
where $q_i$ is the density (or Radon-Nikodym derivative) of $Q_i$ with respect to some dominating measure.
In the non-interactive case where the dependence on $Z_1^{i-1}$ is dropped,~\eqref{eq:seqInteractiveLikelihood} simplifies to the usual $\eps$-LDP definition \citep{dwork2014algorithmic}.
To put $\eps > 0$ in a real-world context, Apple uses privacy levels in the range of $\eps \in \{2,4,8\}$ on macOS and iOS devices for various $\eps$-LDP data collection tasks, including health data type usage, emoji suggestions, and lookup hints \citep{appleEps}. See \cref{fig:eps} to intuit how $\eps$ affects the widths of confidence intervals that we develop.
Next, we review time-uniform confidence sequences and how they differ from fixed-time confidence intervals.

\begin{figure}[!htbp]
    \centering
    \includegraphics[width=\ciplotwidth]{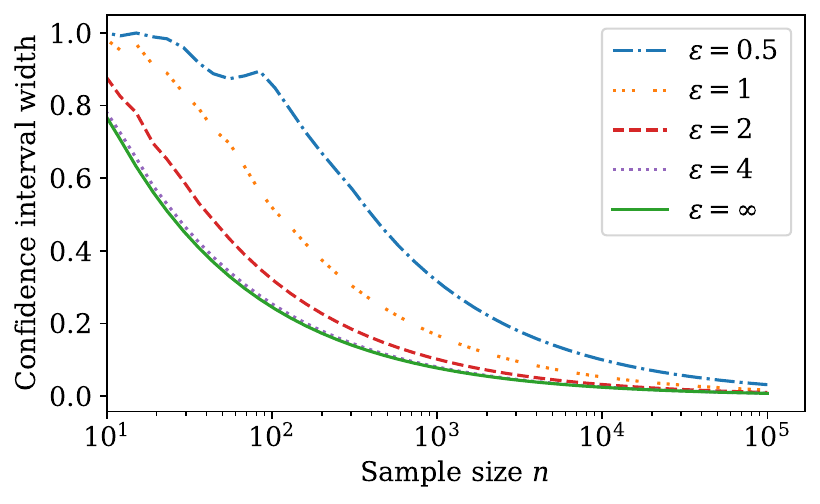}
    \caption{Widths of private 90\%-\ci{}s for the mean of a uniform distribution using our private Hoeffding \ci{} given in \eqref{eq:hoeffding-ci-simple-case} for various levels of $\eps$ ranging from $\eps = 0.5$ (very private) to $\eps = \infty$ (no privacy). Unsurprisingly, less privacy leads to sharper inference, but notice that inference is still practical, especially for $\eps \geq 2$. For context, Apple uses $\eps \in \{2,4,8\}$ for various data collection tasks on iPhones \citep{appleEps}. At these levels of privacy, our \ci{}s perform nearly as well as --- and are in some cases indistinguishable from --- the non-private Hoeffding \ci{}.}
    \label{fig:eps}
    \vspace{-0.1in}
\end{figure}

\subsection{Background II:~Confidence Sequences}\label{section:background-cs-martingale}
One of the most fundamental tasks in statistical inference is the derivation of confidence intervals (CI) for a parameter of interest $\mu^\star \in \RR$ (e.g. mean, variance, treatment effect, etc.). Given data $X_1, \dots, X_n$, the interval $\dot C_n \equiv C(X_1, \dots, X_n)$ is said to be a $(1-\alpha)$-\ci{} for $\mu^\star$ if
\begin{equation}
\small
  \label{eq:CI}
  \PP(\mu^\star \notin \dot C_n) \leq \alpha,
\end{equation}
where $\alpha \in (0, 1)$ is a prespecified error tolerance.
Notice that~\eqref{eq:CI} is a ``pointwise'' or ``fixed-time'' statement, meaning that it only holds for a single fixed sample size $n$.

The ``time-uniform'' analogue of \ci{}s are so-called \emph{confidence sequences} (CS) --- sequences of confidence intervals that are uniformly valid over a (potentially infinite) time horizon~\citep{darling1967confidence,robbins1970statistical}. We say that the sequence $\infseq{\bar C}{t}{1}$ is a $(1-\alpha)$-\cs{}\footnote{As a mnemonic, we will use overhead bars $\bar C_t$ and dots $\dot C_n$ for time-uniform \cs{}s and fixed-time \ci{}s, respectively.} for $\mu^\star$ if
\begin{equation}
\small
  \label{eq:CS}
  \PP(\exists t \geq 1 : \mu^\star \notin \bar C_t ) \leq \alpha.
\end{equation}
The guarantee \eqref{eq:CS} has important implications for data analysis, giving practitioners the ability to (a) update inferences as new data become available, (b) continuously monitor studies without any statistical penalties for ``peeking'', and (c) make decisions based on valid inferences at arbitrary stopping times: for any stopping time $\tau$, $\PP(\mu^\star \notin \bar C_\tau ) \leq \alpha$.

\subsection{Contributions and Outline}\label{section:contributions}
Our primary contributions are threefold: (a) privacy mechanisms, (b) \ci{}s, and (c) time-uniform \cs{}s.
\begin{enumerate}[(a), itemsep=0em]
    \item We prove local DP guarantees of ``Nonparametric randomized response'' (\NPRR{}) --- a sequentially interactive, nonparametric generalization of Warner's randomized response \citep{warner1965randomized} for bounded data (Section~\ref{section:nprr}).
    \item We derive several \ci{}s for the mean of bounded random variables that are privatized by \NPRR{} (Section~\ref{section:nprr-ci}). We believe Section~\ref{section:nprr-ci} introduces the first private nonparametric and nonasymptotic \ci{}s for means of bounded random variables.
    \item We derive time-uniform \cs{}s for the mean of bounded random variables that are privatized by \NPRR{}, enabling private nonparametric sequential inference (Section~\ref{section:nprr-cs}). We also introduce a \cs{} that is able to capture means that change over time under no stationarity conditions on the time-varying means (Section~\ref{section:cs-mean-so-far}). We believe Sections~\ref{section:nprr-cs} and \ref{section:cs-mean-so-far} are the first private nonparametric \cs{}s in the DP literature.
\end{enumerate}
Furthermore, we show how all of the aforementioned techniques can be used to conduct private online A/B tests (Section~\ref{section:a/b-testing}). Finally,~\cref{section:summary} summarizes our findings and discusses some additional results whose details can be found in the appendix. A Python package implementing our methods as well as code to reproduce the figures can be found on GitHub at \href{https://github.com/WannabeSmith/nprr}{github.com/WannabeSmith/nprr}.


\subsection{Related Work}
The literature on differentially private statistical inference is rich, including nonparametric estimation rates \citep{wasserman2010statistical,duchi2013local-FOCS, duchi2013local-NeurIPS,duchi2018minimax,kamath2020private,butucea2020local,acharya2021differentially}, parametric hypothesis testing and confidence intervals \citep{vu2009differential,wang2015revisiting,gaboardi2016differentially,awan2018differentially,karwa2018finite,canonne2019structure,joseph2019locally,ferrando2020general,covington2021unbiased}, median estimation \citep{drechsler2021non}, independence testing \citep{couch2019differentially}, online convex optimization \citep{jun2019parameter}, and parametric sequential hypothesis testing \citep{wang2020differential}. A more detailed summary  can be found in Section~\ref{section:related-detailed}.

The aforementioned works do not study the problem of private nonparametric confidence sets for population means. Prior work does exist on confidence intervals for the \emph{sample mean of the data} \citep{ding2017collecting,wang2019collecting}. The most closely related work is that of \citet[Section 2.1]{ding2017collecting} who introduce the ``1BitMean'' mechanism which can be viewed as a special case of \NPRR{} (Algorithm~\ref{algorithm:NPRR}). They derive a private Hoeffding-type confidence interval for the \emph{sample} mean of the data, but it is important to distinguish this from the more classical statistical task of \emph{population} mean estimation. For example, if $X_1, \dots, X_n$ are random variables drawn from a distribution with mean $\mu^\star$, then the \emph{population mean} is $\mu^\star$, while the \emph{sample mean} is $\widehat \mu_n := \frac{1}{n}\sum_{i=1}^n X_i$. A private \ci{} for $\mu^\star$ incorporates randomness from both the mechanism \emph{and} the data, while a \ci{} for $\widehat \mu_n$ incorporates randomness from the mechanism \emph{only}. Neither is a special case of the other, and some of our techniques allow for the (sequential) estimation of sample means (see \cref{section:confidence-sets-sample-means} for details and explicit bounds) but this paper is primarily focused on the problem of private \emph{population} mean estimation.

\section{Extending Warner's Randomized Response}
\label{section:nprr}

Before discussing a nonparametric extension of randomized response, let us briefly review Warner's classical randomized response mechanism as well as the Laplace mechanism, discuss their shortcomings, and present a different mechanism that remedies them.

\noindent \textbf{Warner's randomized response.}
When the raw data $\infseq Xt1$ are binary, one of the oldest and simplest privacy mechanisms is Warner's \textit{randomized response} (RR) \citep{warner1965randomized}. Warner's RR was introduced decades before the very definition of DP, but was later shown to satisfy LDP by~\citet{dwork2014algorithmic}.
RR was introduced as a method to provide plausible deniability to subjects when answering sensitive survey questions \citep{warner1965randomized}, and proceeds as follows: when presented with a sensitive Yes/No question (e.g. ``have you ever used illicit drugs?''), the subject flips a biased coin with $\PP(\text{Heads}) = r \in (0, 1]$. If the coin comes up heads, then the subject answers truthfully; if tails, the subject answers ``Yes'' or ``No'' (encoded as 1 and 0, respectively) with equal probability 1/2. It is easy to see that this mechanism satisfies $\varepsilon$-LDP with $\varepsilon = \log ( 1 + \frac{2r}{1-r} )$ by bounding the likelihood ratio of the privatized response distributions: for any true response $x \in \{0, 1\}$, let $q(z \mid X= x) = r\1(z=x) + (1-r)/2 $ denote the conditional probability mass function of its privatized view. Then for any $x, \widetilde x \in \{0, 1\}$,
\begin{equation}
\small
    \sup_{z \in \{0, 1\}}\frac{q(z \mid X=x)}{q(z \mid X = \widetilde x)}
  \leq 1 + \frac{2r}{1-r}, 
\end{equation}
and hence RR satisfies $\eps$-LDP with $\varepsilon = \log ( 1 + \frac{2r}{1-r}  )$. In \cref{section:confidence-sets-rr}, we show how one can derive a \ci{} for the mean of Bernoulli random variables when they are privatized via RR, but as we will see in \cref{section:nprr-ci}, this will be an immediate corollary of a more general result for bounded random variables (\cref{theorem:hoeffding-nprr-ci}).

One downside of RR, however, is that it takes binary data as input. On the other hand, the famous Laplace mechanism satisfies $\eps$-LDP for \emph{bounded} data, including binary ones.



\noindent \textbf{The Laplace mechanism.}
The Laplace mechanism appeared in the very same paper that introduced DP \citep{dwork2006calibrating}. \cref{algorithm:seqIntLaplace} recalls the (sequentially interactive) Laplace mechanism \citep{duchi2013local-FOCS}.
\begin{algorithm}[!htbp]
  \caption{Sequentially interactive Laplace mechanism}
  \label{algorithm:seqIntLaplace}
  \begin{algorithmic}
    \FOR{t = 1,2,...}
    \STATE Choose $\eps_t$ based on $Z_1^{t-1}$.
    \STATE Generate $\lapNoise_t \sim \mathrm{Laplace}(1/\eps_t)$\;
    \STATE $Z_t \gets X_t + \lapNoise_t$\;
    \ENDFOR
  \end{algorithmic}
\end{algorithm}
It is well-known that $Z_t$ is (conditionally) $\eps_t$-LDP (given $Z_1^{t-1}$) for each $t$ \citep{dwork2006calibrating}. 
\cref{section:laplace} derives novel \ci{}s and \cs{}s for population means under the Laplace mechanism, but we omit them here for brevity as a different mechanism (to be described shortly) will yield better bounds.



\noindent \textbf{Nonparametric randomized response (\NPRR{}).} The mechanism we use, which we call ``Nonparametric randomized response'' (\NPRR{}) serves as a sequentially interactive generalization of RR for arbitrary bounded data by simply combining stochastic rounding \citep{barnes1951electronic,forsythe1959reprint,hull1966tests} with $k$-RR --- a categorical but non-interactive generalization of Warner's RR introduced by \citet{kairouz2014extremal,kairouz2016discrete}, and also considered by \citet{li2020estimating} under the name ``Generalized Randomized Response''. Note that \citet{kairouz2014extremal,kairouz2016discrete} use $k$ to refer to the number of unique values that the input and output data can take on, which is $k = G+1$ in the case of Algorithm~\ref{algorithm:NPRR}. \NPRR{} is explicitly described in Algorithm~\ref{algorithm:NPRR}, and we summarize its LDP guarantees in Theorem~\ref{theorem:NPRR-DP}. 


\begin{algorithm}[h!]
\caption{Nonparametric randomized response (\NPRR{})}
\label{algorithm:NPRR}
  \begin{algorithmic}
    \FOR{$t = 1,2, \dots$}
    \STATE{\textit{// Step 1: Discretize $X_t$ into $Y_t$ via stochastic rounding.}}
    \STATE Choose integer $G_t \geq 1$ based on $Z_1^{t-1}$\;
    \STATE $X_t^\mathrm{ceil} \gets \lceil G_tX_t \rceil / G_t$, $X_t^\mathrm{floor} \gets \lfloor G_tX_t \rfloor / G_t$ \;
    \IF{$X_t^\mathrm{ceil} == X_t^\mathrm{floor}$}
        \STATE $Y_t \gets X_t$
    \ELSE{}
    \STATE Generate $Y_t \sim \begin{cases}
      X_t^\mathrm{ceil} & \text{w.p. } G_t\cdot (X_t - X_t^\mathrm{floor}) \\
      X_t^\mathrm{floor} & \text{w.p. } G_t\cdot (X_t^\mathrm{ceil} - X_t)
    \end{cases}$ \;
    \ENDIF
    \STATE \textit{// Step 2: Privatize $Y_t$ into $Z_t$ via $k$-RR.}
    \STATE{Choose $r_t \in (0, 1]$ based on $Z_1^{t-1}$}\;
    \STATE Generate $\unifNoise_t \sim \mathrm{Unif}\left \{ 0, \frac{1}{G_t}, \frac{2}{G_t}, \dots, \frac{G_t}{G_t} \right\}$\;
    \STATE Generate $Z_t \sim \begin{cases} Y_t & \text{w.p. } r_t \\ \unifNoise_t & \text{w.p. } 1-r_t \end{cases}$
    \ENDFOR
  \end{algorithmic}
\end{algorithm}
Up to rescaling of the outputs, \NPRR{} can be viewed as a sequentially interactive analogue of the local randomizer given in \citet[Algorithm 2]{balle2019privacy} in the context of the shuffle model. We nevertheless refer to this mechanism as ``\NPRR{}'' to emphasize its connection to Warner's RR. Indeed, if $\infseq{X}{t}{1}$ are $\{0, 1\}$-valued, and if we set $G_1 = G_2 = \dots = 1$ and $r_1 = r_2 = \dots = r \in (0, 1]$, then no stochastic rounding occurs and \NPRR{} recovers RR exactly, making \NPRR{} a sequentially interactive and nonparametric generalization for bounded data. Finally, if we let \NPRR{} be non-interactive and set $G_1 = \dots = G_n = 1$, then \NPRR{} recovers the ``1BitMean'' mechanism of \citet{ding2017collecting}. As such, the form of \NPRR{} given in \cref{algorithm:NPRR} makes transparent generalizations of and connections between \citet{warner1965randomized}, \citet{kairouz2016discrete}, \citet{li2020estimating}, \citet{ding2017collecting}, and \citet{balle2019privacy}.\footnote{Notice that \citet{ding2017collecting}'s $\alpha$-point rounding mechanism is different from \NPRR{} as \NPRR{} shifts the mean of the inputs but alpha-point rounding leaves the mean unchanged.}
Let us now formalize \NPRR{}'s LDP guarantees.
\begin{restatable}[\NPRR{} satisfies LDP]{theorem}{NPRRLDP}\label{theorem:NPRR-DP}
Suppose $\infseq{Z}{t}{1}$ are generated according to \NPRR{}. Then for each $t \in \{1,2,\dots\}$,  $Z_t$ is a conditionally $\eps_t$-LDP view of $X_t$ with
\begin{equation}
\small
  \eps_t := \log \left ( 1 + \frac{(G_t+1)r_t}{1-r_t} \right ).
\end{equation}
\end{restatable}

The proof in Section~\ref{proof:NPRR-DP} proceeds by bounding the conditional likelihood ratio for any two data points $x, \widetilde x \in [0, 1]$ similar to~\eqref{eq:seqInteractiveLikelihood}.
In all of the results that follow in the following sections, we will write expressions in terms of $\seq rt1n$, but these can always be chosen given desired $\seq \eps t1n$ levels via the relationship
\begin{equation}\label{eq:r-eps-relationship}\small
  r_t = \frac{\exp\{\eps_t\} - 1}{\exp \{\eps_t\} + G_t}.
\end{equation}
In the familiar special case of $r_t = r \in (0, 1]$ and $G_t = G \in \{1,2,\dots\}$ for each $t$, we have that $\infseq Zt1$ satisfy $\eps$-LDP with $\eps := \log ( 1 + (G+1)r/(1-r) )$. Notice that when $G_t = 1$ for each $t$, we have that \NPRR{} satisfies $\eps$-LDP with the same value of $\eps$ as Warner's RR\@. Consequently, there is no privacy lost from instantiating the more general \NPRR{} to the binary case.

\begin{remark}[Who chooses $\eps_t$, $r_t$, or $G_t$, and how?]\label{remark:who-chooses-r-eps-G}
  Due to the sequential interactivity of \NPRR{}, individuals can specify their own levels of privacy, or the parameters $(r_t, G_t)_{t=1}^\infty$ can be adjusted over time (e.g.~if the data collector chooses to decrease $\eps_t$ for regulatory reasons, or increase $\eps_t$ to obtain sharper inference). Formally, $(r_t, G_t)$ can be chosen in any way as long as they are \emph{predictable}, meaning that they can depend on $Z_1^{t-1}$.
  Nevertheless, sequential interactivity is completely optional, and the data collector is free to set $(r_t, G_t) = (r, G)$ for every $t$ to recover the familiar notion of $\eps$-LDP\@.
\end{remark}

\noindent \textbf{Why use \NPRR{} instead of Laplace?}
While RR is limited to privatizing binary data, the Laplace mechanism can handle bounded data, so why use \NPRR{} as an alternative to the two? The reason stems from our original motivation: to derive locally private nonparametric, nonasymptotic confidence sets for means of bounded random variables. To achieve this, we will ultimately use modern concentration techniques from the literature on (non-private) confidence sets, many of which exploit boundedness in clever ways to yield clean, closed-form expressions and/or empirically tight confidence intervals. Since the Laplace mechanism does not preserve the boundedness of its input, it is not clear how those techniques can be used for Laplace-privatized data (though we do derive novel Laplace-based solutions using a different approach in \cref{section:laplace}, but they are ultimately outperformed by those that we derive based on \NPRR{}). \NPRR{} on the other hand, preserves the input's boundedness, making it possible to apply analogues of these modern concentration techniques for \NPRR{}-privatized data. The efficiency gains that result from this approach are illustrated in Figures~\ref{fig:ci} and~\ref{fig:cs}.

In addition to being useful for deriving simple and efficient confidence sets, \NPRR{} has some other orthogonal advantages over the Laplace mechanism. First, \NPRR{} has reduced storage requirements: Once a $[0, 1]$-bounded random variable has been privatized via Laplace, the output is a floating-point number, requiring 64 bits to store as a double-precision float. In contrast, \NPRR{} outputs one of $(G+1)$ different values, hence requiring only $ \left \lceil \log_2(G+1) \right \rceil$ bits to store.
Moreover, storing the \NPRR{}-privatized view of $x$ will never require more memory than storing $x$ itself (unless $G$ is set to nonsensical values larger than $2^{64}$), while  Laplace-privatized views will always require at least enough memory to represent floating point numbers.

Second, \NPRR{} is automatically resistant to the floating-point attacks that the Laplace mechanism suffers from. \citet{mironov2012significance} showed that storing Laplace output as a floating-point number can leak information about the input $x$, thereby compromising its LDP guarantees. While \citet{mironov2012significance} discusses remedies to this issue, practitioners may still naively apply the Laplace mechanism using common software packages and remain vulnerable to these so-called ``floating-point attacks''. In contrast, the discrete representation of \NPRR{}'s output is not vulnerable to such attacks, without the need for remedies at all. Note that while \NPRR{} may have to deal with floating point numbers as input, they are transformed into discrete random variables \emph{before} any $\eps$-LDP guarantees are added. The privatization step (transforming $Y_t$ into $Z_t$ in \cref{algorithm:NPRR}) takes one of $G_t+1$ values as input and produces one of $G_t+1$ values as output, thereby sidestepping any need to handle floating point numbers.

The remainder of this paper will focus solely on constructing efficient locally private confidence sets, but the above benefits can be seen as ``free'' byproducts of \NPRR{}'s design.

\section{Private CIs for Bounded Data}
\label{section:nprr-ci}

Making matters formal, let $\Pcal_\mu$ be the set of distributions on $[0, 1]$ with population mean $\mu \in [0, 1]$. $\Pcal_\mu$ is a convex set of distributions with no common dominating measure, since it consists of discrete and continuous distributions, as well as their mixtures. We will consider sequences of random variables $(X_i)_{i=1}^n$ drawn from the product distribution $\prodseq{P}{i}{n}$ where $n \in \{1, 2, \dots, \infty\}$ and each $P_i \in \Pcal_\mu$. For succinctness, define the following set of distributions,
\begin{equation}
\small
     \Pcal_\mu^n := \left \{ \prod_{i=1}^n P_i \text{ such that each } P_i \in \Pcal_\mu \right \}, 
\end{equation}
for $n \in \{1, 2, \dots, \infty \}$. In words, $\Pcal_\mu^n$ contains distributions for which the random variables are independent and $[0, 1]$-bounded with mean $\mu$ but need not be identically distributed. 
We use the notation $\seq Xt1n \sim P$ for some $P\in \Pcal_{\mu^\star}^n$ to indicate that $\seq Xt1n$ are independent with mean $\mu^\star$.
The goal is now to derive sharp \ci{}s and time-uniform \cs{}s for $\mu^\star$ given \NPRR{}-privatized views of $\seq{X}{t}{1}{n}$. 


Let us write $\Qcal_{\mu^\star}^n$ to denote the set of joint distributions on \NPRR{}'s output, where we have left the dependence on each $G_t$ and $r_t$ implicit. In other words, given $\seq Xt1n \sim P$ for some $P \in \Pcal_{\mu^\star}^n$, their \NPRR{}-induced privatized views $\seq Zt1n$ have a joint distribution from $Q$ for some $Q \in \Qcal_{\mu^\star}^n$.
\begin{figure}[h!]
  \centering
  \includegraphics[width=\nprrwidth]{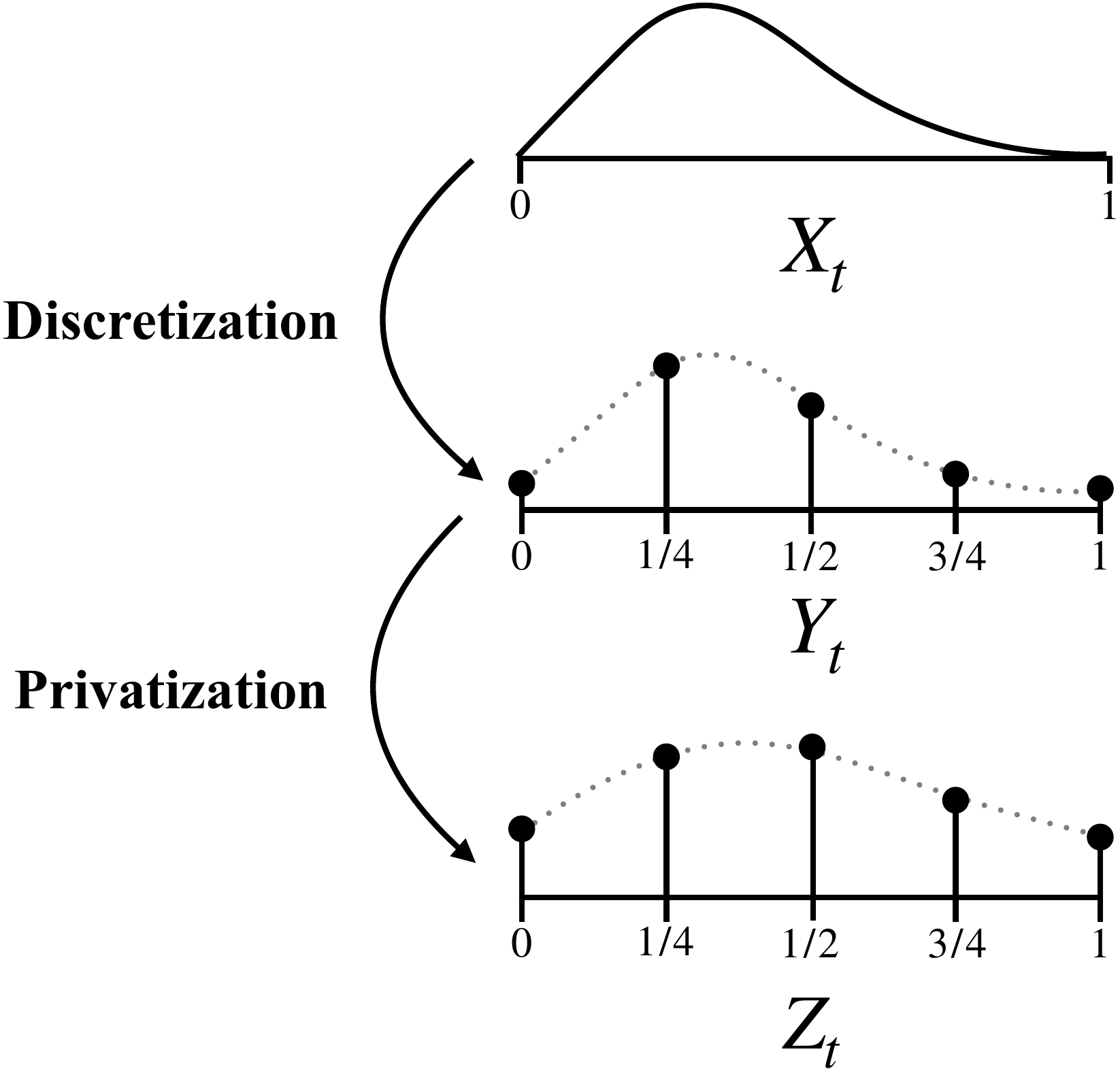}
  \caption{An illustration of how a distribution $Q \in \Qcal_{\mu^\star}^n$ can arise from applying \NPRR{} with $G_t = 4$ to draws from the input distribution $P \in \Pcal_{\mu^\star}^n$. Raw data $X_t$ are discretized into $Y_t$ so that $Y_t$ has finite support but so that $\mu^\star = \EE (X_t) = \EE(Y_t)$. The discrete $Y_t$ are then privatized into $Z_t$ with conditional mean $\EE (Z_t \mid Z_1^{t-1}) = \zeta_t(\mu^\star) = r_t\mu^\star + (1-r_t)/2$ by being mixed with independent uniform noise $\unifNoise_t \sim {\mathrm{Unif}\{0, 1/4, 1/2, 3/4, 1\}}$.}
  \label{fig:sirrg}
\end{figure}

\subsection{What is a Locally Private Confidence Set?}\label{section:lpci}
Let first define what we mean by locally private confidence intervals (\lpci{}) and sequences (\lpcs{}), and subsequently derive them for means of bounded random variables.

\begin{definition}[Locally private confidence sets]\label{definition:lpci}
  Let $\eps \equiv \seq{\eps}{t}{1}{n} \equiv (\eps_t)_t$. We say that $L_n$ is a lower $(1-\alpha, \eps)$-\lpci{} for a parameter $\theta^\star$, and with respect to the raw data $\seq Xt1n$ if $L_n$ is a lower $(1-\alpha)$-\ci{} for $\theta^\star$, meaning
  \begin{equation}\label{eq:lpci-coverage}\small
    \PP(\theta^\star \geq L_n) \geq 1-\alpha,
  \end{equation}
  and if $L_n \equiv L(Z_1, \dots, Z_n)$ is only a function of the $\eps_t$-LDP view $Z_t$ of $X_t$ for each $t$, but not of $\seq Xt1n$ directly.

  Similarly, we say that $\infseqt{L_t}$ is a lower $(1-\alpha, \eps)$-\lpcs{} for $\theta^\star$ if \eqref{eq:lpci-coverage} is replaced with the time-uniform guarantee
  \begin{equation}
    \PP(\forall t,\ \theta^\star \geq L_t) \geq 1-\alpha.
  \end{equation}
  Upper \ci{}s and \cs{}s are defined analogously.
\end{definition}
Note that \lpci{}s and \lpcs{}s also satisfy $\eps$-LDP, since DP is closed under post-processing \citep{dwork2014algorithmic}.


\subsection{A Locally Private Hoeffding \ci{} via \NPRR{}}
First, we present a private generalization of Hoeffding's inequality under \NPRR{}.

\begin{restatable}[\texttt{NPRR-H}]{theorem}{HoeffdingNPRRCI}\label{theorem:hoeffding-nprr-ci}
  Suppose $\seq Xt1n \sim P$ for some $P \in \PcalNPRRstarn$, and let $\seq Zt1n \sim Q \in \QcalNPRRstarn$ be their privatized views via \NPRR{}.
  Define the \NPRR{}-adjusted sample mean
  \begin{equation}
  \small
      \widehat \mu_n := \frac{\frac{1}{n}\sum_{i=1}^n \left ( Z_i - (1-r_i)/2 \right )}{\frac{1}{n}\sum_{i=1}^n r_i}.
  \end{equation}
  Then,
    \begin{equation}
    \small
    \label{eq:hoeffding-ci-simple-case}
        \dot L_n^\Hoeff := \widehat \mu_n - \sqrt{\frac{\log(1/\alpha) }{2n(\frac{1}{n}\sum_{i=1}^n r_i)^2}}
    \end{equation}
    is a lower $(1-\alpha,(\eps_t)_t)$-\lpci{} for $\mu^\star$.
\end{restatable}

The proof in Section~\ref{proof:hoeffding-nprr-ci} uses a locally private supermartingale variant of the Cram\'er-Chernoff bound. We recommend setting $r_t$ for the desired $\eps_t$-LDP level via the relationship in~\eqref{eq:r-eps-relationship} and $G_t := 1$ for all $t$ (the reason behind which we will discuss in \cref{remark:hoeffding-lack-of-variance-adaptivity}). Notice that in the non-private setting where we set $r_i = 1$ for all $i$, then $\dot L_n^\Hoeff$ recovers the classical Hoeffding inequality exactly \citep{hoeffding1963probability}. Moreover, notice that if $\seq Xt1n$ took values in $[a,b]$ instead of $[0, 1]$, then \eqref{eq:hoeffding-ci-simple-case} would simply scale with $(b-a)$ in the same manner as \citet{hoeffding1963probability}. Recall as discussed in \cref{remark:who-chooses-r-eps-G} that $\seq rt1n$ could be chosen either by the data collector or by the subject whose data are being collected, but that sequential interactivity is optional.

\begin{figure}[!ht]
  \centering
  \includegraphics[width=\ciplotwidth]{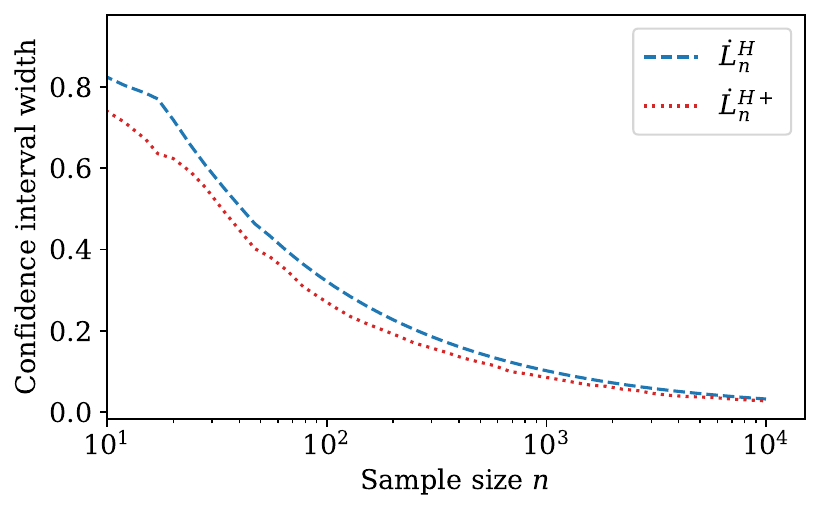}
  \caption{Two $(90\%,2)$-\lpci{}s: $\dot L_n^H$ given in \eqref{eq:hoeffding-ci-simple-case} and $\dot L_n^{H+}$ given in \eqref{eq:nprr-h-ci-tighter} --- i.e.~these are $(1-\alpha,\eps)$-\lpci{}s with $\alpha=0.1$ and $\eps=2$. Notice that the latter can be tighter than the former. Indeed this is because $L_n^{H+}$ is never looser than $L_n^H$ (by definition) but strictly tighter with positive probability.}
  \label{fig:LnH-versus-LnHplus}
\end{figure}
In fact, we can strictly improve on \eqref{eq:hoeffding-ci-simple-case} by exploiting the martingale dependence of this problem. Indeed, under the same assumptions as Theorem~\ref{theorem:hoeffding-nprr-ci}, we have that
\begin{equation}
\label{eq:nprr-h-ci-tighter} 
\small
  \dot L_{n}^{\Hoeff+} := \max_{1 \leq t \leq n} \left \{ \widehat \mu_t  - \frac{\log(1/\alpha) + t \lambda_n^2/8}{\lambda_n \sum_{i=1}^t r_i} \right \} 
\end{equation}
is also a lower $(1-\alpha, \seq{\eps}{t}{1}{n})$-\lpci{} for $\mu^\star$, where $\lambda_n := \sqrt{8 \log(1/\alpha)/n}$. Notice that $\dot L_n^{\Hoeff+}$ is at least as tight as $\dot L_n^\Hoeff$ since the $n^\mathrm{th}$ term in the above $\max_{1\leq t \leq n}$ recovers $\dot L_n^\Hoeff$ exactly. Moreover, $\dot L_n^{\Hoeff+}$ is \emph{strictly} tighter than $\dot L_n^\Hoeff$ with positive probability, and hence strictly tighter in expectation: $\EE( \dot L_n^{\Hoeff+}) > \EE( \dot L_n^\Hoeff)$.
\begin{figure}[!ht]
     \centering
     \includegraphics[width=\ciplotwidth]{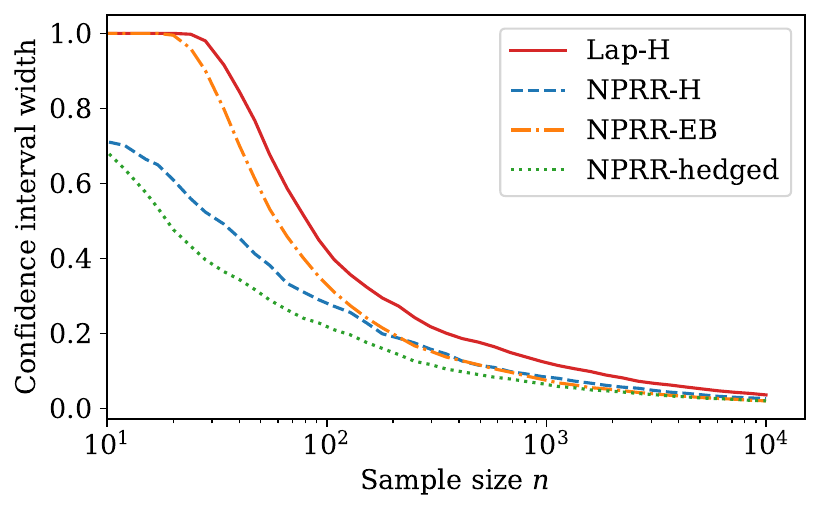}
     \caption{Widths of $(90\%,2)$-\lpci{}s for the mean of a Beta(50, 50) distribution. Hoeffding-based methods (Lap-H and NPRR-H found in \cref{corollary:laplace-hoeffding-ci} and \cref{theorem:hoeffding-nprr-ci}) do slightly worse than the variance-adaptive ones (NPRR-EB and NPRR-hedged in \cref{proposition:ldp-eb-ci} and \cref{theorem:ldp-hedged-ci}), but in all cases, \ci{}s that rely on \NPRR{} seem to outperform Lap-H in both small and large $n$ regimes.}
     \label{fig:ci}
   \end{figure}
   
\begin{remark}[Minimax rate optimality of \eqref{eq:hoeffding-ci-simple-case}]
  In the case of $\eps_1 = \dots = \eps_n = \eps \in (0, 1]$, \citet[Proposition 1]{duchi2013local-FOCS} give minimax estimation rates for the problem of nonparametric mean estimation. Their lower bounds say that for any $\eps$-LDP mechanism and estimator $\widehat \mu_n$ for $\mu^\star$, the root mean squared error $\sqrt{\EE ( \widehat \mu_n - \mu^\star )^2}$ cannot scale faster than $O(1/\sqrt{n\eps^2})$. Since \NPRR{} is $\eps$-LDP with $\eps = \log(1 + 2r/(1-r))$, we have that $r \asymp \eps$ up to constants on $\eps \in (0, 1]$. It follows that $\dot L_n^\Hoeff \asymp 1/\sqrt{n \eps^2}$, matching the minimax estimation rate. Of course, the midpoint of a \ci{} for $\mu^\star$ can always be used as an estimator for $\mu^\star$, and hence we cannot expect the width of the \ci{} to shrink faster than the minimax estimation rate. While explicit minimax lower bounds do not exist for the setting where $\eps_i \neq \eps_j$ for some $i, j$, notice that instead of scaling with $r^{-1}$ (which we would have if $r_i = r_j$ for $i \neq j$), $\dot L_n^\Hoeff$ scales with $(\frac{1}{n}\sum_{i=1}^n r_i)^{-1}$, and hence our bounds seem to be of the right order when $\eps$ is permitted to change.
\end{remark}

\begin{remark}[The relationship between $\eps$ and \eqref{eq:hoeffding-ci-simple-case} for practical levels of privacy]
    As mentioned in the introduction and in Figure~\ref{fig:eps}, Apple uses values of $\eps \in \{2,4,8\}$ for various $\eps$-LDP data collection tasks on iPhones~\citep{appleEps}. Note that for $G = 1$, having $\eps$ take values of 2, 4, and 8 corresponds to $r$ being roughly 0.762, 0.964, and 0.999, respectively, via the relationship $r = (\exp(\eps) - 1)/(\exp(\eps) + 1)$. As such, \eqref{eq:hoeffding-ci-simple-case} simply inflates the width of the non-private Hoeffding \ci{} by $0.762^{-1}$, $0.964^{-1}$, and $0.999^{-1}$, respectively. Hence larger  $\eps$ (e.g. $\eps \geq 4$) leads to \ci{}s that are nearly indistinguishable from the non-private case (Figure~\ref{fig:eps}).
\end{remark}


\begin{remark}\label{remark:hoeffding-lack-of-variance-adaptivity}
Since Hoeffding-type bounds are not variance-adaptive (meaning they use a worst-case upper-bound on the variance of bounded random variables as in \citet{hoeffding1963probability}), they do not benefit from the additional granularity when setting $G_t \geq 2$ (see Section~\ref{section:hoeffding-G=1} for a detailed mathematical explanation). As such, we set $G_t = 1$ for each $t$ when running \texttt{NPRR-H}. Nevertheless, other \ci{}s are capable of adapting to the variance with $G_t \geq 2$, and these are discussed in \cref{section:variance-adaptive}, with some suggestions for how to choose $G_t\geq 2$ in \cref{section:choosing-rG}.
Nevertheless, the empirical performance of our variance-adaptive \ci{}s is illustrated in \cref{fig:ci}.
\end{remark}

\subsection{Time-uniform Confidence Sequences \texorpdfstring{for $\mu^\star$}{}}
\label{section:nprr-cs}
Previously, we focused on constructing a (lower) \ci{} $L_n$ for $\mu^\star$, meaning that $L_n$ satisfies the high-probability guarantee $\PP(\mu^\star \geq L_n) \geq 1-\alpha$ for the prespecified sample size $n$. We will now derive \cs{}s --- i.e.~entire \emph{sequences} of \ci{}s $\infseqt{L_t}$ --- which have the stronger \emph{time-uniform} coverage guarantee $\PP(\forall t,\ \mu^\star \geq L_t) \geq 1-\alpha$, enabling anytime-valid inference in sequential regimes. See~\cref{section:background-cs-martingale} for a review of the mathematical and practical differences between \ci{}s and \cs{}s. In summary, if $\infseqt{L_t}$ is a lower $(1-\alpha)$-\cs{}, then $L_\tau$ forms a valid $(1-\alpha)$-\ci{} at arbitrary stopping times $\tau$ (including random and data-dependent times) and hence a practitioner can continuously update inferences as new data are collected, without any penalties for ``peeking'' at the data early. Let us now present a Hoeffding-type \cs{} for $\mu^\star$, serving as a time-uniform analogue of~\cref{theorem:hoeffding-nprr-ci}.

\begin{restatable}[\texttt{NPRR-H-CS}]{theorem}{LDPHoeffdingCS}\label{theorem:ldp-hoeffding-cs}
  Let $\infseq Zt1 \sim Q$ for some $Q \in \Qcal_{\mu^\star}^\infty$. Define the modified mean estimator under \NPRR{}:
    \begin{equation}
    \small
    \widehat \mu_{t}(\lambda_1^t) := \frac{\sum_{i=1}^t \lambda_i \cdot (Z_i - (1-r_i)/2)}{\sum_{i=1}^t r_i \lambda_i},
    \end{equation}
   and let $\infseq \lambda t1$ be a real-valued sequence of tuning parameters (discussed in \eqref{eq:laplace-lambda-cs}). Then,
   \begin{equation}
   \small
       \bar L_t^\Hoeff := \widehat \mu_t(\lambda_1^t) - \frac{\log(1/\alpha) + \sum_{i=1}^t \lambda_i^2 / 8}{\sum_{i=1}^t r_i\lambda_i}
   \end{equation}
  forms a lower $(1-\alpha, (\eps_t)_t)$-\lpcs{} for $\mu^\star$.
\end{restatable}

\begin{figure}[!htbp]
    \centering
    \includegraphics[width=\ciplotwidth]{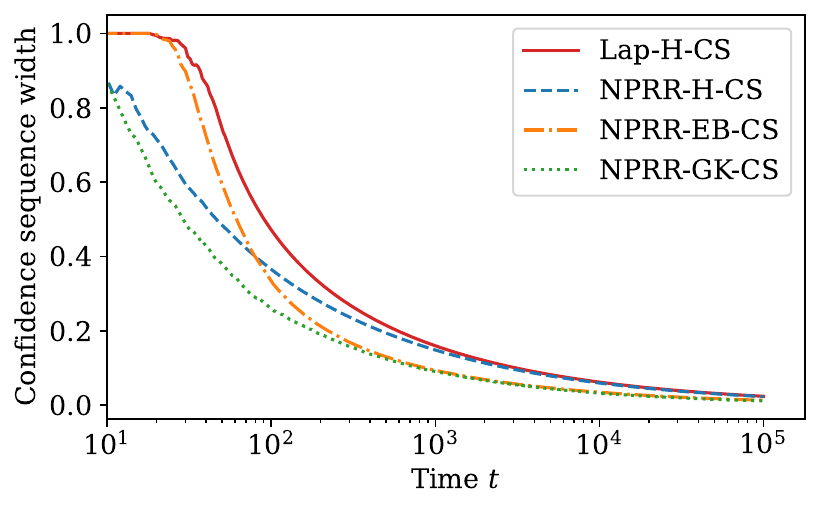}
    \caption{Widths of (90\%, 2)-\lpcs{}s for the mean of a Beta(50, 50) distribution. Like Figure~\ref{fig:ci}, Hoeffding-based methods (Lap-H-CS and NPRR-H-CS found in \cref{proposition:laplace-hoeffding-cs} and \cref{theorem:ldp-hoeffding-cs}) do worse than the variance-adaptive ones (NPRR-EB-CS and NPRR-GK-CS in \cref{proposition:ldp-eb-cs} and \cref{theorem:ldp-dkelly}) for large $t$, though NPRR-H-CS does outperform NPRR-EB-CS for small $t$. Nevertheless, in all cases, we find that \NPRR{}-based \cs{}s outperform Lap-H-CS in both small and large $t$ regimes.}
    \label{fig:cs}
\end{figure}
The proof can be found in Section~\ref{proof:ldp-hoeffding-cs}.
Unlike Theorem~\ref{theorem:hoeffding-nprr-ci}, we suggest setting
\begin{equation}\label{eq:nprr-h-cs-lambdas}
\small
     \lambda_t:= \sqrt{\frac{8 \log(1/\alpha)}{t \log (t+1)}} \land 1, 
\end{equation}
to ensure that $\bar L_t^\Hoeff \asymp O(\sqrt{\log t / t})$ up to $\log \log t$ factors.
\citet[Section 3.3]{waudby2020estimating} give a derivation and discussion of $\lambda_t$ and the $O(\sqrt{\log t / t})$ rate.
Similar to \cref{theorem:hoeffding-nprr-ci}, we recommend setting $r_t$ for the desired $\eps_t$-LDP level via~\eqref{eq:r-eps-relationship} and $G_t := 1$ for all $t$.

The similarity between Theorem~\ref{theorem:ldp-hoeffding-cs} and Theorem~\ref{theorem:hoeffding-nprr-ci} is no coincidence: indeed, Theorem~\ref{theorem:hoeffding-nprr-ci} is a corollary of Theorem~\ref{theorem:ldp-hoeffding-cs} where we instantiated a \cs{} at a fixed sample size $n$ and set $\lambda_1 = \cdots = \lambda_n = \sqrt{8\log (1/\alpha)/n}$. In fact, every Cram\'er-Chernoff bound (even in the non-private regime) has an underlying supermartingale and \cs{} that are rarely exploited \citep{howard2020time}, but setting $\lambda$'s as in \cref{theorem:hoeffding-nprr-ci} tightens these \cs{}s for the fixed time $n$ --- yielding $O(1/\sqrt{n})$ rates but only for a fixed $n$ --- while tuning $\lambda_t$ as in \eqref{eq:nprr-h-cs-lambdas} allows them to spread their efficiency over all $t$ --- yielding $O(\sqrt{\log t / t})$ rates but for all $t$ simultaneously.
Notice that both the time-uniform and fixed-time bounds in Theorems~\ref{theorem:hoeffding-nprr-ci} and~\ref{theorem:ldp-hoeffding-cs} cover an unchanging real-valued mean $\mu^\star \in \RR$ --- in the following section, we will relax this assumption and allow for the mean of each $X_i$ to change over time in an arbitrary matter, but still derive~\cs{}s for sensible parameters.

\subsection{Confidence Sequences for Time-varying Means}
\label{section:cs-mean-so-far}
All of the bounds derived thus far have been concerned with estimating some common $\mu^\star$ under the nonparametric assumption $\seq Xt1\infty \sim P$ for some $P \in \PcalNPRRstarinf$ and hence $\seq Zt1\infty \sim Q$ for some $Q \in \QcalNPRRstarinf$. Let us now consider the more general (and challenging) task of constructing \cs{}s for \emph{the average mean so far} $\widetilde \mu_t^\star := \frac{1}{t}\sum_{i=1}^t\mu_i^\star$ under the assumption that each $X_t$ has a different mean $\mu_t^\star$. In what follows, we require that \NPRR{} is non-interactive, i.e. $r_t = r \in (0, 1]$ and $G_t = G \in \{1, 2, \dots\}$ for each $t$.

\begin{restatable}[Confidence sequences for time-varying means]{theorem}{TwoSidedCSMeanSoFar}\label{theorem:two-sided-cs-mean-so-far}
  Suppose $X_1, X_2, \dots$ are independent $[0, 1]$-bounded random variables with individual means $\EE X_t = \mu_t^\star$ for each $t$, and let $Z_1, Z_2 \dots$ be their privatized views according to \NPRR{} without sequential interactivity.
  Define
  {\small\begin{align}
           \widehat \mu_t &:= \frac{\sum_{i=1}^t (Z_i - (1-r) / 2)}{t r}, \label{eq:changing-means-sample-mean}\\
           \text{and}~~ \widetilde B_t^\pm &:= \sqrt{\frac{t\beta^2 + 1 }{2(tr\beta)^2}\log \left (\frac{\sqrt{t \beta^2 + 1}}{\alpha} \right)}, \label{eq:two-sided-changing-means-boundary}
         \end{align}}%
       for any $\beta > 0$. Then, $\widetilde C_t^\pm := (\widehat \mu_t \pm \widetilde B_t^\pm)$
       forms a two-sided $(1-\alpha,\eps)$-\lpcs{} for $\widetilde \mu_t^\star$, where $\eps = \log(1 + \frac{2r}{1-r})$.
\end{restatable}

The proof in Section~\ref{proof:two-sided-cs-mean-so-far} uses a sub-Gaussian mixture supermartingale technique similar to \citet{robbins1970statistical} and \citet{howard2020time,howard2021time}.
The parameter $\beta>0$ is a tuning parameter dictating a time for which the \cs{} boundary is optimized. Regardless of how $\beta$ is chosen, $\widetilde C_t^\pm$ has the time-uniform coverage guarantee given in \cref{theorem:two-sided-cs-mean-so-far} but finite-sample performance can be improved near a particular time $t_0$ by selecting
\begin{equation}\label{eq:rho-opt}
  \small
  \beta_{\alpha}(t_0) := \sqrt{\frac{-2\log \alpha + \log(-2\log\alpha + 1)}{t_0}},
\end{equation}
which approximately minimizes $\widetilde B_{t_0}$; see \citet[Section 3.5]{howard2021time} for details.

Notice that in the non-private case where $r=1$, we have that $\widetilde C_t^\pm$ recovers Robbins' sub-Gaussian mixture \cs{} \citep{robbins1970statistical,howard2021time}. Notice that while \cref{theorem:two-sided-cs-mean-so-far} handles a strictly more general and challenging problem than the previous sections (by tracking a time-varying mean $\infseqt{\widetilde \mu_t}$), it has the restriction that \NPRR{} must be non-interactive. There is a technical reason for this that boils down to it being difficult to combine time-varying \emph{tuning parameters} (such as those in \cref{theorem:ldp-hoeffding-cs}) with time-varying \emph{estimands} in the same \cs{}. This challenge has appeared in other (non-private) works on \cs{}s \citep{waudby2020estimating,howard2021time}. In short, this paper has methods for tracking a time-varying mean under non-interactive \NPRR{} or a fixed mean under sequentially interactive \NPRR{}, but not both simultaneously --- this would be an interesting direction to explore in future work.

\begin{figure}[!htbp]
    \centering
    \includegraphics[width=\ciplotwidth]{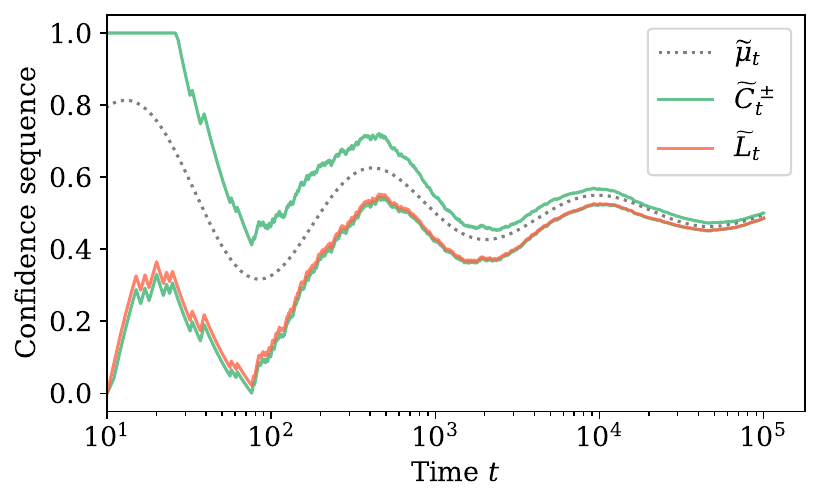}
    \caption{$(90\%, 2)$-\lpcs{}s for the average time-varying mean so far $\widetilde \mu_t^\star$ with the boundary optimized for $t_0 = 100$. In this example, we set $\mu_t^\star = \frac{1}{2}\left [1 - \sin(2\log(e + t)) / \log(e + 0.01t) \right ]$ to produce the displayed sinusoidal behavior. Notice that $\widetilde L_t$ is tighter at the expense of only being one-sided. In either case, however, the \cs{}s adapt to non-stationarity and capture $\widetilde \mu_t^\star$ uniformly over time.}
    \label{fig:wavy-cs}
    \vspace{-0.1in}
\end{figure}

A one-sided analogue of \cref{theorem:two-sided-cs-mean-so-far} is presented in \cref{section:wavy-one-sided} via slightly different techniques.

\section{Illustration: Private Online A/B Testing}
\label{section:a/b-testing}
Our methods can be used to conduct locally private \emph{online A/B tests} (sequential randomized experiments). Broadly, an A/B test is a statistically principled way of comparing two different \emph{treatments} --- e.g. administering drug A versus drug B in a clinical trial. In its simplest form, A/B testing proceeds by (i) randomly assigning subjects to receive treatment A with some probability $\pi \in (0, 1)$ and treatment B with probability $1-\pi$, (ii) collecting some outcome measurement $Y_t$ for each subject $t \in \{1,2, \dots\}$ --- e.g. severity of headache after taking drug A or B --- and (iii) measuring the difference in that outcome between the two groups. An \emph{online} A/B test is one that is conducted sequentially over time --- e.g. a sequential clinical trial where patients are recruited one after the other or in batches.

We now illustrate how to sequentially test for the mean difference in outcomes between groups A and B when only given access to locally private data. To set the stage, suppose that $(A_1, Y_1), (A_2, Y_2), \dots$ are random variables such that $A_t \sim \mathrm{Bernoulli}(\pi)$ is 1 if subject $t$ received treatment A and 0 if they received treatment B, and $Y_t$ is a $[0, 1]$-bounded outcome of interest after being assigned treatment $A_t$.

Using the techniques of Section~\ref{section:cs-mean-so-far}, we will construct $(1-\alpha)$-\cs{}s for the \emph{time-varying mean} $\widetilde \Delta_t := \frac{1}{t} \sum_{i=1}^t \Delta_i$ where $\Delta_i := {\EE(Y_i \mid A_i = 1)} - {\EE(Y_i \mid A_i = 0)}$ is the mean difference in the outcomes at time $i$. In words, $\widetilde \Delta_t$ is the mean difference in outcomes \emph{among the subjects so far}.

Unlike Section~\ref{section:cs-mean-so-far}, however, we will not directly privatize $\infseq Yt1$, but instead will apply \NPRR{} to some ``pseudo-outcomes'' $\varphi_t \equiv \varphi_t(Y_t, A_t)$ --- functions of $Y_t$ and $A_t$,
{\small\begin{equation*}
\label{eq:varphi}
    \varphi_t := \frac{f_t + \frac{1}{1-\pi}}{\frac{1}{\pi} + \frac{1}{1-\pi}}, ~~\text{where}~~f_t := \left [ \frac{Y_tA_t}{\pi} - \frac{Y_t(1-A_t)}{1-\pi} \right ].
  \end{equation*}
}%

Notice that due to the fact that $Y_t, A_t \in [0, 1]$, we have $f_t \in [-1/(1-\pi), 1/\pi]$, and hence $\varphi_t \in [0, 1]$.    
Now that we have $[0, 1]$-bounded random variables $\infseq \varphi t1$, we can obtain their \NPRR{}-induced $\eps$-LDP views $\infseq \psi t1$ by setting $G_t = 1$ and $r_t = \exp\{\eps - 1\} / \exp\{\eps + 1\}$ for each $t$. Notice that we are privatizing $\varphi_t$ which is a function of both $Y_t$ and $A_t$, so both the outcome \emph{and} the treatment are protected with $\eps$-LDP\@.

\begin{corollary}[Locally private online A/B estimation]\label{corollary:a/b-estimation}
    Following the setup above, let $\infseq \psi t1$ be the \NPRR{}-induced privatized views of $\infseq \varphi t1$.
    Define the estimator
    {\small
    \begin{align}
        \widehat \varphi_t &:= \frac{\sum_{i=1}^t (\psi_i - (1-r)/2)}{tr},
    \end{align}}%
  and set $\widetilde B_t$ as in \eqref{eq:one-sided-changing-means-boundary}. Then,
    \begin{equation}
      \label{eq:a/b-lower-cs}
    \small
        \widetilde L_t^\Delta := -\frac{1}{1-\pi} + \left ( \frac{1}{\pi} + \frac{1}{1-\pi} \right ) \left ( \widehat \varphi_t - \widetilde B_t \right )
    \end{equation} 
    is a lower $(1-\alpha, \eps)$-\lpcs{} for $\widetilde \Delta_t$.
\end{corollary}

The proof is an immediate consequence of the well-known fact about ``inverse-probability-weighted'' estimators that $\EE f_t = \Delta_t$ for every $t$ \citep{horvitz1952generalization,robins1994estimation}, combined with Proposition~\ref{proposition:one-sided-cs-mean-so-far}. Similarly, a two-sided \cs{} can be obtained by replacing $\widehat \varphi_t - \widetilde B_t$ in \eqref{eq:a/b-lower-cs} with $\widehat \varphi_t \pm \widetilde B_t^\pm$, where $\widetilde B_t^\pm$ is given in \eqref{eq:two-sided-changing-means-boundary}.

\begin{figure}[!htbp]
  \centering
  \includegraphics[width=\ciplotwidth]{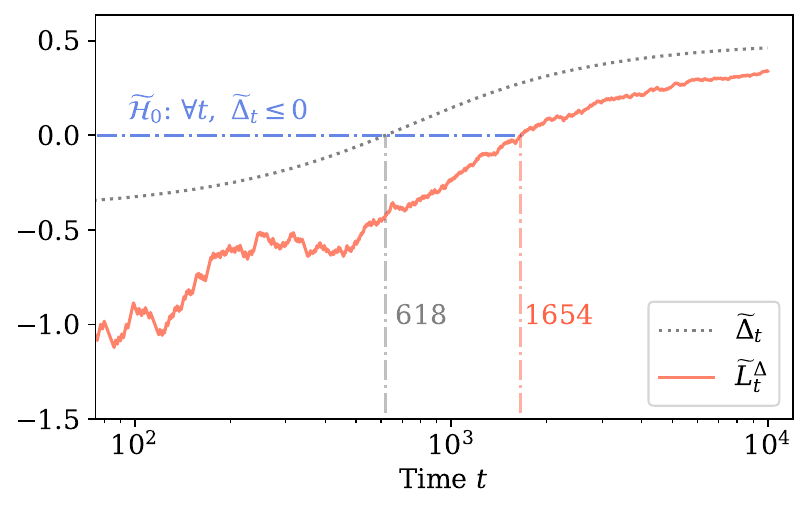}
  \caption{An example of Corollary~\ref{corollary:a/b-estimation} applied to the time-varying mean given by $\Delta_t := 1.8 (\exp\{t / 300\} / (1 + \exp\{t / 300\}) - 1/2) $.
    In this particular example, we have that $\widetilde \Delta_t := \frac{1}{t}\sum_{i=1}^t \Delta_i$ changes from negative to positive at time 618, and yet our lower \cs{} $\widetilde L_t^\Delta$ later detects this change at time 1654, at which point the weak null $\widetilde \Hcal_0$: $\forall t,\ \widetilde \Delta_t \leq 0$ can be rejected (see \cref{section:weak-null} for details regarding the composite hypothesis $\widetilde \Hcal_0$ and how to test it).}
  \label{fig:a/b-test}
\end{figure}

\noindent \textbf{Practical implications.} The implications of Corollary~\ref{corollary:a/b-estimation}
for the practitioner are threefold:
\begin{enumerate}[itemsep=0em]
  \item The \cs{}s can be continuously monitored from the start of the A/B test and for an indefinite amount of time;
  \item Inferences made from $\widetilde L_\tau^\Delta$ are valid at any stopping time $\tau$, regardless of why the test is stopped; and
  \item $\widetilde L_t^\Delta$
        adapts to non-stationarity: if the treatment differences $\Delta_t$ drift over time, $\widetilde L_t^\Delta$ still forms an \lpcs{} for $\widetilde \Delta_t$. But if $\Delta_1 = \Delta_2 = \dots = \Delta^\star$ is constant, then $\widetilde L_t^\Delta$  forms an \lpcs{} for $\Delta^\star$.
      \end{enumerate}



\section{Additional Results \& Summary}\label{section:summary}
Both \NPRR{} and our proof techniques are general-purpose tools with several other implications for locally private statistical inference, including confidence sets via the Laplace mechanism, variance-adaptive inference, and sequential hypothesis testing. We briefly expand on these implications here, and leave their details to the appendix.

\begin{itemize}[itemsep=0em]
\item \S\ref{section:laplace}: \textbf{Confidence sets via the Laplace mechanism.} We used \NPRR{} as an extension of randomized response for arbitrary bounded data (rather than just binary), but of course the Laplace mechanism also handles bounded data. While \NPRR{} enjoys advantages over Laplace as discussed in \cref{section:nprr}, it may still be of interest to derive confidence sets from data that are privatized via Laplace, given its ubiquity and simplicity. \cref{section:laplace} presents new nonparametric \ci{}s and \cs{}s for population means under the Laplace mechanism.
\item \S\ref{section:variance-adaptive}: \textbf{Variance-adaptive inference.} Notice that the \ci{}s and \cs{}s presented in~\cref{section:nprr-ci} were not variance-adaptive due to the fact that they relied on sub-Gaussianity of bounded random variables. However, this is not necessary, and we present other locally private \emph{variance-adaptive} \ci{}s and \cs{}s in \cref{section:variance-adaptive}. 
\item \S\ref{section:testing}: \textbf{Sequential hypothesis testing.} While the statistical procedures of this paper have taken the form of \ci{}s and \cs{}s rather than hypothesis tests, there is a deep relationship between the two, and our results have analogues that could have been presented in the language of the latter. \cref{section:testing} articulates this relationship and presents explicit (sequential) tests.
\item \S\ref{section:adaptive-a/b-testing}: \textbf{Adaptive online A/B testing.} \cref{corollary:a/b-estimation} assumes a common propensity score $\pi$ among all subjects for simplicity of exposition, but it is also possible to derive \cs{}s for $\widetilde \Delta_t$ under an adaptive framework where propensity scores $\infseqt{\pi_t(X_t)}$ can change over time in a data-dependent fashion, and be functions of some measured covariates $\infseqt{X_t}$. The details of this more complex setup are left to \cref{section:adaptive-a/b-testing}.
\end{itemize}
Another followup problem that we do not explicitly address here but that can be solved using our techniques is locally private \emph{variance} estimation. Notice that the variance $\Var(X) := \EE(X^2) - (\EE(X))^2$ is a function of two expectations, $\EE(X^2)$ and $\EE(X)$. Since $X^2$ is also $[0, 1]$-bounded if $X$ is, we can use all of the techniques in this paper to derive two separate $(1-\alpha/2, \eps/2)$-\lpci{}s (or \lpcs{}s) to derive a $(1-\alpha, \eps)$-\lpci{} for $\Var(X)$. Of course this requires collecting privatized views of both $X^2$ and $X$ separately. As a further generalization, a similar argument can be made for the construction of \lpci{}s for the covariance of $X$ and $Y$ since $\Cov(X,Y) = \EE(XY) - \EE(X)\EE(Y)$ (though here we would need to construct $(1-\alpha/3, \eps/3)$-\lpci{}s, etc.).

A limitation of the present paper is that we have only discussed confidence sets for univariate parameters. Indeed, it is not immediately clear to us what is the right way to generalize \NPRR{} to the multivariate case, or how to derive \lpci{}s and \lpcs{}s for means of random \emph{vectors} given such a generalization. This is an open direction for future work.

With the growing interest in protecting user privacy, an increasingly important addition to the statistician's toolbox are methods that can extract population information from privatized data. In this paper, we derived nonparametric confidence intervals and time-uniform confidence sequences for population means from locally private data. We studied and used \NPRR{} a nonparametric and sequentially interactive extension of Warner's randomized response for bounded data. The privatized output from \NPRR{} can then be harnessed to produce confidence sets for the mean of the raw data distribution. Importantly, our confidence sets are sharp, some attaining optimal theoretical convergence rates and others simply having excellent empirical performance, not only making private nonparametric (sequential) inference possible, but practical. In future work, we aim to apply these general-purpose tools to changepoint detection, two-sample testing, and (conditional) independence testing.

\paragraph{Acknowledgments.} We thank Vitaly Feldman for helpful conversations. ZSW was supported in part by NSF CNS2120667. AR acknowledges NSF DMS2310718.


\bibliographystyle{plainnat}
\bibliography{references}

\begin{thebibliography}{73}
\providecommand{\natexlab}[1]{#1}
\providecommand{\url}[1]{\texttt{#1}}
\expandafter\ifx\csname urlstyle\endcsname\relax
  \providecommand{\doi}[1]{doi: #1}\else
  \providecommand{\doi}{doi: \begingroup \urlstyle{rm}\Url}\fi

\bibitem[Acharya et~al.(2020{\natexlab{a}})Acharya, Canonne, Han, Sun, and
  Tyagi]{acharya2020domain}
Jayadev Acharya, Cl{\'e}ment~L Canonne, Yanjun Han, Ziteng Sun, and Himanshu
  Tyagi.
\newblock Domain compression and its application to randomness-optimal
  distributed goodness-of-fit.
\newblock In \emph{Proceedings of Thirty Third Conference on Learning Theory},
  volume 125, pages 3--40. PMLR, 2020{\natexlab{a}}.

\bibitem[Acharya et~al.(2020{\natexlab{b}})Acharya, Canonne, and
  Tyagi]{acharya2020inference}
Jayadev Acharya, Clément~L. Canonne, and Himanshu Tyagi.
\newblock Inference under information constraints {I}: Lower bounds from
  {Chi}-square contraction.
\newblock \emph{IEEE Transactions on Information Theory}, 66\penalty0
  (12):\penalty0 7835--7855, 2020{\natexlab{b}}.

\bibitem[Acharya et~al.(2021{\natexlab{a}})Acharya, Canonne, Freitag, Sun, and
  Tyagi]{acharya2021inference}
Jayadev Acharya, Clément~L. Canonne, Cody Freitag, Ziteng Sun, and Himanshu
  Tyagi.
\newblock Inference under information constraints {III}: Local privacy
  constraints.
\newblock \emph{IEEE Journal on Selected Areas in Information Theory},
  2\penalty0 (1):\penalty0 253--267, 2021{\natexlab{a}}.

\bibitem[Acharya et~al.(2021{\natexlab{b}})Acharya, Sun, and
  Zhang]{acharya2021differentially}
Jayadev Acharya, Ziteng Sun, and Huanyu Zhang.
\newblock Differentially private {Assouad}, {Fano}, and {Le Cam}.
\newblock In \emph{Algorithmic Learning Theory}, pages 48--78. PMLR,
  2021{\natexlab{b}}.

\bibitem[Acharya et~al.(2022)Acharya, Canonne, Liu, Sun, and
  Tyagi]{acharya2022interactive}
Jayadev Acharya, Clément~L. Canonne, Yuhan Liu, Ziteng Sun, and Himanshu
  Tyagi.
\newblock Interactive inference under information constraints.
\newblock \emph{IEEE Transactions on Information Theory}, 68\penalty0
  (1):\penalty0 502--516, 2022.

\bibitem[Amin et~al.(2020)Amin, Joseph, and Mao]{amin2020pan}
Kareem Amin, Matthew Joseph, and Jieming Mao.
\newblock Pan-private uniformity testing.
\newblock In \emph{Proceedings of Thirty Third Conference on Learning Theory},
  volume 125, pages 183--218. PMLR, 2020.

\bibitem[{Apple Inc.}(2022)]{appleEps}
{Apple Inc.}
\newblock Differential privacy overview.
\newblock
  \url{https://www.apple.com/privacy/docs/Differential_Privacy_Overview.pdf},
  2022.
\newblock Accessed: 2022-02-01.

\bibitem[Audibert et~al.(2007)Audibert, Munos, and
  Szepesv{\'a}ri]{audibert2007tuning}
Jean-Yves Audibert, R{\'e}mi Munos, and Csaba Szepesv{\'a}ri.
\newblock Tuning bandit algorithms in stochastic environments.
\newblock In \emph{International conference on algorithmic learning theory},
  pages 150--165. Springer, 2007.

\bibitem[Awan and Slavkovi{\'c}(2018)]{awan2018differentially}
Jordan Awan and Aleksandra Slavkovi{\'c}.
\newblock Differentially private uniformly most powerful tests for binomial
  data.
\newblock \emph{Advances in Neural Information Processing Systems},
  31:\penalty0 4208--4218, 2018.

\bibitem[Balle et~al.(2019)Balle, Bell, Gasc{\'o}n, and
  Nissim]{balle2019privacy}
Borja Balle, James Bell, Adri{\`a} Gasc{\'o}n, and Kobbi Nissim.
\newblock The privacy blanket of the shuffle model.
\newblock In \emph{Advances in Cryptology--CRYPTO 2019: 39th Annual
  International Cryptology Conference, Santa Barbara, CA, USA, August 18--22,
  2019, Proceedings, Part II 39}, pages 638--667. Springer, 2019.

\bibitem[Barnes et~al.(2020)Barnes, Chen, and Özgür]{barnes2020fisher}
Leighton~Pate Barnes, Wei-Ning Chen, and Ayfer Özgür.
\newblock Fisher information under local differential privacy.
\newblock \emph{IEEE Journal on Selected Areas in Information Theory},
  1\penalty0 (3):\penalty0 645--659, 2020.

\bibitem[Barnes et~al.(1951)Barnes, Cooke-Yarborough, and
  Thomas]{barnes1951electronic}
RCM Barnes, EH~Cooke-Yarborough, and DGA Thomas.
\newblock An electronic digital computor using cold cathode counting tubes for
  storage.
\newblock \emph{Electronic Engineering}, 23:\penalty0 286--91, 1951.

\bibitem[Bentkus(2004)]{bentkus2004hoeffding}
Vidmantas Bentkus.
\newblock On {Hoeffding}’s inequalities.
\newblock \emph{The Annals of Probability}, 32\penalty0 (2):\penalty0
  1650--1673, 2004.

\bibitem[Berrett and Butucea(2020)]{berrett2020locally}
Thomas Berrett and Cristina Butucea.
\newblock Locally private non-asymptotic testing of discrete distributions is
  faster using interactive mechanisms.
\newblock In \emph{Advances in Neural Information Processing Systems},
  volume~33, pages 3164--3173, 2020.

\bibitem[Berrett and Yu(2021)]{berrett2021locally}
Tom Berrett and Yi~Yu.
\newblock Locally private online change point detection.
\newblock \emph{Advances in Neural Information Processing Systems}, 34, 2021.

\bibitem[Butucea and Issartel(2021)]{butucea2021locally}
Cristina Butucea and Yann Issartel.
\newblock Locally differentially private estimation of functionals of discrete
  distributions.
\newblock In M.~Ranzato, A.~Beygelzimer, Y.~Dauphin, P.S. Liang, and J.~Wortman
  Vaughan, editors, \emph{Advances in Neural Information Processing Systems},
  volume~34, pages 24753--24764. Curran Associates, Inc., 2021.

\bibitem[Butucea et~al.(2020)Butucea, Dubois, Kroll, and
  Saumard]{butucea2020local}
Cristina Butucea, Amandine Dubois, Martin Kroll, and Adrien Saumard.
\newblock Local differential privacy: {Elbow} effect in optimal density
  estimation and adaptation over {Besov} ellipsoids.
\newblock \emph{Bernoulli}, 26\penalty0 (3):\penalty0 1727--1764, 2020.

\bibitem[Canonne et~al.(2019)Canonne, Kamath, McMillan, Smith, and
  Ullman]{canonne2019structure}
Cl{\'e}ment~L Canonne, Gautam Kamath, Audra McMillan, Adam Smith, and Jonathan
  Ullman.
\newblock The structure of optimal private tests for simple hypotheses.
\newblock In \emph{Proceedings of the 51st Annual ACM SIGACT Symposium on
  Theory of Computing}, pages 310--321, 2019.

\bibitem[Couch et~al.(2019)Couch, Kazan, Shi, Bray, and
  Groce]{couch2019differentially}
Simon Couch, Zeki Kazan, Kaiyan Shi, Andrew Bray, and Adam Groce.
\newblock Differentially private nonparametric hypothesis testing.
\newblock In \emph{Proceedings of the 2019 ACM SIGSAC Conference on Computer
  and Communications Security}, pages 737--751, 2019.

\bibitem[Covington et~al.(2021)Covington, He, Honaker, and
  Kamath]{covington2021unbiased}
Christian Covington, Xi~He, James Honaker, and Gautam Kamath.
\newblock Unbiased statistical estimation and valid confidence intervals under
  differential privacy.
\newblock \emph{arXiv preprint arXiv:2110.14465}, 2021.

\bibitem[Darling and Robbins(1967)]{darling1967confidence}
Donald~A Darling and Herbert Robbins.
\newblock Confidence sequences for mean, variance, and median.
\newblock \emph{Proceedings of the National Academy of Sciences of the United
  States of America}, 58\penalty0 (1):\penalty0 66, 1967.

\bibitem[Ding et~al.(2017)Ding, Kulkarni, and Yekhanin]{ding2017collecting}
Bolin Ding, Janardhan Kulkarni, and Sergey Yekhanin.
\newblock Collecting telemetry data privately.
\newblock In \emph{Proceedings of the 31st International Conference on Neural
  Information Processing Systems}, pages 3574--3583, 2017.

\bibitem[Drechsler et~al.(2021)Drechsler, Globus-Harris, McMillan, Sarathy, and
  Smith]{drechsler2021non}
Joerg Drechsler, Ira Globus-Harris, Audra McMillan, Jayshree Sarathy, and Adam
  Smith.
\newblock Non-parametric differentially private confidence intervals for the
  median.
\newblock \emph{arXiv preprint arXiv:2106.10333}, 2021.

\bibitem[Duchi and Ruan(2018)]{duchi2018right}
John~C Duchi and Feng Ruan.
\newblock The right complexity measure in locally private estimation: It is not
  the {F}isher information.
\newblock \emph{arXiv preprint arXiv:1806.05756}, 2018.

\bibitem[Duchi et~al.(2013{\natexlab{a}})Duchi, Jordan, and
  Wainwright]{duchi2013local-FOCS}
John~C Duchi, Michael~I Jordan, and Martin~J Wainwright.
\newblock Local privacy and statistical minimax rates.
\newblock In \emph{2013 IEEE 54th Annual Symposium on Foundations of Computer
  Science}, pages 429--438. IEEE, 2013{\natexlab{a}}.

\bibitem[Duchi et~al.(2013{\natexlab{b}})Duchi, Jordan, and
  Wainwright]{duchi2013local-NeurIPS}
John~C Duchi, Michael~I Jordan, and Martin~J Wainwright.
\newblock Local privacy and minimax bounds: sharp rates for probability
  estimation.
\newblock In \emph{Proceedings of the 26th International Conference on Neural
  Information Processing Systems-Volume 1}, pages 1529--1537,
  2013{\natexlab{b}}.

\bibitem[Duchi et~al.(2018)Duchi, Jordan, and Wainwright]{duchi2018minimax}
John~C Duchi, Michael~I Jordan, and Martin~J Wainwright.
\newblock Minimax optimal procedures for locally private estimation.
\newblock \emph{Journal of the American Statistical Association}, 113\penalty0
  (521):\penalty0 182--201, 2018.

\bibitem[Dwork and Roth(2014)]{dwork2014algorithmic}
Cynthia Dwork and Aaron Roth.
\newblock The algorithmic foundations of differential privacy.
\newblock \emph{Found. Trends Theor. Comput. Sci.}, 9\penalty0 (3-4):\penalty0
  211--407, 2014.

\bibitem[Dwork et~al.(2006)Dwork, McSherry, Nissim, and
  Smith]{dwork2006calibrating}
Cynthia Dwork, Frank McSherry, Kobbi Nissim, and Adam Smith.
\newblock Calibrating noise to sensitivity in private data analysis.
\newblock In \emph{Theory of cryptography conference}, pages 265--284.
  Springer, 2006.

\bibitem[Fan et~al.(2015)Fan, Grama, and Liu]{fan2015exponential}
Xiequan Fan, Ion Grama, and Quansheng Liu.
\newblock Exponential inequalities for martingales with applications.
\newblock \emph{Electronic Journal of Probability}, 20:\penalty0 1--22, 2015.

\bibitem[Ferrando et~al.(2022)Ferrando, Wang, and Sheldon]{ferrando2020general}
Cecilia Ferrando, Shufan Wang, and Daniel Sheldon.
\newblock Parametric bootstrap for differentially private confidence intervals.
\newblock In \emph{International Conference on Artificial Intelligence and
  Statistics}, pages 1598--1618. PMLR, 2022.

\bibitem[Forsythe(1959)]{forsythe1959reprint}
George~E Forsythe.
\newblock Reprint of a note on rounding-off errors.
\newblock \emph{SIAM review}, 1\penalty0 (1):\penalty0 66, 1959.

\bibitem[Gaboardi et~al.(2016)Gaboardi, Lim, Rogers, and
  Vadhan]{gaboardi2016differentially}
Marco Gaboardi, Hyun Lim, Ryan Rogers, and Salil Vadhan.
\newblock Differentially private chi-squared hypothesis testing: Goodness of
  fit and independence testing.
\newblock In \emph{International conference on machine learning}, pages
  2111--2120. PMLR, 2016.

\bibitem[Gaboardi et~al.(2019)Gaboardi, Rogers, and
  Sheffet]{gaboardi2019locally}
Marco Gaboardi, Ryan Rogers, and Or~Sheffet.
\newblock Locally private mean estimation: $ z $-test and tight confidence
  intervals.
\newblock In \emph{The 22nd International Conference on Artificial Intelligence
  and Statistics}, pages 2545--2554. PMLR, 2019.

\bibitem[Gr{\"u}nwald et~al.(2019)Gr{\"u}nwald, de~Heide, and
  Koolen]{grunwald2019safe}
Peter Gr{\"u}nwald, Rianne de~Heide, and Wouter Koolen.
\newblock Safe testing.
\newblock \emph{arXiv preprint arXiv:1906.07801}, 2019.

\bibitem[Hoeffding(1963)]{hoeffding1963probability}
Wassily Hoeffding.
\newblock Probability inequalities for sums of bounded random variables.
\newblock \emph{Journal of the American Statistical Association}, 58\penalty0
  (301):\penalty0 13--30, 1963.

\bibitem[Horvitz and Thompson(1952)]{horvitz1952generalization}
Daniel~G Horvitz and Donovan~J Thompson.
\newblock A generalization of sampling without replacement from a finite
  universe.
\newblock \emph{Journal of the American statistical Association}, 47\penalty0
  (260):\penalty0 663--685, 1952.

\bibitem[Howard et~al.(2020)Howard, Ramdas, McAuliffe, and
  Sekhon]{howard2020time}
Steven~R Howard, Aaditya Ramdas, Jon McAuliffe, and Jasjeet Sekhon.
\newblock Time-uniform {Chernoff} bounds via nonnegative supermartingales.
\newblock \emph{Probability Surveys}, 17:\penalty0 257--317, 2020.

\bibitem[Howard et~al.(2021)Howard, Ramdas, McAuliffe, and
  Sekhon]{howard2021time}
Steven~R Howard, Aaditya Ramdas, Jon McAuliffe, and Jasjeet Sekhon.
\newblock Time-uniform, nonparametric, nonasymptotic confidence sequences.
\newblock \emph{The Annals of Statistics}, 49\penalty0 (2):\penalty0
  1055--1080, 2021.

\bibitem[Hull and Swenson(1966)]{hull1966tests}
Thomas~E Hull and J~Richard Swenson.
\newblock Tests of probabilistic models for propagation of roundoff errors.
\newblock \emph{Communications of the ACM}, 9\penalty0 (2):\penalty0 108--113,
  1966.

\bibitem[Johari et~al.(2017)Johari, Koomen, Pekelis, and
  Walsh]{johari2017peeking}
Ramesh Johari, Pete Koomen, Leonid Pekelis, and David Walsh.
\newblock Peeking at {A/B} tests: {Why} it matters, and what to do about it.
\newblock In \emph{Proceedings of the 23rd ACM SIGKDD International Conference
  on Knowledge Discovery and Data Mining}, pages 1517--1525, 2017.

\bibitem[Joseph et~al.(2019)Joseph, Kulkarni, Mao, and Wu]{joseph2019locally}
Matthew Joseph, Janardhan Kulkarni, Jieming Mao, and Steven~Z Wu.
\newblock Locally private {Gaussian} estimation.
\newblock \emph{Advances in Neural Information Processing Systems},
  32:\penalty0 2984--2993, 2019.

\bibitem[Jun and Orabona(2019)]{jun2019parameter}
Kwang-Sung Jun and Francesco Orabona.
\newblock Parameter-free online convex optimization with sub-exponential noise.
\newblock In \emph{Conference on Learning Theory}, pages 1802--1823. PMLR,
  2019.

\bibitem[Kairouz et~al.(2014)Kairouz, Oh, and Viswanath]{kairouz2014extremal}
Peter Kairouz, Sewoong Oh, and Pramod Viswanath.
\newblock Extremal mechanisms for local differential privacy.
\newblock \emph{Advances in neural information processing systems}, 27, 2014.

\bibitem[Kairouz et~al.(2016)Kairouz, Bonawitz, and
  Ramage]{kairouz2016discrete}
Peter Kairouz, Keith Bonawitz, and Daniel Ramage.
\newblock Discrete distribution estimation under local privacy.
\newblock In \emph{International Conference on Machine Learning}, pages
  2436--2444. PMLR, 2016.

\bibitem[Kamath et~al.(2020)Kamath, Singhal, and Ullman]{kamath2020private}
Gautam Kamath, Vikrant Singhal, and Jonathan Ullman.
\newblock Private mean estimation of heavy-tailed distributions.
\newblock In \emph{Conference on Learning Theory}, pages 2204--2235. PMLR,
  2020.

\bibitem[Karwa and Vadhan(2018)]{karwa2018finite}
Vishesh Karwa and Salil Vadhan.
\newblock Finite sample differentially private confidence intervals.
\newblock In \emph{9th Innovations in Theoretical Computer Science Conference
  (ITCS 2018)}. Schloss Dagstuhl-Leibniz-Zentrum fuer Informatik, 2018.

\bibitem[Kasiviswanathan et~al.(2011)Kasiviswanathan, Lee, Nissim,
  Raskhodnikova, and Smith]{kasiviswanathan2011can}
Shiva~Prasad Kasiviswanathan, Homin~K Lee, Kobbi Nissim, Sofya Raskhodnikova,
  and Adam Smith.
\newblock What can we learn privately?
\newblock \emph{SIAM Journal on Computing}, 40\penalty0 (3):\penalty0 793--826,
  2011.

\bibitem[Kuchibhotla and Zheng(2021)]{kuchibhotla2021near}
Arun~K Kuchibhotla and Qinqing Zheng.
\newblock Near-optimal confidence sequences for bounded random variables.
\newblock In \emph{International Conference on Machine Learning}, pages
  5827--5837. PMLR, 2021.

\bibitem[Li et~al.(2020)Li, Wang, Lopuha{\"a}-Zwakenberg, Li, and
  {\v{S}}koric]{li2020estimating}
Zitao Li, Tianhao Wang, Milan Lopuha{\"a}-Zwakenberg, Ninghui Li, and Boris
  {\v{S}}koric.
\newblock Estimating numerical distributions under local differential privacy.
\newblock In \emph{Proceedings of the 2020 ACM SIGMOD International Conference
  on Management of Data}, pages 621--635, 2020.

\bibitem[Maurer and Pontil(2009)]{maurer2009empirical}
Andreas Maurer and Massimiliano Pontil.
\newblock Empirical {Bernstein} bounds and sample variance penalization.
\newblock In \emph{Conference on Learning Theory}, pages 2372--2387. PMLR,
  2009.

\bibitem[Mironov(2012)]{mironov2012significance}
Ilya Mironov.
\newblock On significance of the least significant bits for differential
  privacy.
\newblock In \emph{Proceedings of the 2012 ACM conference on Computer and
  communications security}, pages 650--661, 2012.

\bibitem[Orabona and Jun(2021)]{orabona2021tight}
Francesco Orabona and Kwang-Sung Jun.
\newblock Tight concentrations and confidence sequences from the regret of
  universal portfolio.
\newblock \emph{arXiv preprint arXiv:2110.14099}, 2021.

\bibitem[Ramdas et~al.(2021)Ramdas, Ruf, Larsson, and
  Koolen]{ramdas2021testing}
Aaditya Ramdas, Johannes Ruf, Martin Larsson, and Wouter~M Koolen.
\newblock Testing exchangeability: Fork-convexity, supermartingales and
  e-processes.
\newblock \emph{International Journal of Approximate Reasoning}, 2021.

\bibitem[Robbins(1970)]{robbins1970statistical}
Herbert Robbins.
\newblock Statistical methods related to the law of the iterated logarithm.
\newblock \emph{The Annals of Mathematical Statistics}, 41\penalty0
  (5):\penalty0 1397--1409, 1970.

\bibitem[Robins et~al.(1994)Robins, Rotnitzky, and Zhao]{robins1994estimation}
James~M Robins, Andrea Rotnitzky, and Lue~Ping Zhao.
\newblock Estimation of regression coefficients when some regressors are not
  always observed.
\newblock \emph{Journal of the American statistical Association}, 89\penalty0
  (427):\penalty0 846--866, 1994.

\bibitem[Shafer(2021)]{shafer2021testing}
Glenn Shafer.
\newblock Testing by betting: A strategy for statistical and scientific
  communication.
\newblock \emph{Journal of the Royal Statistical Society: Series A (Statistics
  in Society)}, 184\penalty0 (2):\penalty0 407--431, 2021.

\bibitem[Shafer et~al.(2011)Shafer, Shen, Vereshchagin, and
  Vovk]{shafer2011test}
Glenn Shafer, Alexander Shen, Nikolai Vereshchagin, and Vladimir Vovk.
\newblock Test martingales, {Bayes} factors and p-values.
\newblock \emph{Statistical Science}, 26\penalty0 (1):\penalty0 84--101, 2011.

\bibitem[ter Schure and Gr{\"u}nwald(2021)]{ter2021all}
Judith ter Schure and Peter Gr{\"u}nwald.
\newblock {ALL-IN} meta-analysis: breathing life into living systematic
  reviews.
\newblock \emph{arXiv preprint arXiv:2109.12141}, 2021.

\bibitem[Ville(1939)]{ville1939etude}
Jean Ville.
\newblock Etude critique de la notion de collectif.
\newblock \emph{Bull. Amer. Math. Soc}, 45\penalty0 (11):\penalty0 824, 1939.

\bibitem[Vovk and Wang(2021)]{vovk2021values}
Vladimir Vovk and Ruodu Wang.
\newblock E-values: Calibration, combination and applications.
\newblock \emph{The Annals of Statistics}, 49\penalty0 (3):\penalty0
  1736--1754, 2021.

\bibitem[Vu and Slavkovic(2009)]{vu2009differential}
Duy Vu and Aleksandra Slavkovic.
\newblock Differential privacy for clinical trial data: Preliminary
  evaluations.
\newblock In \emph{2009 IEEE International Conference on Data Mining
  Workshops}, pages 138--143. IEEE, 2009.

\bibitem[Wald(1945)]{wald1945sequential}
Abraham Wald.
\newblock Sequential tests of statistical hypotheses.
\newblock \emph{The annals of mathematical statistics}, 16\penalty0
  (2):\penalty0 117--186, 1945.

\bibitem[Wang et~al.(2019)Wang, Xiao, Yang, Zhao, Hui, Shin, Shin, and
  Yu]{wang2019collecting}
Ning Wang, Xiaokui Xiao, Yin Yang, Jun Zhao, Siu~Cheung Hui, Hyejin Shin,
  Junbum Shin, and Ge~Yu.
\newblock Collecting and analyzing multidimensional data with local
  differential privacy.
\newblock In \emph{2019 IEEE 35th International Conference on Data Engineering
  (ICDE)}, pages 638--649. IEEE, 2019.

\bibitem[Wang and Ramdas(2022)]{wang2020false}
Ruodu Wang and Aaditya Ramdas.
\newblock False discovery rate control with e-values.
\newblock \emph{Journal of the Royal Statistical Society, Series B}, 2022.

\bibitem[Wang et~al.(2022)Wang, Sibai, Yen, Mitra, and
  Dullerud]{wang2020differential}
Yu~Wang, Hussein Sibai, Mark Yen, Sayan Mitra, and Geir~E Dullerud.
\newblock Differentially private algorithms for statistical verification of
  cyber-physical systems.
\newblock \emph{IEEE Open Journal of Control Systems}, 1:\penalty0 294--305,
  2022.

\bibitem[Wang et~al.(2015)Wang, Lee, and Kifer]{wang2015revisiting}
Yue Wang, Jaewoo Lee, and Daniel Kifer.
\newblock Differentially private hypothesis testing, revisited.
\newblock \emph{arXiv preprint arXiv:1511.03376}, 2015.

\bibitem[Warner(1965)]{warner1965randomized}
Stanley~L Warner.
\newblock Randomized response: A survey technique for eliminating evasive
  answer bias.
\newblock \emph{Journal of the American Statistical Association}, 60\penalty0
  (309):\penalty0 63--69, 1965.

\bibitem[Wasserman and Zhou(2010)]{wasserman2010statistical}
Larry Wasserman and Shuheng Zhou.
\newblock A statistical framework for differential privacy.
\newblock \emph{Journal of the American Statistical Association}, 105\penalty0
  (489):\penalty0 375--389, 2010.

\bibitem[Waudby-Smith and Ramdas(2020)]{waudby2020confidence}
Ian Waudby-Smith and Aaditya Ramdas.
\newblock Confidence sequences for sampling without replacement.
\newblock \emph{Advances in Neural Information Processing Systems}, 33, 2020.

\bibitem[Waudby-Smith and Ramdas(2023)]{waudby2020estimating}
Ian Waudby-Smith and Aaditya Ramdas.
\newblock Estimating means of bounded random variables by betting.
\newblock \emph{Journal of the Royal Statistical Society, Series B (to appear
  with discussion)}, 2023.

\bibitem[Waudby-Smith et~al.(2022)Waudby-Smith, Wu, Ramdas, Karampatziakis, and
  Mineiro]{waudby2022anytime}
Ian Waudby-Smith, Lili Wu, Aaditya Ramdas, Nikos Karampatziakis, and Paul
  Mineiro.
\newblock Anytime-valid off-policy inference for contextual bandits.
\newblock \emph{arXiv preprint arXiv:2210.10768}, 2022.

\bibitem[Zhao et~al.(2016)Zhao, Zhou, Sabharwal, and Ermon]{zhao2016adaptive}
Shengjia Zhao, Enze Zhou, Ashish Sabharwal, and Stefano Ermon.
\newblock Adaptive concentration inequalities for sequential decision problems.
\newblock \emph{Advances in Neural Information Processing Systems}, 29, 2016.

\end{thebibliography}

\appendix
\section{Proofs of main results}

\subsection{Prelude: filtrations, supermartingales, and Ville's inequality}
\label{section:background-filtration-supermartingales-ville}

By far the most common way to derive a \cs{}
is by constructing a nonnegative supermartingales and then applying Ville's maximal inequality to it. Indeed, all of the proofs for our \cs{} and \ci{} results employ this technique. However, in order to discuss supermartingales we must first review \emph{filtrations}. A filtration $\Fcal \equiv \infseq \Fcal t0$ is a nondecreasing sequence of sigma-algebras $\Fcal_0 \subseteq \Fcal_1 \subseteq \cdots$, and a stochastic process $\infseq Mt0$ is said to be \emph{adapted} to $\Fcal$ if $M_t$ is $\Fcal_t$-measurable for all $t \in \NN$. On the other hand, $\infseq Mt1$ is said to be $\Fcal$-predictable if each $M_t$ is $\Fcal_{t-1}$-measurable --- informally ``$M_t$ depends on the past''.

For example, the canonical filtration $\Xcal$ generated by a sequence of random variables $\infseq Xt1$ is given by the sigma-algebra generated by $X_1^t$, i.e. $\Xcal_t := \sigma(X_1^t)$ for each $t \in \{1,2,\dots\}$, and $\Xcal_0$ is the trivial sigma-algebra. A function $M_t \equiv M(X_1^t)$ depending only on $X_1^t$ forms a $\Xcal$-adapted process, while $(M_{t-1})_{t=1}^\infty$ is $\Xcal$-predictable. Likewise, if we obtain a privatized view $\infseq Zt1$ of $\infseq Xt1$ using some locally private mechanism, a different filtration $\Zcal$ emerges, given by $\Zcal_t := \sigma(Z_1^t)$. Throughout our proofs, $\Zcal$-adapted and $\Zcal$-predictable processes will be central mathematical objects.

A process $\infseq Mt0$ adapted to $\Fcal$ is a \emph{supermartingale} if
\begin{equation}
\small
  \label{eq:martingale}
  \EE (M_t \mid \Fcal_{t-1}) \leq M_{t-1} ~~\text{for each $t \geq 1$.}
\end{equation}
If the above inequality is replaced by an equality, then $\infseq Mt0$ is a \emph{martingale}.
The methods in this paper will involve derivations of (super)martingales which are nonnegative and begin at one --- often referred to as ``test (super)martingales'' \citep{shafer2011test} or simply ``nonnegative (super)martingales'' (NMs or NSMs for martingales and supermartingales, respectively) \citep{robbins1970statistical,howard2020time}. NSMs $\infseq Mt0$ satisfy the following powerful concentration inequality due to \citet{ville1939etude}:
\begin{equation}
\small
  \label{eq:ville}
  \PP(\exists t \in \NN : M_t \geq 1/\alpha) \leq \alpha.
\end{equation}
In other words, they are unlikely to ever grow too large.

In the \cs{} proofs that follow, we will focus on deriving processes $(M_t(\mu))_{t=1}^\infty$ for any $\mu \in [0, 1]$ such that when $\mu$ is equal to the true mean of interest $\mu^\star$, we have that $M_t(\mu^\star)$ forms a NSM. In this case, it turns out that the set of $\mu$ such that $M_t(\mu)$ is less than $1/\alpha$ --- i.e. $C_t := \{ \mu \in [0, 1] : M_t(\mu) < 1/\alpha \}$ --- forms a $(1-\alpha)$-\cs{} for $\mu^\star$. This is easy to see since $\mu^\star \notin C_t$ if and only if $M_t(\mu^\star) \geq 1/\alpha$, and thus
\begin{equation}
  \PP(\exists t \in \NN : \mu^\star \notin C_t) = \PP(\exists t \in \NN : M_t(\mu^\star) \geq 1/\alpha) \leq \alpha,
\end{equation}
where the last inequality is precisely \eqref{eq:ville}. The \cs{} proofs that follow will make the exact processes $(M_t(\mu))_{t=1}^\infty$ explicit.

\subsection{Proof of Theorem~\ref{theorem:NPRR-DP}}\label{proof:NPRR-DP}
\NPRRLDP*
\begin{proof}
We will prove the result for fixed $r \in (0, 1), G \geq 1$ but it is straightforward to generalize the proof for $r_t$ depending on $Z_1^{t-1}$. It suffices to verify that the likelihood ratio $L(x, \widetilde x)$ is bounded above by $\exp(\eps)$ for any $x, \widetilde x \in [0,1]$. Writing out the likelihood ratio $L(x, \widetilde x)$, we have
\[ L(x, \widetilde x) := \frac{ \frac{1-r}{G+1} + rG \cdot \left \{ \1(Z = x^\mathrm{ceil}) (x - x^\mathrm{floor}) + \1(Z = x^\mathrm{floor}) \left [1/G - (x-x^\mathrm{floor}) \right ] \right \} }{ \frac{1-r}{G+1} + rG \cdot \left \{ \1(Z = \widetilde x^\mathrm{ceil}) ( \widetilde x - \widetilde x^\mathrm{floor}) + \1(Z = \widetilde x^\mathrm{floor}) \left [1/G - (\widetilde x-\widetilde x^\mathrm{floor}) \right ] \right \}}, \]
which is dominated by the counting measure.
Notice that the numerator of $L$ is maximized when $x$ already lies in the discretized range, i.e. $Z = x = x^\mathrm{ceil} = x^\mathrm{floor}$ so that the numerator becomes $\frac{1-r}{G+1} + r$, while the denominator is minimized when $Z \neq \widetilde x^\mathrm{ceil}$ and $Z \neq \widetilde x^\mathrm{floor}$ so that the denominator becomes $\frac{1-r}{G+1}$. Therefore, we have that with probability one,
\begin{align*}
    L(x, \widetilde x)\leq \frac{\frac{1-r}{G+1} + r}{\frac{1-r}{G+1}} = 1 + \frac{(G+1)r}{1-r},
\end{align*}
and thus \NPRR{} is $\eps$-locally DP with $\eps := \log (1 + (G+1)r / (1-r))$.
\end{proof}

\subsection{Proof of Theorem~\ref{theorem:hoeffding-nprr-ci}}\label{proof:hoeffding-nprr-ci}
\HoeffdingNPRRCI*
\begin{proof}
The proof proceeds in two steps. First we note that $\bar L_{t}^\Hoeff$ forms a $(1-\alpha)$-lower confidence \emph{sequence}, and then instantiate this fact at the sample size $n$.

\paragraph{Step 1. $\bar L_{t}^\Hoeff$ forms a $(1-\alpha)$-lower \cs{}.}
This is exactly the statement of Theorem~\ref{theorem:ldp-hoeffding-cs}.

\paragraph{Step 2. $\dot L_n^\Hoeff$ is a lower-\ci{}.}
By Step 1, we have that $\bar L_{t}^\Hoeff$ forms a $(1-\alpha)$-lower \cs{}, meaning
\[ \PP(\forall t \in \{ 1, \dots, n\}, \ \mu^\star \geq \bar L_{t}^\Hoeff) \geq 1-\alpha. \]
Therefore,
\[ \PP\left (\mu^\star \geq \max_{1\leq t \leq n} \bar L_{t}^\Hoeff \right ) = \PP (\mu^\star \geq \dot L_n^\Hoeff) \geq 1-\alpha,\]
which completes the proof.
\end{proof}

\subsection{Proof of Theorem~\ref{theorem:ldp-hoeffding-cs}}\label{proof:ldp-hoeffding-cs}
\LDPHoeffdingCS*
\begin{proof}
  The proof proceeds in two steps. First, we construct an NSM adapted to the private filtration $\Zcal \equiv \infseq \Zcal t0$. Second and finally, we apply Ville's inequality to obtain a high-probability upper bound on the NSM, and show that this inequality results in the \cs{} given in Theorem~\ref{theorem:ldp-hoeffding-cs}.

  \paragraph{Step 1.} Consider the nonnegative process starting at one given by
  \begin{equation}
    M_t(\mu^\star) := \prod_{i=1}^t \exp \left \{ \lambda_i ( Z_i - \zeta_i(\mu^\star) ) - \lambda_i^2 / 8 \right \},
    \label{eq:nprr-hoeffding-nsm}
\end{equation}
where $\infseq \lambda t1$ is a real-valued sequence\footnote{The proof also works if $\infseq \lambda t1$ is $\Zcal$-predictable but we omit this detail since we typically recommend using real-valued sequences anyway.} and $\zeta_t(\mu^\star) := r_t \mu^\star + (1-r_t)/2$ as usual. We claim that $(M_t(\mu^\star))_{t=0}^\infty$ is a supermartingale, meaning $\EE(M_t(\mu^\star) \mid \Zcal_{t-1}) \leq M_{t-1}(\mu^\star)$. Writing out the conditional expectation of $M_t(\mu^\star)$, we have
  \begin{align*}
    &\EE \left ( M_t(\mu^\star) \mid \Zcal_{t-1} \right  )\\
    =\ & \EE \left ( \prod_{i=1}^t \exp \left \{ \lambda_i ( Z_i - \zeta_i(\mu^\star) ) - \lambda_i^2 / 8 \right \} \biggm \vert \Zcal_{t-1} \right ) \\
    =\ & \underbrace{\prod_{i=1}^{t-1} \exp \left \{ \lambda_i ( Z_i - \zeta_i(\mu^\star) ) - \lambda_i^2 / 8 \right \}}_{M_{t-1}(\mu^\star)} \cdot \underbrace{\EE \left ( \exp \left \{ \lambda_t ( Z_t - \zeta_t(\mu^\star) ) - \lambda_t^2 / 8 \right \} \biggm \vert \Zcal_{t-1} \right )}_{(\dagger)},
  \end{align*}
  since $M_{t-1}(\mu^\star)$ is $\Zcal_{t-1}$-measurable, and thus it can be written outside of the conditional expectation. It now suffices to show that $(\dagger) \leq 1$. To this end, note that $Z_t$ is a $[0, 1]$-bounded random variable with conditional mean $\EE(Z_t \mid \Zcal_{t-1}) = \zeta_t(\mu^\star)$ by design of \NPRR{} (Algorithm~\ref{algorithm:NPRR}). Since bounded random variables are sub-Gaussian \citep{hoeffding1963probability}, we have that
  \[ \EE (\lambda_t (Z_t - \zeta_t(\mu^\star)) \mid \Zcal_{t-1} ) \leq \exp \left \{ \lambda_t^2 / 8 \right \}, \]
  and hence $(\dagger) \leq 1$. Therefore, $(M_t(\mu^\star))_{t=0}^\infty$ is a $\QcalNPRRstarinf$-NSM.

  \paragraph{Step 2.} By Ville's inequality for NSMs \citep{ville1939etude}, we have that
  \[ \PP (\exists t : M_t(\mu^\star) \geq 1/\alpha) \leq \alpha. \]
  In other words, we have that $M_t(\mu^\star) < 1/\alpha$ for all $t \in \NN$ with probability at least $1-\alpha$. Using some algebra to rewrite the inequality $M_t(\mu^\star) < 1/\alpha$, we have
  \begin{align*}
    M_t(\mu^\star) < 1/\alpha &\iff \prod_{i=1}^t \exp \left \{ \lambda_i ( Z_i - \zeta_i(\mu^\star) ) - \lambda_i^2 / 8 \right \} < \frac{1}{\alpha} \\
                              &\iff \sum_{i=1}^t \left [ \lambda_i(Z_i - \zeta_i(\mu^\star)) - \lambda_i^2 / 8 \right ] < \log(1/\alpha) \\
                              &\iff \sum_{i=1}^t \lambda_i Z_i - \mu^\star \sum_{i=1}^t\lambda_ir_i - \sum_{i=1}^t \lambda_i \cdot (1-r_i)/2 - \sum_{i=1}^t \lambda_i^2/8 < \log(1/\alpha) \\
                              &\iff \mu^\star > \underbrace{\frac{\sum_{i=1}^t \lambda_i \cdot (Z_i - (1-r_i)/2)}{\sum_{i=1}^t r_i\lambda_i}}_{\widehat \mu_t(\lambda_1^t)} - \underbrace{\frac{\log(1/\alpha) + \sum_{i=1}^t \lambda_i^2/8}{\sum_{i=1}^tr_i\lambda_i}}_{\bar B_t(\lambda_1^t)}
  \end{align*}
  Therefore, $\bar L_t := \widehat \mu_t(\lambda_1^t) - \bar B_t(\lambda_1^t)$ forms a lower $(1-\alpha)$-\cs{} for $\mu^\star$. The upper \cs{} $\bar U_t := \widehat \mu_t(\lambda_1^t) + \bar B_t(\lambda_1^t)$ can be derived by applying the above proof to $\infseq {-Z}t1$ and their conditional means $(-\zeta_i(\mu^\star))_{t=1}^\infty$. This completes the proof

\end{proof}

\subsection{Proof of Theorem~\ref{theorem:two-sided-cs-mean-so-far}}\label{proof:two-sided-cs-mean-so-far}
\TwoSidedCSMeanSoFar*
\begin{proof}
The proof proceeds in three steps. First, we derive a sub-Gaussian NSM indexed by a parameter $\lambda \in \RR$. Second, we mix this NSM over $\lambda$ using the density of a Gaussian distribution, and justify why the resulting process is also an NSM. Third and finally, we apply Ville's inequality and invert the NSM to obtain $\infseq{\widetilde C^\pm}{t}{1}$.
\paragraph{Step 1: Constructing the $\lambda$-indexed NSM.} Let $\infseq Xt1$ be independent $[0, 1]$-bounded random variables with individual means given by $\EE X_t = \mu_t^\star$, and let $\infseq Zt1$ be the \NPRR{}-induced private views of $\infseq Xt1$. Define $\zeta(\mu) := r\mu + (1-r)/2$ for any $\mu \in [0, 1]$, and $r \in (0, 1]$.
Let $\lambda \in \RR$ and consider the process,
\begin{equation}
    \label{eq:pre-mix-hoeffding-nsm}
    M_t(\lambda) := \prod_{i=1}^t \exp \left \{ \lambda(Z_i - \zeta(\mu_i^\star)) - \lambda^2 / 8 \right \},
\end{equation}
with $M_0(\lambda) \equiv 0$.
We claim that \eqref{eq:pre-mix-hoeffding-nsm} forms an NSM with respect to the private filtration $\Zcal$. The proof technique is nearly identical to that of Theorem~\ref{theorem:ldp-hoeffding-cs} but with changing means and $\lambda = \lambda_1 = \lambda_2 = \cdots \in \RR$. Indeed, $M_t(\lambda)$ is nonnegative with initial value one by construction, so it remains to show that $(M_t(\lambda))_{t=0}^\infty$ is a supermartingale. That is, we need to show that for every $t$, we have $\EE(M_t(\lambda) \mid \Zcal_{t-1}) \leq M_{t-1}(\lambda)$. Writing out the conditional expectation of $M_t(\lambda)$, we have
\begin{align*}
    \EE(M_t(\lambda) \mid \Zcal_{t-1}) &= \EE\left ( \prod_{i=1}^t \exp \left \{ \lambda(Z_i - \zeta(\mu_i^\star)) - \lambda^2 / 8 \right \} \Bigm \vert Z_1^{t-1} \right) \\
    &= \underbrace{\prod_{i=1}^{t-1} \exp \left \{ \lambda(Z_i - \zeta(\mu_i^\star)) - \lambda^2 / 8 \right \}}_{M_{t-1}(\lambda)} \cdot \EE \left (\exp \left \{ \lambda(Z_t - \zeta(\mu_t^\star)) - \lambda^2 / 8 \right \} \mid Z_1^{t-1} \right ) \\
    &= M_{t-1}(\lambda) \cdot \underbrace{\EE \left (\exp \left \{ \lambda(Z_t - \zeta(\mu_t^\star)) - \lambda^2 / 8 \right \} \right )}_{(\dagger)},
\end{align*}
where the last inequality follows by independence of $\infseq{Z}{t}{1}$, and hence the conditional expectation becomes a marginal expectation.
Therefore, it now suffices to show that $(\dagger) \leq 1$. Indeed, $Z_t$ is a $[0, 1]$-bounded, mean-$\zeta(\mu_t^\star)$ random variable. By Hoeffding's sub-Gaussian inequality for bounded random variables \citep{hoeffding1963probability}, we have that ${\EE[\exp \{ \lambda (Z_t - \zeta(\mu_t^\star)) \}]\leq \exp \{ \lambda^2 / 8 \}}$, and thus
\[ (\dagger) = \EE \left [\exp \left \{ \lambda(Z_t - \zeta(\mu_t^\star)) \right \} \right ]\cdot \exp \left \{- \lambda^2 / 8 \right \} \leq 1.
\]
It follows that $(M_t(\lambda))_{t=0}^\infty$ is an NSM.

\paragraph{Step 2.} Let us now construct a sub-Gaussian mixture NSM. Note that the mixture of an NSM with respect to a probability distribution is itself an NSM \citep{robbins1970statistical,howard2020time} --- a straightforward consequence of Fubini's theorem. Concretely, let $f_{\rho^2}(\lambda)$ be the probability density function of a mean-zero Gaussian random variable with variance $\rho^2$,
\[ f_{\rho^2} (\lambda) := \frac{1}{\sqrt{2\pi \rho^2}} \exp \left \{ \frac{-\lambda^2}{2\rho^2} \right \}.\]
Then, since mixtures of NSMs are themselves NSMs, the process $\infseq Mt0$ given by
\begin{equation}
\label{eq:mixture-implicit}
    M_t := \int_{\lambda \in \RR} M_t(\lambda) f_{\rho^2}(\lambda) d\lambda
\end{equation}
is an NSM. We will now find a closed-form expression for $M_t$. To ease notation, define the partial sum ${S^\star_t := \sum_{i=1}^t (Z_i - \zeta(\mu_i^\star))}$. Writing out the definition of $M_t$, we have

\begin{align*}
    M_t &:= \int_{\lambda \in \mathbb R} \prod_{i=1}^t \exp\left \{ \lambda (Z_i - \zeta(\mu_i^\star)) - \lambda^2/8 \right \}f_{\rho^2}(\lambda) d\lambda \\
    &= \int_{\lambda} \exp\left \{ \lambda \underbrace{\sum_{i=1}^t (Z_i - \zeta(\mu_i^\star))}_{S^\star_t} - t\lambda^2/8 \right \}f_{\rho^2}(\lambda) d\lambda \\
    &= \int_{\lambda} \exp\left \{ \lambda S^\star_t - t\lambda^2/8 \right \} \frac{1}{\sqrt{2\pi \rho^2}} \exp \left \{ \frac{-\lambda^2}{2\rho^2} \right \}d\lambda \\
    &= \frac{1}{\sqrt{2\pi \rho^2}}\int_{\lambda} \exp\left \{ \lambda S^\star_t - t\lambda^2/8 \right \} \exp \left \{ \frac{-\lambda^2}{2\rho^2} \right \}d\lambda \\
    &= \frac{1}{\sqrt{2\pi \rho^2}} \int_\lambda \exp \left \{ \lambda S^\star_t - \frac{\lambda^2 (t\rho^2/4 + 1)}{2 \rho^2}\right \} d\lambda \\
    &= \frac{1}{\sqrt{2\pi \rho^2}} \int_\lambda \exp \left \{ \frac{-\lambda^2 (t\rho^2/4 + 1) + 2\lambda \rho^2 S^\star_t }{2\rho^2} \right \} d\lambda \\
    &= \frac{1}{\sqrt{2\pi \rho^2}} \int_\lambda \exp \left \{ \frac{-a(\lambda^2 - \frac{b}{a} 2\lambda) }{2\rho^2} \right \} d\lambda,
\end{align*}
where we have set $a:= t\rho^2/4 + 1$ and $b := \rho^2 S^\star_t$. Completing the square in the exponent, we have that
\begin{align*}
    \exp \left \{ \frac{-\lambda^2 - 2\lambda \frac{b}{a} + \left ( \frac{b}{a} \right )^2 - \left ( \frac{b}{a} \right )^2 }{2 \rho^2 /a} \right \} &= \exp \left \{ \frac{-(\lambda - b/a)^2}{2\rho^2/a} + \frac{a \left ( b/a \right )^2}{2\rho^2} \right \} \\
    &= \underbrace{\exp \left \{ \frac{-(\lambda - b/a)^2}{2\rho^2/a} \right \}}_{(\star)} \exp \left \{  \frac{b^2}{2a\rho^2} \right \}.
\end{align*}
Now notice that $(\star)$ is proportional to the density of a Gaussian random variable with mean $b/a$ and variance $\rho^2/a$.
Plugging the above back into the integral and multiplying the entire quantity by $a^{-1/2}/a^{-1/2}$, we obtain the closed-form expression of the mixture NSM,
\begin{align}
    M_t &:= \underbrace{\frac{1}{\sqrt{2\pi \rho^2/a}} \int_{\lambda \in \mathbb R} \exp\left \{ \frac{-(\lambda - b/a)^2}{2 \rho^2/a} \right \} d\lambda}_{=1} \frac{\exp \left \{ \frac{b^2}{2a\rho^2} \right \} }{\sqrt{a}} \nonumber \\
    &=\frac{1}{\sqrt{t\rho^2/4 + 1}}\exp \left \{ \frac{\rho^2 (S_t^\star)^2}{2(t\rho^2/4 + 1)} \right \}.
    \label{eq:two-sided-mixture-NSM}
\end{align}

\paragraph{Step 3.} Now that we have computed the mixture NSM $\infseq Mt0$, we are ready to apply Ville's inequality and invert the process. Since $\infseq Mt0$ is an NSM, we have by Ville's inequality \citep{ville1939etude},
\[ \PP(\exists t : M_t \geq 1/\alpha) \leq \alpha ~~~\text{or equivalently,}~~~ \PP(\forall t, \ M_t < 1/\alpha) \geq 1-\alpha. \]
Therefore, with probability at least $(1-\alpha)$, we have that for all $t \in \{1, 2, \dots\}$,
\begin{align*}
    M_t < 1/\alpha &\iff \frac{1}{\sqrt{t\rho^2/4 + 1}}\exp \left \{ \frac{\rho^2 (S_t^\star)^2}{2(t\rho^2/4 + 1)} \right \} < 1/\alpha \\
    &\iff \frac{\rho^2 (S_t^\star)^2}{2(t\rho^2/4 + 1)}- \log\left ( \sqrt{t\rho^2/4 + 1}\right)  < \log(1/\alpha ) \\
    &\iff \frac{\rho^2 (S_t^\star)^2}{2(t\rho^2/4 + 1)}  < \log\left ( \frac{\sqrt{t\rho^2/4 + 1} }{\alpha} \right) \\
    &\iff (S_t^\star)^2 < \frac{2(t\rho^2/4 + 1) }{\rho^2}\log\left ( \frac{\sqrt{t\rho^2/4 + 1} }{\alpha} \right) \\
    &\iff \frac{(S_t^\star)^2}{t^2r^2} < \underbrace{\frac{2(t(\rho/2)^2 + 1) }{(tr\rho)^2}\log\left ( \frac{\sqrt{t(\rho/2)^2 + 1} }{\alpha} \right)}_{(\star \star)}.
\end{align*}
Set $\beta := \rho/2$ and notice that $(\star \star) = (\widetilde B_t^\pm)^2$ where $\widetilde B_t^\pm$ is the boundary given by~\eqref{eq:two-sided-changing-means-boundary} in the statement of Theorem~\ref{theorem:two-sided-cs-mean-so-far}. Also recall from Theorem~\ref{theorem:two-sided-cs-mean-so-far} the private estimator $\widehat \mu_t := \frac{1}{tr}\sum_{i=1}^t [Z_i -(1-r)/2]$ and the quantity we wish to capture --- the moving average of \emph{population means} $\widetilde \mu_t^\star := \frac{1}{t}\sum_{i=1}^t \mu_i^\star$, where $\mu_i^\star = \EE X_i$. Putting these together with the above high-probability bound, we have that with probability $\geq (1-\alpha)$, for all $t$,
\begin{align*}
    M_t < 1/\alpha &\iff \frac{(S_t^\star)^2}{t^2r^2} < (\widetilde B_t^\pm)^2\\
    &\iff -\widetilde B_t^\pm < \frac{S^\star_t}{tr} < \widetilde B_t^\pm. \\
    &\iff -\widetilde B_t^\pm < \frac{\sum_{i=1}^t [Z_i - \zeta(\mu_i^\star)]}{tr} < \widetilde B_t^\pm. \\
    &\iff -\widetilde B_t^\pm < \frac{\sum_{i=1}^t [Z_i - (r\mu_i^\star + (1-r)/2)]}{tr} < \widetilde B_t^\pm. \\
    &\iff -\widetilde B_t^\pm < \frac{\sum_{i=1}^t [Z_i - (1-r)/2]}{tr} - \frac{\bcancel{r}\sum_{i=1}^t \mu_i^\star}{t \bcancel{r}} < \widetilde B_t^\pm. \\
    &\iff - \frac{\sum_{i=1}^t [Z_i - (1-r)/2]}{tr} - \widetilde B_t^\pm  <  - \frac{\sum_{i=1}^t \mu_i^\star}{t} < - \frac{\sum_{i=1}^t [Z_i - (1-r)/2]}{tr} + \widetilde B_t^\pm. \\
    &\iff - \widehat \mu_t - \widetilde B_t^\pm  <  - \widetilde \mu_t^\star < - \widehat \mu_t + \widetilde B_t^\pm. \\
    &\iff \widehat \mu_t - \widetilde B_t^\pm <  \widetilde \mu_t^\star < \widehat \mu_t + \widetilde B_t^\pm . \\
\end{align*}
In summary, we have that $\widetilde C^\pm_t := (\widehat \mu_t \pm \widetilde B_t^\pm)$ forms a $(1-\alpha)$-\cs{} for the time-varying parameter $\widetilde \mu_t^\star$, meaning
\[ \PP\left (\forall t,\ \widetilde \mu_t^\star \in \widetilde C^\pm_t \right ) \geq 1-\alpha. \]
This completes the proof.

\end{proof}


\section{Additional results}\label{section:additional-results}

\subsection{Confidence sets under randomized response}\label{section:confidence-sets-rr}
Since \NPRR{} is a strict generalization for bounded random variables, it can be used to construct confidence sets for the mean of Bernoulli random variables which are privatized via randomized response (RR). The following corollary provides a Hoeffding-type \ci{} for the mean under RR\@.
\begin{restatable}[Locally private Hoeffding inequality under RR]{corollary}{hoeffdingBinomialCi}
  \label{corollary:hoeffding-binomial-ci}
  Let $X_1, \dots, X_n \sim \mathrm{Bernoulli}(p^\star)$, and let $Z_1,\dots, Z_n$ be their privatized views according to RR for some fixed $r \in (0, 1]$. Then,
  \begin{equation}
  \small
      \dot L_n^\Hoeff := \frac{\sum_{i=1}^n (Z_i - (1-r)/2)}{n r} - \sqrt{\frac{\log (1/\alpha)}{2 n r^2}}
  \end{equation}
  is a $(1-\alpha, \eps)$-lower \lpci{} for $p^\star$, where $\eps = \log(1 + 2r/(1-r))$.
\end{restatable}
\cref{corollary:hoeffding-binomial-ci} is a special case of Theorem~\ref{theorem:hoeffding-nprr-ci}. Notice that in the non-private setting when $r = 1$, \cref{corollary:hoeffding-binomial-ci} recovers Hoeffding's inequality exactly \citep{hoeffding1963probability}.

\subsection{Confidence sets for sample means}\label{section:confidence-sets-sample-means}
While we primarily focused on deriving \ci{}s and \cs{}s for population means, our techniques can also be applied to the construction of \ci{}s and \cs{}s for the \emph{sample mean}. Indeed, in the non-interactive case, the proof of \cref{theorem:hoeffding-nprr-ci} can be modified so that the bound~\eqref{eq:hoeffding-ci-simple-case} is a lower $(1-\alpha)$-\ci{} for the sample mean $\mu^\star := \frac{1}{n}\sum_{i=1}^n x_i$, recovering essentially the same result as \citet[Theorem 1]{ding2017collecting}.\footnote{Technically, a one-sided \ci{} is more general than \citet{ding2017collecting}'s since theirs is a two-sided \ci{} that we recover after taking a union bound over lower and upper \ci{}s, but the lower \ci{} is also implicit in their proof.} However, implicit in our results are also time-uniform \cs{}s for the \emph{running sample mean so far}. Concretely, we have the following corollary.
\begin{corollary}[A confidence sequence for the running sample mean]\label{corollary:cs-sample-mean}
  Let $\infseqt{x_t}$ be a sequence of $[0, 1]$-bounded numbers and let $\infseqt{Z_t}$ be their privatized views according to \NPRR{} without sequential interactivity. Then, the same bound as given in \cref{theorem:two-sided-cs-mean-so-far}, i.e.
  \begin{equation}
    \widetilde C_t := \left ( \frac{\sum_{i=1}^t (Z_i - (1-r) / 2)}{t r} \pm \sqrt{\frac{t\beta^2 + 1 }{2(tr\beta)^2}\log \left (\frac{\sqrt{t \beta^2 + 1}}{\alpha} \right)} \right )
  \end{equation}
  forms a $(1-\alpha, \eps)$-\lpcs{} for the running sample mean $\widetilde \mu_t^\star := \frac{1}{t}\sum_{i=1}^t x_i$, i.e.
  \begin{equation}
    \PP \left ( \forall t,\ \widetilde \mu_t^\star \in \widetilde C_t \right ) \geq 1-\alpha.
  \end{equation}
\end{corollary}
The above corollary is an immediate consequence of \cref{theorem:two-sided-cs-mean-so-far} instantiated for random variables $\infseqt{X_t}$ with degenerate distributions. (and hence $\EE X_t = X_t = \mu_t^\star$).

\cref{corollary:cs-sample-mean} also sheds some light on how the two estimands (population vs sample means) are related but fundamentally different. Both the (a) stochastic setting with data $X_1, X_2, \dots$ that have a constant mean $\EE X_1 = \mu^\star \in [0, 1]$ and (b) nonstochastic setting with deterministic data $x_1, x_2, \dots$ are special cases of the stochastic setting with data that have time-varying means $\EE X_t = \mu_t$ for $t \geq 1$. Setting (a) is recovered by assuming that $\mu_1 = \mu_2 = \cdots = \mu^\star$, while setting (b) is recovered by assuming $\infseqt{X_t}$ have degenerate distributions (or by conditioning on them). Clearly, neither is a special case of the other, and hence we cannot expect CIs/CSs for one to work for the other in general (though in this case, $\infseqt{\widetilde C_t}$ works for both).


\subsection{Why one should set \texorpdfstring{$G = 1$}{G to one} for Hoeffding-type methods}
\label{section:hoeffding-G=1}
In \cref{section:nprr-ci}, we recommended setting $G$ to the smallest possible value of $1$ because Hoeffding-type bounds cannot benefit from larger values. We will now justify mathematically where this recommendation came from.

Suppose $\seq Xt1n \sim P$ for some $\PcalNPRRstarn$ where we have chosen $r \in (0, 1]$ and an integer $G \geq 1$ to satisfy $\eps$-LDP with
\begin{equation}
\label{eq:hoeffding-G=1:eps}
    \eps := \log\left ( 1 + \frac{(G+1)r}{1-r} \right ).
\end{equation}
Recall the \NPRR{}-Hoeffding lower \lpci{} given \eqref{eq:hoeffding-ci-simple-case},
\begin{equation}
    \dot L_n^\Hoeff := \frac{\sum_{i=1}^n (Z_i - (1-r)/2)}{nr} - \underbrace{\sqrt{\frac{\log(1/\alpha)}{2nr^2}}}_{\dot B_n^\Hoeff},
\end{equation}
and take particular notice of $\dot B_n^\Hoeff$, the ``boundary''. Making this bound as sharp as possible amounts to minimizing $\dot B_n^\Hoeff$, which is clearly when $r = 1$ --- the non-private case --- but what if we want to minimize $\dot B_n^\Hoeff$ \emph{subject to} $\eps$-LDP? Given the relationship between $\eps$, $r$, and $G$, we have that $r$ can be written as
\[ r := \frac{\exp\{\eps\} - 1}{\exp\{\eps\} + G}. \]
Plugging this into $\dot B_n^\Hoeff$, we have
\[ \dot B_n^\Hoeff := \sqrt{\frac{\log(1/\alpha)}{2n \left ( \frac{\exp\eps - 1}{\exp\eps + G} \right )^2}} = \left (\frac{\exp\{\eps\} + G}{\exp\{\eps\} - 1} \right ) \cdot \sqrt{\frac{\log(1/\alpha)}{2n}},\]
which is a strictly increasing function of $G$. It follows that $G$ should be set to the minimal value of $1$ to make $\dot L_n^\Hoeff$ as sharp as possible.

\subsection{Confidence sets under the sequentially interactive Laplace mechanism}\label{section:laplace}

\begin{restatable}[\texttt{Lap-H-CS}]{proposition}{LaplaceHoeffdingCS}\label{proposition:laplace-hoeffding-cs}
Suppose $\infseq Xt1 \sim P$ for some $\PcalNPRRstarinf$ and let $\infseq Zt1$ be their privatized views according to Algorithm~\ref{algorithm:seqIntLaplace}.
Let $\psi_t^\lapNoise(\lambda) := -\log(1 - \lambda^2/\eps_t^2)$ be the (conditional) cumulant generating function of a mean-zero Laplace random variable with scale $1/\eps_t$.
Let $\infseq \lambda t1$ be a sequence of random variables such that $\lambda_t$ depends on $Z_1^{t-1}$ --- formally $\sigma(Z_1^{t-1})$-measurable --- and $[0, \eps_t)$-valued. Then,
\begin{equation}
\small
    \bar L_t^\lapNoise := \frac{\sum_{i=1}^t \lambda_i Z_i }{\sum_{i=1}^t \lambda_i } - \frac{\log(1/\alpha) + \sum_{i=1}^t \left ( \lambda_i^2/8 + \psi_i^\lapNoise(\lambda_i) \right ) }{\sum_{i=1}^t \lambda_i}
\end{equation}
forms a lower $(1-\alpha, \infseqt{\eps_t})$-\lpcs{} for $\mu^\star$.
\end{restatable}

To obtain sharp \cs{}s for $\mu^\star$, we recommend setting
\begin{equation}\label{eq:laplace-lambda-cs}
  \lambda_t := \sqrt{\frac{\log(1/\alpha)}{\sum_{i=1}^t (1/8 + 1/\eps_i^2) \log(t+1)}} \land c\cdot\eps_t,
\end{equation}
for some prespecified truncation scale $c \in (0, 1)$. We choose $\lambda_t$ as scaling like $1/\sqrt{t \log t}$ so that the \cs{} $\bar L_t^\lapNoise$ is $O(\sqrt{\log t / t})$ up to $\log \log$ factors (see \citet[Table 1]{waudby2020estimating} for more details).\footnote{This specific rate assumes $\eps_t = \eps \in (0, 1)$ for each $t$.} The constants provided in~\eqref{eq:laplace-lambda-cs} arise from approximating $\psi^\lapNoise(\lambda)$ by $\lambda^2/\eps^2$ for $\lambda$ near $0$ --- an approximation that can be justified by a simple application of L'Hopital's rule --- and attempting to minimize the \ci{} width.

Similar to Section~\ref{section:nprr-ci}, we can choose $\infseq \lambda t1$ so that $\bar L_t^\lapNoise$ is tight for a fixed sample size $n$. Indeed, we have the following Laplace-Hoeffding \ci{}s for $\mu^\star$.
\begin{corollary}[\texttt{Lap-H}]
\label{corollary:laplace-hoeffding-ci}
Given the same assumptions as Proposition~\ref{proposition:laplace-hoeffding-cs} for a fixed sample size $n$, define
\begin{equation}
\label{eq:laplace-lambda-ci}
    \lambda_{t,n} := \sqrt{\frac{\log(1/\alpha)}{\frac{n}{t}\sum_{i=1}^{t} (1/8 + 1/\eps_i^2)}} \land c\cdot\eps_t,
\end{equation}
and plug it into $\bar L_{t}^\lapNoise$ as given above. Then,
\[ \dot L_n^\lapNoise := \max_{1\leq t \leq n} \bar L_{t} \]
is a $(1-\alpha, (\eps_t)_t )$-lower \lpci{} for $\mu^\star$.
\end{corollary}
The proof of Proposition~\ref{proposition:laplace-hoeffding-cs} (and hence Corollary~\ref{corollary:laplace-hoeffding-ci}) can be found in Section~\ref{proof:laplace-hoeffding}. Note that any prespecified value of $c \in (0, 1)$ yields valid \cs{}s and \ci{}s, we find that smaller values (e.g. near $0.1$) yield tighter intervals, and we set $c = 0.1$ in our simulations (Figures~\ref{fig:ci} and~\ref{fig:cs}). 

\subsection{Variance-adaptive confidence intervals and sequences}\label{section:variance-adaptive}

\subsubsection{Variance-adaptive confidence intervals}\label{section:variance-adaptive-ci}
Notice that if $G_t = 1$ for each $t$, then regardless of how low-variance $\seq Xt1n$ are, the observations that are ultimately used for confidence set construction are still Bernoulli. In other words, it does not matter whether $\seq Xt1n$ are Bernoulli(1/2), Uniform[0, 1], or Beta(100, 100) --- with variances of roughly 0.25, 0.083, and 0.0012, respectively --- the privatized observations $\seq Zt1n$ are all Bernoulli(1/2) with a maximal variance 0.25. Unfortunately, this means that variance-adaptive techniques cannot be used to derive tighter \ci{}s from $\seq Zt1n$ directly. The story changes, however, when $G_t \geq 2$. Concretely, for the same value of $r_t$, setting $G_t$ to be very large does not change the conditional mean of $Z_t$ but it can substantially lower its conditional variance (e.g. if $X_t$ has a continuous distribution, such as Beta($\alpha, \beta$)). Of course, given the fact that \NPRR{} satisfies $\eps_t$-LDP with $\eps_t = \log \left ( 1 + \frac{(G_t+1)r_t}{1-r_t} \right )$, there are privacy implications to increasing $G_t$, and hence there is a tradeoff that must be carefully navigated when choosing $(r_t, G_t)$ to satisfy $\eps_t$ when attempting to derive variance-adaptive \ci{}s. We will leave that delicate discussion for later --- for now, it is just important to keep in mind that larger $G_t$ can lower the variance of $\seq Zt1n$, and our goal will be to exploit this fact for the sake of tighter \ci{}s.

We will proceed by turning to the literature on nonasymptotic \ci{}s for bounded random variables, focusing on the (super)martingale-based \ci{}s of~\citet{waudby2020estimating} and adapting their techniques to the locally private setting. Specifically, we will derive private analogues of the product martingales outlined in \citet[Section 4]{waudby2020estimating} as well as the so-called ``predictable plug-in'' supermartingales of \citet[Section 3]{waudby2020estimating}. As we will see in \cref{theorem:ldp-hedged-ci} and \cref{proposition:ldp-eb-ci}, the former product martingales yield tighter \ci{}s but at the expense of a closed-form expression, while the latter supermartingales are looser (but still variance-adaptive) and are available in closed-form.

\paragraph{Product ``betting'' martingales.} Beginning with the former, we follow the discussions in \citet[Remark 1 \& Section 5.1]{waudby2020estimating} and set
\begin{align}
  \label{eq:lambda-betting-ci}
  &\lambda_{t,n}(\mu) := \sqrt{\frac{2 \log(1/\alpha)}{\widehat \gamma_{t-1}^2 n}} \land \frac{c}{\zeta_t(\mu)},~\text{where}\\
  &\widehat \gamma_t^2 := \frac{1/4 + \sum_{i=1}^t (Z_i - \widehat \zeta_i)^2}{t + 1},\ \widehat \zeta_t := \frac{1/2 + \sum_{i=1}^t Z_i}{t + 1},\nonumber
\end{align}
and $c \in (0, 1)$ is some prespecified truncation scale (e.g. 1/2 or 3/4). Given the above, we have the following variance-adaptive \ci{} for $\mu^\star$ under \NPRR{}.

\begin{restatable}[\texttt{NPRR-hedged}]{theorem}{LDPHedgedCI}\label{theorem:ldp-hedged-ci}
  Suppose $\seq{X}{t}{1}{n} \sim P$ for some $P \in \Pcal_{\mu^\star}^n$ and let $\seq{Z}{t}{1}{n} \sim Q$ be their \NPRR{}-privatized views where $Q \in \Qcal_{\mu^\star}^n$. Define
  \begin{equation}
    \label{eq:pmeb-bet-process}
    \Kcal_{t,n}(\mu) := \prod_{i=1}^t \left [1 + \lambda_{i,n}(\mu) \cdot (Z_i - \zeta_i(\mu)) \right ]
  \end{equation}
  with $\lambda_{t,n}(\mu)$ given by \eqref{eq:lambda-betting-ci}.
  Then, $\Kcal_{t,n}(\mu)$ is a nonincreasing function of $\mu \in [0, 1]$, and $\Kcal_{t,n}(\mu^\star)$ forms a $\Qcal_{\mu^\star}^n$-NM. Consequently,
\begin{equation}\label{eq:hedged-CI}
  \dot L_n := \max_{1 \leq t \leq n} \inf \left \{ \mu \in [0, 1] : \Kcal_{t,n}(\mu) < 1/\alpha \right \}
\end{equation}
    forms a lower $(1-\alpha, (\eps_t)_t)$-\lpci{} for $\mu^\star$, meaning ${\PP(\mu^\star \geq \dot L_n) \geq 1-\alpha}$.
\end{restatable}
The proof in Section~\ref{proof:ldp-hedged-ci} follows a similar technique to that of Theorem~\ref{theorem:ldp-dkelly}. As is apparent in the proof, $\Kcal_{t,n}(\mu^\star)$ forms a $\QcalNPRRstarn$-NM regardless of how $\lambda_{t,n}(\mu)$ is chosen, in which case the resulting $\dot L_n$ would still be a bona fide lower confidence bound. However, the choice of $\lambda_{t,n}(\mu)$ given in \eqref{eq:lambda-betting-ci} provides excellent empirical performance for the reasons discussed in \citep[Section 5.1]{waudby2020estimating} and guarantees that $\dot L_n$ is an interval (rather than a union of disjoint sets, for example).
We find that Theorem~\ref{theorem:ldp-hedged-ci} has the best empirical performance out of the private \ci{}s in our paper (see Figure~\ref{fig:ci}).
In our simulations (Figure~\ref{fig:ci}), we set $c = 0.8$, and $(r, G)$ were chosen using the technique outlined in Section~\ref{section:choosing-rG}.

\paragraph{Empirical Bernstein supermartingales.} While~\cref{theorem:ldp-hedged-ci} improves on~\cref{theorem:hoeffding-nprr-ci} in terms of variance-adaptivity, the resulting bounds given in~\eqref{eq:hedged-CI} are \emph{implicit}, and hence require numerical methods (e.g.~root-finding algorithms) to compute the downstream \ci{}. The numerical operations required are both computationally efficient and straightforward to implement in code, but closed-form bounds may nevertheless be preferable for the sake of simplicity. Empirical Bernstein \ci{}s occupy a middle ground between the Hoeffding-style \ci{}s of~\cref{theorem:ldp-hoeffding-cs} and the implicit \ci{}s of~\cref{theorem:ldp-hedged-ci} by being both closed-form and variance-adaptive.
To this end, consider the following tuning parameters which are similar (but not identical) to~\eqref{eq:lambda-betting-ci}:
\begin{align}
  \label{eq:lambda-eb-ci}
  &\lambda_{t,n}^\PMEB(\mu) := \sqrt{\frac{2 \log(1/\alpha)}{\widehat \gamma_{t-1}^2 n}} \land c,~\text{where}\\
  &\widehat \gamma_t^2 := \frac{1/4 + \sum_{i=1}^t (Z_i - \widehat \zeta_i)^2}{t + 1},\ \widehat \zeta_t := \frac{1/2 + \sum_{i=1}^t Z_i}{t + 1},\nonumber
\end{align}
and $c \in (0, 1)$. Then, we have the following variance-adaptive empirical Bernstein \ci{}s under \NPRR{}.

\begin{restatable}[\texttt{NPRR-EB}]{proposition}{NprrEbCi}\label{proposition:ldp-eb-ci}
  Under the same assumptions as Theorem~\ref{theorem:ldp-hedged-ci}, let $(\lambda_{t,n}^\EB)_{t=1}^n$ be the $[0, 1)$-valued $\Zcal$-predictable sequence given in~\eqref{eq:lambda-eb-ci} and define
  \begin{align*}
    \widehat \mu_t(\lambda_1^t) &:= \frac{\sum_{i=1}^t \lambda_i \cdot \left ( Z_i - (1-r_i)/2 \right )}{\sum_{i=1}^t r_i \lambda_i},\\
    \bar B_t^\PMEB(\lambda_1^t) &:= \frac{\log(1/\alpha) + \sum_{i=1}^t 4(Z_i - \widehat \zeta_{i-1})^2\psi_E(\lambda_i)}{\sum_{i=1}^t r_i \lambda_i}.
  \end{align*}
  where $\psi_E(\lambda) := (-\log(1-\lambda) - \lambda)/4$. Then,
  \begin{equation}
    \dot L_t^\PMEB := \max_{1 \leq t \leq n} \left \{ \widehat \mu_t - \bar B_t^\PMEB \right \}
  \end{equation}
  forms a lower $(1-\alpha, (\eps_t)_t)$-\lpci{} for $\mu^\star$, meaning $\PP(\mu^\star \geq \dot L_t^\PMEB ) \geq 1-\alpha$.
\end{restatable}

\cref{proposition:ldp-eb-ci} is a corollary of \cref{proposition:ldp-eb-cs} whose proof can be found in \cref{proof:ldp-eb-cs}.
Similar to~\cref{theorem:ldp-hedged-ci}, one can use any $(\lambda_{t,n}^\EB)_{t=1}^n$ as long as they are predictable and $[0, 1)$-valued, but we presented~\eqref{eq:lambda-eb-ci} as it tends to exhibit good empirical performance for the reasons discussed in~\citet{waudby2020estimating}. As previously alluded to, the essential difference between~\cref{theorem:ldp-hedged-ci} and~\cref{proposition:ldp-eb-ci} is that the former tends to be tighter in practice, while the latter has the advantage of having a computationally and analytically simple closed-form expression. 
In principle, the proof and techniques of~\cref{theorem:ldp-hedged-ci} and ~\cref{proposition:ldp-eb-ci} may be adapted to many other variance-adaptive \ci{}s for bounded random variables, including~\citet{bentkus2004hoeffding},~\citet{audibert2007tuning},~\citet{maurer2009empirical}, \citet{orabona2021tight}, or other bounds in~\citet{waudby2020estimating}, but we presented the aforementioned two for simplicity and illustration. Let us now turn our attention to a more challenging but related problem of constructing time-uniform \emph{confidence sequences} instead of fixed-time \emph{confidence intervals}.

\subsubsection{Variance-adaptive time-uniform confidence sequences}
\label{section:variance-adaptive-cs}
In Section~\ref{section:nprr-ci}, we presented Hoeffding-type \cs{}s for $\mu^\star$ under \NPRR{}. As discussed in Section~\ref{section:variance-adaptive-ci}, Hoeffding-type inequalities are not variance-adaptive. In this section, we will derive a simple-to-compute, variance-adaptive \cs{} at the expense of a closed-form expression.
Adapting the so-called ``grid Kelly capital process'' (GridKelly) of \citet[Section 5.6]{waudby2020estimating} to the locally private setting, consider the family of processes for each $\mu \in [0,1]$, and for any user-chosen integer $D \geq 2$,
\begin{align*}
  \Kcal_t^+(\mu) &:= \sum_{d=1}^D\prod_{i=1}^t \left [ 1 + \lambda_{i,d}^+\cdot (Z_i - \zeta_i(\mu))\right ], \\
  \text{and}~~ \Kcal_t^-(\mu) &:= \sum_{d=1}^D\prod_{i=1}^t \left [ 1 - \lambda_{i,d}^-\cdot (Z_i - \zeta_i(\mu))\right ],
\end{align*}
where $\lambda_{i,d}^+ := \frac{d}{(D+1)\zeta_i(\mu)}$ and $\lambda_{i, d}^- := \frac{d}{(D+1)(1-\zeta_i(\mu))}$ for each $i$.
Then we have the following locally private \cs{}s for $\mu^\star$.

\begin{restatable}[\texttt{NPRR-GK-CS}]{theorem}{LDPDKelly}\label{theorem:ldp-dkelly}
    Let $\infseq{Z}{t}{1} \sim Q$ for some $Q \in \Qcal_{\mu^\star}^\infty$ be the output of \NPRR{} as described in \cref{section:nprr}.
    For any prespecified $\theta \in [0, 1]$, define the process $(\Kcal_t^\GK(\mu))_{t=0}^\infty$ given by
    \[ \Kcal_t^\GK(\mu) := \theta \Kcal_t^+(\mu) + (1-\theta) \Kcal_t^-(\mu), \]
    with $\Kcal_0^\GK(\mu) \equiv 1$. Then, $\Kcal_t^\GK(\mu^\star)$ forms a $\QcalNPRRstarinf$-NM, and
    \[ \bar C_t^\GK := \left \{ \mu \in [0, 1] : \Kcal_t^\GK(\mu) < \frac{1}{\alpha} \right \} \]
    forms a $(1-\alpha, (\eps_t)_t)$-\lpcs{} for $\mu^\star$, meaning ${\PP(\forall t,\ \mu^\star \in \bar C_t^\GK) \geq 1- \alpha}$. Moreover, $\bar C_t^\GK$ forms an interval almost surely.
\end{restatable}
The proof of \cref{theorem:ldp-dkelly} is given in Section~\ref{proof:ldp-dkelly} and follows from Ville's inequality for nonnegative supermartingales \citep{ville1939etude,howard2020time}. If a lower or upper \cs{} is desired, one can set $\theta = 1$ or $\theta = 0$, respectively, with $\theta = 1/2$ yielding a two-sided \cs{}.
In our simulations (Figure~\ref{fig:cs}), we set $D = 30$, and $(r, G)$ were chosen using the technique outlined in Section~\ref{section:choosing-rG}.

In~\cref{proposition:ldp-eb-ci}, we presented a closed-form empirical Bernstein \ci{} for $\mu^\star$ under \NPRR{}. Similar to the relationship between the fixed-time NPRR-Hoeffding \ci{} (\cref{theorem:ldp-hoeffding-cs}) and the time-uniform NPRR-Hoeffding \cs{} (\cref{theorem:hoeffding-nprr-ci}),~\cref{proposition:ldp-eb-ci} is a corollary of a more general closed-form empirical Bernstein \emph{\cs{}} instantiated at a fixed sample size. We omitted this \cs{} from the main discussion for brevity, but provide its details here.

\begin{restatable}[\texttt{NPRR-EB-CS}]{proposition}{NprrEbCs}\label{proposition:ldp-eb-cs}
  Given $\infseqt{Z_t} \sim \QcalNPRRstarinf$ and let $\widehat \mu_t(\lambda_1^t)$ and $\widebar B_t^\EB(\lambda_1^t)$ be as in~\cref{proposition:ldp-eb-ci}:
  \begin{align*}
    \widehat \mu_t(\lambda_1^t) &:= \frac{\sum_{i=1}^t \lambda_i \cdot \left ( Z_i - (1-r_i)/2 \right )}{\sum_{i=1}^t r_i \lambda_i},~~\text{and}\\
    \bar B_t^\PMEB(\lambda_1^t) &:= \frac{\log(1/\alpha) + \sum_{i=1}^t 4(Z_i - \widehat \zeta_{i-1})^2\psi_E(\lambda_i)}{\sum_{i=1}^t r_i \lambda_i}.
  \end{align*}
  where $\psi_E(\lambda) := (-\log(1-\lambda) - \lambda)/4$. Then,
  \begin{equation}
    \widebar L_t^\EB := \widehat \mu_t(\lambda_1^t) - \widebar B_t^\EB(\lambda_1^t)
  \end{equation}
  forms a lower $(1-\alpha, (\eps_t)_t)$-\lpcs{} for $\mu^\star$, meaning $\PP(\forall t\geq 1,\ \mu^\star \geq \widebar L_t^\EB ) \geq 1-\alpha$.
\end{restatable}
The proof can be found in~\cref{proof:ldp-eb-cs}, and combines the techniques for deriving private concentration inequalities (such as in~\cref{theorem:ldp-hoeffding-cs}) with those for deriving predictable plug-in empirical Bernstein inequalities (such as in~\citet[Theorem 2]{waudby2020estimating}).

Similar to \cref{theorem:ldp-hedged-ci} and \cref{proposition:ldp-eb-ci}, the proofs and techniques of \cref{theorem:ldp-dkelly} and \cref{proposition:ldp-eb-cs} could potentially be adapted to many other variance-adaptive \cs{}s for bounded random variables, including other bounds contained in \citet{waudby2020estimating}, \citet{kuchibhotla2021near}, or \citet{orabona2021tight}.

\subsection{Choosing \texorpdfstring{$(r, G)$}{r and G} for variance-adaptive confidence sets}\label{section:choosing-rG}
Unlike Hoeffding-type bounds, it is not immediately clear how we should choose $(r,G)$ to satisfy $\eps$-LDP and obtain sharp confidence sets using Theorems~\ref{theorem:ldp-dkelly} and~\ref{theorem:ldp-hedged-ci}, since there is no closed form bound to optimize. Nevertheless, certain heuristic calculations can be performed to choose $(r, G)$ in a principled way.\footnote{Note that ``heuristics'' do not invalidate the method --- no matter what $(r, G)$ are chosen to be, $\eps$-LDP and confidence set coverage are preserved. We are just using heuristic to choose $(r, G)$ in a smart way for the sake of gaining power.}

One approach is to view the raw-to-private data mapping $X \mapsto Z$ as a channel over which information is lost, and we would like to choose the mapping so that as much information is preserved as possible. We will aim to measure ``information lost'' by the conditional entropy $H(Z \mid X)$ and minimize a surrogate of this value.

For the sake of illustration, suppose that $X$ has a continuous uniform distribution. This is a reasonable starting point because it captures the essence of preserving information about a continuous random variable $X$ on a discretely supported output space $\Gcal := \{0,1/G, \dots, G/G\}$. Then, the entropy $H(Z \mid X = x)$ conditioned on $X = x$ is given by
\begin{equation}
\label{eq:cond-entropy}
    H(Z \mid X=x) := \sum_{z \in \Gcal} \PP(Z=z \mid X=x) \log_2 \PP(Z=z \mid X=x),
\end{equation}
and we know that by definition of \NPRR{}, the conditional probability mass function of $(Z \mid X)$ is
\[ \PP(Z= z \mid X=x) = \frac{1-r}{G+1} + rG \cdot \left [\1(z= x^\ceil  )(x-x^\floor) + \1(z = x^\floor)(x^\ceil - x) \right ].\]
We will use the heuristic approximation $x - x^\floor \approx x^\ceil - x \approx 1/(2G)$, which would hold with equality if $x$ were at the midpoint between $x^\floor$ and $x^\ceil$.
With this approximation in mind, we can write
\begin{align}
    \PP(Z = z \mid X = x) &\approx  \frac{1-r}{G+1} + rG \cdot \left [ \frac{1}{2G}\1(z= x^\ceil \text{ or } z = x^\floor ) \right ] \nonumber\\
    &= \frac{1-r}{G+1} + \frac{r}{2} \1 (z = x^\ceil \text{ or } z = x^\floor) \label{eq:pmf-approx}
\end{align}
Given \eqref{eq:pmf-approx}, we can heuristically compute $H(Z \mid X=x)$ because for exactly two terms in the sum $\sum_{z \in \Gcal} \PP(Z=z\mid X=x) \log_2 \PP(Z= z \mid X=x)$, we will have $\1(z = x^\ceil \text{ or } z = x^\floor) = 1$ and the other $G-1$ terms will have the indicator set to 0.
Simplifying notation slightly, let $p_1(r, G) := (1-r)/(G+1) + r/2$ be \eqref{eq:pmf-approx} for those whose indicator is 1, and $p_0(r, G) := (1-r)/(G+1)$ for those whose indicator is 0. Therefore, we can write
\begin{equation}
\label{eq:cond-entropy-x-approx}
    H(Z \mid X = x) \approx (G-1) p_0(r, G) \log_2 p_0(r, G) + 2 p_1(r, G) \log_2 p_1(r,G).
\end{equation}
Finally, the conditional entropy $H(Z\mid X)$ can be approximated by
\begin{equation}
\label{eq:cond-entropy-approx}
    H(Z \mid X) = \int_{0}^1 H(Z \mid X=x) dx \approx  (G-1) p_0(r, G) \log_2 p_0(r, G) + 2 p_1(r, G) \log_2 p_1(r,G),
\end{equation}
since we assumed that $X$ was uniform on [0, 1].

Given a fixed privacy level $\eps\in (0, \infty)$, the approximation \eqref{eq:cond-entropy-approx} gives us an objective function to minimize with respect to $r$ (since $G$ is completely determined by $r$ once $\eps$ is fixed). This can be done using standard numerical minimization solvers. Once an optimal $(r_\opt, \widetilde G_\opt)$ pair is determined numerically, $\widetilde G_\opt$ may not be an integer (but we require $G \geq 1$ to be an integer for \NPRR{}). As such, one can then choose the final $G_\opt$ to be $\lfloor \widetilde G_\opt \rfloor$ or $\lceil \widetilde G_\opt \rceil$, depending on which one minimizes $H(Z\mid X)$ while keeping $\eps$ fixed. If the numerically determined $\widetilde G_\opt$ is $\leq 1$, then one can simply set $G_\opt := 1$ and adjust $r_\opt$ accordingly.

\subsection{One-sided time-varying}\label{section:wavy-one-sided}
The following one-sided analogue of \cref{theorem:two-sided-cs-mean-so-far} can be derived via slightly different techniques; the details can be found in its proof.

\begin{restatable}[]{proposition}{OneSidedCSMeanSoFar}\label{proposition:one-sided-cs-mean-so-far}
  Given the same setup as Theorem~\ref{theorem:two-sided-cs-mean-so-far}, define
  {\small\begin{align}
           \widetilde B_t &:= \sqrt{\frac{t\beta^2 + 1}{2(tr\beta)^2}\log \left (1+\frac{\sqrt{t \beta^2 + 1}}{2\alpha} \right)}. \label{eq:one-sided-changing-means-boundary}
         \end{align}}%
       Then, $\widetilde L_t := \widehat \mu_t - \widetilde B_t$ forms a lower $(1-\alpha, \eps)$-\lpcs{}
       for $\widetilde \mu_t^\star := \frac{1}{t} \sum_{i=1}^t \mu_i^\star$, meaning
       \begin{equation}
         \small
         \label{eq:one-sided-cs-mean-so-far}
         \PP\left (\forall t,\ \widetilde \mu_t^\star \geq \widetilde L_t \right ) \geq 1-\alpha.
       \end{equation}
\end{restatable}
The proof is provided in Section~\ref{proof:one-sided-cs-mean-so-far} and uses a one-sided sub-Gaussian mixture supermartingale technique similar to \citet[Proposition 6]{howard2021time}. Since $\widetilde B_t$ resembles $\widetilde B_t^\pm$ but with $\alpha$ doubled, we suggest choosing $\beta$ using \eqref{eq:rho-opt} but with $\beta_{2\alpha}(t_0)$. We display $\widetilde L_t$ alongside the two-sided bound $\widetilde C_t^\pm$ of Theorem~\ref{theorem:two-sided-cs-mean-so-far} in Figure~\ref{fig:wavy-cs}.

\subsection{Private hypothesis testing and \texorpdfstring{$p$}{p}-values}
\label{section:testing}
So far, we have focused on the use of \emph{confidence sets} for statistical inference, but another closely related perspective is through the lens of hypothesis testing and $p$-values (and their sequential counterparts). Fortunately, we do not need any additional techniques to derive methods for testing, since they are byproducts of our previous results.

Following the nonparametric conditions\footnote{The discussion that follows also applies to the parametric case.} outlined in \cref{section:nprr-ci}, suppose that $\seq Xt1\infty \sim P$ for some $P\in \Pcal_{\mu^\star}^\infty$ which are then privatized into $\seq Zt1\infty \sim Q \in \Qcal_{\mu^\star}^\infty$ via \NPRR{}. The goal now --- ``locally private sequential testing'' --- is to use the private data $\seq Zt1\infty$ to test some null hypothesis $\Hcal_0$. For example, to test $\mu^\star = \mu_0$, we set $\Hcal_0 = \Qcal_{\mu_0}^\infty$ or to test $\mu^\star \leq \mu_0$, we set ${\Hcal_0 = \{ Q \in \Qcal_{\mu}^\infty : \mu \leq \mu_0 \}}$.

\sloppy Concretely, we are tasked with designing a binary-valued function $\bar \phi_t \equiv {\bar \phi(Z_1, \dots, Z_t) \to \{0, 1\}}$ with outputs of 1 and 0 being interpreted as ``reject $\Hcal_0$'' and ``fail to reject $\Hcal_0$'', respectively, so that
\begin{equation}
\small
  \label{eq:sequential-test}
  \sup_{Q \in \Hcal_0} Q(\exists t : \bar \phi_t = 1) \leq \alpha.
\end{equation}
A sequence of functions $\infseq {\bar \phi} t1$ satisfying \eqref{eq:sequential-test} is known as a \emph{level-$\alpha$ sequential test}. Another common tool in hypothesis testing is the $p$-value, which also has a sequential counterpart, known as the \emph{anytime $p$-value} \citep{johari2017peeking,howard2021time}. We say that a sequence of $p$-values $\infseq {\bar p}t1$ is an \emph{anytime $p$-value} if
\begin{equation}
  \label{eq:anytime-p-val}
  \small
  \sup_{Q \in \Hcal_0} Q(\exists t : \bar p_t \leq \alpha) \leq \alpha.
\end{equation}

There are at least two ways to achieve \eqref{eq:sequential-test} and \eqref{eq:anytime-p-val}: (a) by using \cs{}s to reject non-intersecting null hypotheses, and (b) by explicitly deriving $e$-processes. We will first discuss (a) and leave (b) to \cref{section:e-processes} as the discussion is more involved.

\subsubsection{Private hypothesis testing using confidence sets.}
The simplest and most direct way to test hypotheses using the results of this paper is to exploit the duality between \cs{}s and sequential tests (or \ci{}s and fixed-time tests). Suppose $(\bar C_t(\alpha))_{t=1}^\infty$ is an LDP $(1-\alpha)$-\cs{} for $\mu^\star$, and let $\Hcal_0: \{ Q\in \Qcal_{\mu}^\infty : \mu \in \Theta_0 \}$ be a null hypothesis that we wish to test. Then, for any $\alpha \in (0, 1)$,
\begin{equation}
  \small
  \bar \phi_t := \1\left (\bar C_t(\alpha) \cap \Theta_0 = \emptyset \right)
\end{equation}
forms an LDP level-$\alpha$ sequential test for $\Hcal_0$, meaning it satisfies \eqref{eq:sequential-test}.
In particular, if $\bar C_t(\alpha)$ shrinks to a single point as $t\to \infty$, then $\infseq {\bar \phi}t1$ has asymptotic power one.
Furthermore, $\inf\{\alpha : \bar C_t(\alpha) \cap \Theta_0 = \emptyset\}$ forms an anytime $p$-value for $\Hcal_0$, meaning it satisfies \eqref{eq:anytime-p-val}.

Similarly, if $\dot C_n(\alpha)$ is a $(1-\alpha)$ \ci{} for $\mu^\star$, then $\dot \phi_n := \1(\dot C_n(\alpha) \cap \Theta_0 = \emptyset)$ is a level-$\alpha$ test: $\sup_{Q \in \Hcal_0}Q(\dot \phi_n = 1) \leq \alpha$, and $\dot p_n := \inf\{\alpha: \dot C_n(\alpha) \cap \Theta_0 = \emptyset\}$ is a $p$-value for $\Hcal_0$: $\sup_{Q \in \Hcal} Q(\dot p_n \leq \alpha) \leq \alpha$.

  One can also derive sequential tests using so-called \emph{$e$-processes} --- processes that are upper-bounded by nonnegative supermartingales under a given null hypothesis. In fact, every single one of our \cs{}s is derived by first deriving an explicit $e$-process. Let us now discuss how one can derive sequential tests and \cs{}s using $e$-processes.

\subsubsection{Testing via $e$-processes}\label{section:e-processes}

To achieve \eqref{eq:sequential-test} and \eqref{eq:anytime-p-val}, it is also sufficient to derive an \emph{$e$-process} $\infseq {\bar E}t1$ --- a $\Zcal$-adapted process that is upper bounded by an NSM \emph{for every element of $\Hcal_0$}.
Formally, $\infseq {\bar E}t1$ is an $e$-process for $\Hcal_0$ if for every $Q \in \Hcal_0$, there exists a $Q$-NSM $\infseq {M^Q}t1$ such that
\begin{equation}
\small
   \label{eq:e-proc}
   \forall t,\ \bar E_t \leq M_t^Q, ~~\text{$Q$-almost surely.}
\end{equation}
Here, $\infseq {M^Q}t1$ being a $Q$-NSM means that $\EE_Q M^Q_t \leq M^Q_{t-1}$, and $M_0^Q \equiv 1$, and $M_t^Q \geq 0$, $Q$-almost surely. Note that these upper-bounding NSMs need not be the same, i.e. $\infseq {\bar E}t1$ can be upper bounded by a different $Q$-NSM for each $Q \in \Hcal_0$.

Importantly, if $\infseq {\bar E}t1$ is an $e$-process under $\Hcal_0$, then $\phi_t := \1({\bar E}_t \geq 1/\alpha)$ forms a level-$\alpha$ sequential test satisfying \eqref{eq:sequential-test} by applying Ville's inequality to the NSM that upper bounds $\infseq {\bar E}t1$:
\begin{equation}
  \small
\label{eq:e-proc-timeuniform}
    \sup_{Q \in \Hcal_0} Q(\exists t: \bar E_t \geq 1/\alpha) \leq
    \alpha.
\end{equation}

Using the same technique it is easy to see that, $\bar p_t := 1/\bar E_t$ forms an anytime $p$-value satisfying \eqref{eq:anytime-p-val}.
Similarly to Section~\ref{section:nprr-ci}, if we are only interested in inference at a fixed sample size $n$, we can still leverage $e$-processes to obtain sharp finite-sample $p$-values from private data by simply taking
\begin{equation}
\small
    \dot p_n := \min_{1\leq t \leq n} 1/\bar E_t. \
\end{equation}
As an immediate consequence of \eqref{eq:e-proc-timeuniform}, we have $\sup_{Q \in \Hcal_0} Q(\dot p_n \leq \alpha) \leq \alpha$.

With all of this in mind, the question becomes: where can we find $e$-processes? The answer is simple: every single \cs{} and \ci{} in this paper was derived by first constructing an $e$-process under a point null, and
Table~\ref{table:cs-eproc-map} explicitly links all of these \cs{}s to their corresponding $e$-processes.\footnote{We do not link \ci{}s to $e$-processes since all of our \ci{}s are built using the aforementioned \cs{}s.} For more complex composite nulls however, there may exist $e$-processes that are not NSMs \citep{ramdas2021testing}, and we touch on one such example in Proposition~\ref{proposition:a/b-testing}.
\begin{table}[!htbp]
  \caption{A mapping between theorems containing confidence sequences and equations containing the explicit $e$-processes that underlie them.}
  \label{table:cs-eproc-map}
  \centering
\begin{tabular}{l|c}
\textbf{Confidence sequence}                           & \textbf{$e$-process}             \\ \hline
Theorem~\ref{theorem:ldp-hoeffding-cs}                 & \eqref{eq:nprr-hoeffding-nsm}    \\
Theorem~\ref{theorem:two-sided-cs-mean-so-far}         & \eqref{eq:two-sided-mixture-NSM} \\
Proposition~\ref{proposition:one-sided-cs-mean-so-far} & \eqref{eq:one-sided-mixture-NSM} \\
\end{tabular}
\end{table}

\paragraph{A note on locally private $e$-values.}
Similar to how the $p$-value is the fixed-time version of an anytime $p$-value, the so-called \emph{$e$-value} is the fixed-time version of an $e$-process. An $e$-value for a null $\Hcal_0$ is a nonnegative random variable $\dot E$ with $Q$-expectation at most one, meaning $\EE_Q (\dot E) \leq 1$ for any $Q \in \Hcal_0$ \citep{grunwald2019safe,vovk2021values}, and clearly by Markov's inequality, $1/\dot E$ is a $p$-value for $\Hcal_0$. Indeed, the time-uniform property \eqref{eq:e-proc-timeuniform} for the $e$-process $\infseq Et1$ is \emph{equivalent} to saying $E_\tau$ is an $e$-value for any stopping time $\tau$ \citep[Lemma 3]{howard2021time}; \citep[Proposition 1]{zhao2016adaptive}.



Given the shared goals between $e$- and $p$-values, a natural question arises: ``Should one use $e$-values or $p$-values for inference?''.
While $p$-values are the canonical measure of evidence in hypothesis testing, there are several reasons why one may prefer to work with $e$-values directly; some practical, and others philosophical. From a purely practical perspective, $e$-values make it simple to combine evidence across several studies \citep{grunwald2019safe,ter2021all,vovk2021values} or to control the false discovery rate under arbitrary dependence \citep{wang2020false}. They have also received considerable attention for philosophical reasons including how they relate testing to betting \citep{shafer2021testing} and connect frequentist and Bayesian notions of uncertainty \citep{grunwald2019safe,waudby2020confidence}. While the details of these advantages are well outside the scope of this paper, they are advantages that can now be enjoyed in locally private inference using our methods.

\subsection{A/B testing the weak null}\label{section:weak-null}
As described in Section~\ref{section:testing}, there is a close connection between \cs{}s and sequential hypothesis tests. The lower \cs{} $\infseq {\widetilde L^\Delta}{t}{1}$ presented in \cref{proposition:a/b-testing} is no exception, and can be used to test the weak null hypothesis, $\widetilde \Hcal_0$: $\forall t,\ \widetilde \Delta_t \leq 0$ (see \cref{fig:a/b-test}). In words, $\widetilde \Hcal_0$ is testing ``is the new treatment as bad or worse than placebo \emph{among the patients so far}?''.
Indeed, adapting \eqref{eq:e-process-mean-so-far} from the proof of Proposition~\ref{proposition:one-sided-cs-mean-so-far} to the current setting, we have the following anytime $p$-value for the weak null under locally private online A/B tests.

\begin{restatable}[]{proposition}{ABTesting}\label{proposition:a/b-testing}
  Consider the same setup as Corollary~\ref{corollary:a/b-estimation}, and let $\Phi(\cdot)$ be the cumulative distribution function of a standard Gaussian. Define for any $\beta > 0$,
  {\small \begin{equation*}
      \widetilde E^\Delta_t := \frac{2}{\sqrt{t\beta^2 + 1}} \exp \left \{ \frac{2 \beta^2 (S_{t,0}^\Delta)^2}{t\beta^2 + 1} \right \} \Phi\left( \frac{2\beta S_{t,0}^\Delta}{\sqrt{t\beta^2 + 1}} \right ),
    \end{equation*}}
  where $S_{t,0}^\Delta := \sum_{i=1}^t (\psi_i - (1-r)/2) - tr \frac{1/(1-\pi)}{1/\pi + 1/(1-\pi)}$ and $\beta > 0$. Then, $\widetilde E^\Delta_t$ forms an $e$-process and hence $\widetilde p_t^\Delta := 1/\widetilde E^\Delta_t$ forms an anytime $p$-value, and ${\widetilde \phi_t^\Delta := \1(\widetilde p^\Delta_t \leq \alpha)}$ forms a level-$\alpha$ sequential test for the weak null $\widetilde \Hcal_0$.
\end{restatable}
The proof provided in Section~\ref{proof:a/b-testing} relies on the simple observation that under $\widetilde \Hcal_0$, $\infseq {\widetilde E^\Delta}t1$ is upper bounded by a nonnegative supermartingale, and is hence an ``$e$-process''. We suggest choosing $\beta > 0$ in a similar manner to \cref{proposition:one-sided-cs-mean-so-far}.

\subsection{Locally private adaptive online A/B testing}\label{section:adaptive-a/b-testing}
In \cref{section:a/b-testing}, we demonstrated how our techniques can be used to conduct online A/B tests. However, those A/B tests were non-adaptive, in the sense that the propensity score $\pi \in (0, 1)$ was required to be the same constant for all individuals (e.g.~in a Bernoulli experiment). In this section, we briefly describe an alternative \cs{} that can be used to conduct \emph{adaptive} online A/B tests, where the propensity scores $\infseqt{\pi_t(X_t)}$ can change over time in a data-dependent fashion and be a function of some measured baseline covariates $\infseqt{X_t}$. Note that while we will still consider private tests in the sense of the outcomes $\infseqt{Y_t}$ being privatized, we will not be privatizing the covariates $\infseqt{X_t}$ (though this is an interesting direction for future work).

To set the stage, suppose that $(X_1, A_1, Y_1), (X_2, A_2, Y_2), \dots$ are joint random variables such that covariates $X_t \sim p_X(\cdot)$, are drawn according to some common distribution, treatments $A_t \sim \text{Bernoulli}(\pi_t(X_t))$ are drawn from a conditional distribution  $\pi_t$ (called the propensity score) which can be chosen based on $(X_i, A_i, Y_i)_{i=1}^{t-1}$, and $Y_t \sim p_Y(\cdot \mid A_t, X_t)$ is drawn from a common conditional distribution.\footnote{These distributional assumptions can be substantially weakened as in \citet{waudby2022anytime}, but we present this simplified setting for the sake of exposition.} In words, we have that for each subject $t$, covariates $X_t$ are observed, a propensity score $\pi_t$ is chosen based on all previous subjects, a binary treatment $A_t$ is drawn with probability $\pi_t(X_t)$, and a $[0, 1]$-bounded outcome $Y_t$ is observed based on subject $t$'s covariates and their treatment $A_t$. Of course, if $\pi(X_t) \equiv \pi$ for each $t$, then the above setup recovers the classical (non-adaptive) A/B testing setup considered in \cref{section:a/b-testing}.

Similarly to Section~\ref{section:a/b-testing}, we will construct $(1-\alpha)$-\cs{}s for the \emph{time-varying mean} $\widetilde \Delta_t := \frac{1}{t} \sum_{i=1}^t \Delta_i$ where
\begin{equation}
\small
  \Delta_i := \EE \left \{{\EE(Y_i \mid X_i, A_i = 1)} - {\EE(Y_i \mid X_i, A_i = 0)} \right \}
\end{equation}
is the individual treatment effect for subject $i$.
To state our main result, we need to prepare some notation. Let $w_t^{(1)} := \frac{\1(A_t = 1)}{\pi_t(X_t)}$ and $w_t^{(0)} := \frac{\1(A_t = 0)}{1-\pi_t(X_t)}$ denote the inverse propensity score weights for treatment and control groups, respectively, and define the following pseudo-outcomes $\theta_t := [w_t^{(1)}Z_t - (1-w_t^{(0)}(1-Z_t))]/r$, and the resulting variance process
\begin{equation}
\small
  V_t := \frac{1}{t}\sum_{i=1}^t \left ( \theta_i - \widehat \theta_{i-1} \right )^2,~\text{where}~\widehat \theta_t := \left (\frac{1}{t} \sum_{i=1}^t \theta_i \right ) \land 1
\end{equation}
We are now ready to state the main result of this section.

\begin{restatable}[Locally private adaptive A/B estimation]{theorem}{AdaptiveAbTest}\label{theorem:a/b-test}
  Let $S_t(\widetilde \Delta_t') := (\sum_{i=1}^t \theta_i - t\widetilde \Delta_t')/2$ for any $\widetilde \Delta_t' \in [0, 1]$ and define for any $\rho > 0$,
  \begin{equation}\label{eq:EB-conjugate-mixture-supermartingale}
    \small
    \widetilde M_t^\EB(\widetilde \Delta_t') := \left ( \frac{\rho^\rho e^{-\rho}}{\Gamma(\rho) - \Gamma(\rho, \rho)} \right ) \left ( \frac{1}{V_t+\rho} \right ) F_t(\widetilde \Delta_t'),
  \end{equation}
  where $F_t(\widetilde \Delta_t') := \onefone(1, V_t+\rho+1, S_t(\widetilde \Delta_t') + V_t + \rho)$, and $\onefone$ is Kummer's confluent hypergeometric function, and $\Gamma(\cdot, \cdot)$ is the upper incomplete gamma function. Then, when evaluated at the true $\widetilde \Delta_t$, we have that $\widetilde M_t^\EB(\widetilde \Delta_t)$ forms a nonnegative supermartingale. Consequently,
  \begin{equation}
    \small
    \widetilde L_t^\Delta := \inf \left \{ \widetilde \Delta_t \in [0, 1] : \widetilde M_t^\EB(\widetilde \Delta_t) < 1/\alpha \right \}
  \end{equation}
  forms a lower $(1-\alpha)$-\cs{} for the running ATE $\widetilde \Delta_t$.
\end{restatable}

The proof can be found in \cref{proof:variance-adaptive-a/b-test}. Readers familiar with the semiparametric causal inference literature will notice that $\theta_t$ takes the form of a modified inverse-probability-weighted (IPW) influence function, and that doubly robust (also known as ``augmented IPW'') approaches are often superior both theoretically and empirically. In principle, the above discussion can be modified to handle doubly robust pseudo-outcomes and \cs{}s using the ideas contained in \citet[Section 2.1]{waudby2022anytime}, but we presented the IPW-based approach instead for the sake of simplicity.

\section{Proofs of additional results}
\label{proofs:additional-results}

\subsection{Proof of \cref{proposition:laplace-hoeffding-cs}}
\label{proof:laplace-hoeffding}

\LaplaceHoeffdingCS*

\begin{proof}
The proof proceeds in two steps. First, we construct an exponential NSM using the cumulant generating function of a Laplace distribution. Second and finally, we apply Ville's inequality to the NSM and invert it to obtain the lower \cs{}.
\paragraph{Step 1.}
Consider the following process for any $\mu \in [0, 1]$,
\[ M_t^\lapNoise(\mu) := \prod_{i=1}^t \exp\left \{ \lambda_i (Z_i - \mu) - \lambda_i^2 / 8 - \psi_i^\lapNoise(\lambda_i) \right \}, \]
with $M_t^\lapNoise(\mu)\equiv 1$. We claim that $(M_t^\lapNoise(\mu^\star))_{t=0}^\infty$ forms an NSM with respect to the private filtration $\Zcal$. Indeed, $(M_t^\lapNoise(\mu^\star))_{t=0}^\infty$ is nonnegative and starts at one by construction. It remains to prove that $(M_t^\lapNoise(\mu^\star))$ is a supermartingale, meaning $\EE(M_t^\lapNoise(\mu^\star) \mid \Zcal_{t-1}) \leq M_{t-1}^\lapNoise(\mu^\star)$. Writing out the conditional expectation of $M_t^\lapNoise(\mu^\star)$, we have
  \begin{align*}
    &\EE \left ( M_t^\lapNoise(\mu^\star) \mid \Zcal_{t-1} \right  )\\
    =\ &\EE \left ( \prod_{i=1}^t \exp \left \{ \lambda_i ( Z_i - \mu^\star ) - \lambda_i^2 / 8 - \psi_i^\lapNoise(\lambda_i) \right \} \biggm \vert \Zcal_{t-1} \right ) \\
                                                               =\ &\underbrace{\prod_{i=1}^{t-1} \exp \left \{ \lambda_i ( Z_i - \mu^\star ) - \lambda_i^2 / 8 -\psi_i^\lapNoise(\lambda_i) \right \}}_{M_{t-1}(\mu^\star)} \cdot \underbrace{\EE \left ( \exp \left \{ \lambda_t ( Z_t - \mu^\star ) - \lambda_t^2 / 8  - \psi_t^\lapNoise(\lambda_t)\right \} \biggm \vert \Zcal_{t-1} \right )}_{(\dagger)},
  \end{align*}
  since $M_{t-1}^\lapNoise(\mu^\star)$ is $\Zcal_{t-1}$-measurable, and thus it can be written outside of the conditional expectation. It now suffices to show that $(\dagger) \leq 1$. To this end, note that $Z_t = X_t + \lapNoise_t$ where $X_t$ is a $[0, 1]$-bounded, mean-$\mu^\star$ random variable, and $\lapNoise_t$ is a mean-zero Laplace random variable (conditional on $\Zcal_{t-1}$).
  Consequently, $\EE (\exp\left \{ \lambda_tX_t \right \} \mid \Zcal_{t-1}) \leq \exp\{ \lambda_t^2/8 \}$ by Hoeffding's inequality \citep{hoeffding1963probability}, and $\EE(\exp\{ \lambda_t \lapNoise_t \} \mid \Zcal_{t-1}) = \exp\left \{ \lambda_t^\lapNoise(\lambda_t) \right \}$ by definition of a Laplace random variable. Moreover, note that by design of Algorithm~\ref{algorithm:seqIntLaplace}, $X_t$ and $\lapNoise_t$ are conditionally independent. It follows that
  \begin{align*}
      (\dagger) &= \EE \left ( \exp \left \{ \lambda_t ( Z_t - \mu^\star ) - \lambda_t^2 / 8  - \psi_t^\lapNoise(\lambda_t)\right \} \Bigm \vert \Zcal_{t-1} \right )\\
      &=\EE \left ( \exp \left \{ \lambda_t ( X_t - \mu^\star ) + \lambda_t \lapNoise_t - \lambda_t^2/8 - \psi_t^\lapNoise(\lambda_t)\right \} \Bigm \vert \Zcal_{t-1} \right )\\
      &= \underbrace{\EE \left ( \exp \left \{ \lambda_t ( X_t - \mu^\star ) - \lambda_t^2/8 \right \} \Bigm \vert \Zcal_{t-1} \right )}_{\leq 1} \cdot \underbrace{\EE \left ( \exp \left \{ \lambda_t \lapNoise_t - \psi_t^\lapNoise(\lambda_t)\right \} \Bigm \vert \Zcal_{t-1} \right )}_{= 1} \leq 1,
  \end{align*}
  where the third equality follows from the conditional independence of $X_t$ and $\lapNoise_t$. Therefore, $(M_t^\lapNoise(\mu^\star))_{t=0}^\infty$ is an NSM.

  \paragraph{Step 2.} By Ville's inequality, we have that
  \[ \PP\left (\forall t,\ M_t^\lapNoise(\mu^\star) < 1/\alpha \right ) \geq 1-\alpha. \]
  Let us rewrite the inequality $M_t^\lapNoise(\mu^\star) < 1/\alpha$ so that we obtain the desired lower \cs{}.

  \begin{align*}
    M_t^\lapNoise(\mu^\star) < 1/\alpha &\iff \prod_{i=1}^t \exp\left \{ \lambda_t (Z_t - \mu^\star) - \lambda_i^2 / 8 - \psi_i^\lapNoise(\lambda_i) \right \} < 1/\alpha\\
    &\iff \sum_{i=1}^t\left [ \lambda_t (Z_t - \mu^\star) - \lambda_i^2 / 8 - \psi_i^\lapNoise(\lambda_i) \right] < \log(1/\alpha)\\
    &\iff \sum_{i=1}^t \lambda_t Z_t - \sum_{i=1}^t [\lambda_i^2 / 8 + \psi_i^\lapNoise(\lambda_i)] -\mu^\star \sum_{i=1}^t \lambda_i   < \log(1/\alpha) \\
    &\iff \mu^\star > \frac{\sum_{i=1}^t \lambda_t Z_t}{\sum_{i=1}^t \lambda_i} - \frac{\log(1/\alpha) + \sum_{i=1}^t (\lambda_i^2 / 8 + \psi_i^\lapNoise(\lambda_i))}{\sum_{i=1}^t \lambda_i}.
  \end{align*}
  In summary, the above inequality holds uniformly for all $t \in \{1,2,\dots\}$ with probability at least $(1-\alpha)$. In other words,
  \[ \frac{\sum_{i=1}^t \lambda_t Z_t}{\sum_{i=1}^t \lambda_i} - \frac{\log(1/\alpha) + \sum_{i=1}^t (\lambda_i^2 / 8 + \psi_i^\lapNoise(\lambda_i))}{\sum_{i=1}^t \lambda_i}\]
  forms a $(1-\alpha)$-lower \cs{} for $\mu^\star$. An analogous upper-\cs{} can be derived by applying the same technique to $-Z_1, -Z_2, \dots$ and their mean $-\mu^\star$. This completes the proof.
\end{proof}

\subsection{A lemma for Theorems~\ref{theorem:ldp-hedged-ci} and~\ref{theorem:ldp-dkelly}}
\label{proof:ldp-betting}

To prove Theorems~\ref{theorem:ldp-dkelly} and~\ref{theorem:ldp-hedged-ci}, we will prove a more general result (Lemma~\ref{lemma:ldp-betting}), and use it to instantiate both theorems as immediate consequences. The proof follows a similar technique to \citet{waudby2020estimating} but adapted to the locally private setting.

\begin{lemma}
  \label{lemma:ldp-betting}
  Suppose $\infseq Xt1 \sim P$ for some $P \in \Pcal_{\mu^\star}^\infty$ and let $\infseq Zt1$ be their \NPRR{}-induced privatized views. Let $\theta_1, \dots, \theta_D \in [0, 1]$ be convex weights satisfying $\sum_{d=1}^{D} \theta_d = 1$ and let $\infseq{\lambda^{(d)}}{t}{1}$ be a $\Zcal$-predictable sequence for each $d \in \{1, \dots, D\}$ such that $\lambda_t \in (-(1-\zeta_t(\mu^\star))^{-1}, \zeta_t(\mu^\star)^{-1})$. Then the process formed by
  \begin{equation}
    \label{eq:general-betting-martingale}
    M_t := \sum_{d=1}^{D} \theta_d \prod_{i=1}^t (1 + \lambda_i^{(d)}\cdot (Z_i - \zeta_i(\mu^\star)))
  \end{equation}
  is a nonnegative martingale starting at one. Further suppose that $(\breve M_t(\mu))_{t=0}^\infty$ is a process for any $\mu \in (0, 1)$ that when evaluated at $\mu^\star$, satisfies $\breve M_t(\mu^\star) \leq M_t$ almost surely for each $t$. Then
  \begin{equation}
    \label{eq:general-betting-cs}
    \breve C_t := \left \{ \mu \in (0, 1) : \breve M_t(\mu) < 1/\alpha \right \}
  \end{equation}
  forms a $(1-\alpha)$-\cs{} for $\mu^\star$.
\end{lemma}
\begin{proof}
  The proof proceeds in three steps. First, we will show that the product processes given by $\prod_{i=1}^t(1 + \lambda_i\cdot (Z_i - \zeta_i(\mu^\star)))$ form nonnegative martingales with respect to $\Zcal$. Second, we argue that $\sum_{d=1}^{D} \theta_d M_t^{(d)}$ forms a martingale for any $\Zcal$-adapted martingales. Third and finally, we argue that $\breve C_t$ forms a $(1-\alpha)$-\cs{} despite not being constructed from a martingale directly.

  \paragraph{Step 1.}
  We wish to show that $M_t^{(d)} := \prod_{i=1}^t (1 + \lambda_i^{(d)} \cdot (Z_i - \zeta_i(\mu^\star)))$ forms a nonnegative martingale starting at one given a fixed $d \in \{1,\dots, D\}$. Nonnegativity follows immediately from the fact that $\lambda_t \in (-[1-\zeta_t(\mu^\star)]^{-1},\ \zeta_t(\mu^\star)^{-1})$, and $M_t^{(d)}$ begins at one by design. It remains to show that $M_t^{(d)}$ forms a martingale. To this end, consider the conditional expectation of $M_t^{(d)}$ for any $t \in \{1, 2, \dots\}$,
  \begin{align*}
    \EE \left ( M_t^{(d)} \mid \Zcal_{t-1} \right ) &= \EE \left ( \prod_{i=1}^t (1 + \lambda_i^{(d)} \cdot (Z_i - \zeta_i(\mu^\star))) \biggm \vert \Zcal_{t-1} \right ) \\
                                                &= \underbrace{\prod_{i=1}^{t-1}(1 + \lambda_i^{(d)} \cdot (Z_i - \zeta_i(\mu^\star)))}_{M_{t-1}^{(d)}} \cdot \EE \left (1 + \lambda_t^{(d)} \cdot (Z_t - \zeta_t(\mu^\star)) \mid \Zcal_{t-1} \right )\\
                                                &= M_{t-1}^{(d)} \cdot \left ( 1 + \lambda_t^{(d)} \cdot \underbrace{\left [\EE(Z_t\mid \Zcal_{t-1}) - \zeta_t(\mu^\star) \right ]}_{=0} \right ) \\
                                                &= M_{t-1}^{(d)}.
  \end{align*}
  Therefore, $\infseq {M^{(d)}}{t}{0}$ forms a martingale.
  \paragraph{Step 2.}
  Now, suppose that $M_t^{(1)}, \dots, M_t^{(D)}$ are test martingales with respect to the private filtration $\Zcal$, and let ${\theta_1, \dots, \theta_D \in [0, 1]}$ be convex weights, i.e. satisfying $\sum_{d=1}^D \theta_d = 1$. Then $M_t := \sum_{d=1}^D M_t^{(d)}$ also forms a martingale since
  \begin{align*}
    \EE (M_t \mid \Zcal_{t-1}) &=  \EE \left ( \sum_{d=1}^D \theta_d M_t^{(d)} \Bigm \vert \Zcal_{t-1} \right )\\
                               &= \sum_{d=1}^D \theta_d \EE \left ( M_t^{(d)} \mid \Zcal_{t-1} \right ) \\
                               &= \sum_{d=1}^D \theta_d M_{t-1}^{(d)} \\
                               &= M_{t-1}.
  \end{align*}
  Moreover, $\infseq Mt0$ starts at one since $M_0 := \sum_{d=1}^D \theta_d M_0^{(d)} = \sum_{d=1}^D \theta_d = 1$. Finally, nonnegativity follows from the fact that $\theta_1, \dots, \theta_D$ are convex and each $\infseq {M^{(d)}}t0$ is almost-surely nonnegative. Therefore, $\infseq Mt0$ is a test martingale.
  \paragraph{Step 3.}
  Now, suppose $(\breve M_t(\mu))_{t=0}^\infty$ is a process that is almost-surely upper-bounded by $\infseq Mt0$. Define ${\breve C_t := \left \{ \mu \in (0, 1) : \breve M_t(\mu) < 1/\alpha \right \}}$. Writing out the probability of $\breve C_t$ miscovering $\mu^\star$ for any $t$, we have
  \begin{align*}
    \PP(\exists t : \mu^\star \notin \breve C_t) &= \PP(\exists t : \breve M_t(\mu^\star) \geq 1/\alpha) \\
                                                     &\leq \PP(\exists t : M_t \geq 1/\alpha)\\
                                                     &\leq \alpha,
  \end{align*}
  where the first inequality follows from the fact that $\breve M_t(\mu^\star) \leq M_t$ almost surely for each $t$, and the second follows from Ville's inequality \citep{ville1939etude}. This completes the proof of Lemma~\ref{lemma:ldp-betting}.
\end{proof}

In fact, a more general ``meta-algorithm'' extension of Lemma~\ref{lemma:ldp-betting} holds, following the derivation of the ``Sequentially Rebalanced Portfolio'' in \citet[Section 5.8]{waudby2020estimating} but we omit these details for the sake of simplicity.

\subsection{Proof of Theorem~\ref{theorem:ldp-hedged-ci}}
\label{proof:ldp-hedged-ci}
\LDPHedgedCI*

\begin{proof}
The proof of Theorem~\ref{theorem:ldp-hedged-ci} proceeds in three steps. First, we show that $\Kcal_{t,n}$ is nonincreasing and continuous in $\mu \in [0, 1]$, making $\dot L_n$ simple to compute via line/grid search. Second, we show that $\Kcal_{t,n}(\mu^\star)$ forms a $\QcalNPRRstarinf$-NM. Third and finally, we show that $\dot L_n$ is a lower \ci{} by constructing a lower \cs{} that yields $\dot L_n$ when instantiated at $n$.

\paragraph{Step 1. $\Kcal_{t,n}(\mu)$ is nonincreasing and continuous.}

  To simplify the notation that follows, write $g_{i,n}(\mu) := 1 + \lambda_{i, n}(\mu) \cdot (Z_i - \zeta_i(\mu))$ so that
  \[ \Kcal_{t,n}(\mu) \equiv \prod_{i=1}^t g_{i,n}(\mu). \]

Now, recall the definition of $\lambda_{i,n}(\mu)$,
\begin{align*}
  &\lambda_{t,n}(\mu) := \underbrace{\sqrt{\frac{2 \log(1/\alpha)}{\widehat \gamma_{t-1}^2 n}}}_{\eta} \land \frac{c}{\zeta_t(\mu)},~\text{where}\\
  &\widehat \gamma_t^2 := \frac{1/4 + \sum_{i=1}^t (Z_i - \widehat \zeta_i)^2}{t + 1},\ \widehat \zeta_t := \frac{1/2 + \sum_{i=1}^t Z_i}{t + 1}.
\end{align*}
Notice that $\lambda_{t,n}(\mu) \equiv \eta \land c/\zeta_t(\mu)$ is nonnegative and does not depend on $\mu$ except through the truncation with $ c / \zeta_t(\mu)$. In particular we can write $g_{i,n}(\mu)$ as
\begin{align*}
    g_{i,n}(\mu) &\equiv 1 + \left (\eta \land \frac{c}{\zeta_i(\mu)} \right) (Z_i -\zeta_i(\mu)) \\
    &= 1 + (\eta Z_i) \land \frac{cZ_i}{\zeta_i(\mu)} - \eta\zeta_i(\mu) \land c,
\end{align*}
which is a nonincreasing (and continuous) function of $\zeta_i(\mu)$. Since $\zeta_i(\mu) := r_i\mu + (1-r_i)/2$ is an increasing (and continuous) function of $\mu$, we have that $g_{i,n}(\mu)$ is nonincreasing and continuous in $\mu$.

Moreover, we have that $g_{i,n}(\mu) \geq 0$ by design, and the product of nonnegative nonincreasing functions is also nonnegative and nonincreasing, so $\Kcal_{t,n} = \prod_{i=1}^t g_{i,n}(\mu)$ is nonincreasing.

\paragraph{Step 2. $\Kcal_{t,n}(\mu^\star)$ is a $\QcalNPRRstarinf$-NM.}

  Recall the definition of $\Kcal_{t,n}(\mu^\star)$
  \begin{equation*}
    \Kcal_{t,n}(\mu^\star) := \prod_{i=1}^t \left [1 + \lambda_{i,n}(\mu^\star) \cdot (Z_i - \zeta_i(\mu^\star)) \right ]
  \end{equation*}

Then by Lemma~\ref{lemma:ldp-betting} with $D=1$ and $\theta_1 = 1$, $\Kcal_{t,n}(\mu^\star)$ is a $\QcalNPRRstarn$-NM.

\paragraph{Step 3. $\dot L_n$ is a lower \ci{}.}
First, note that by Lemma~\ref{lemma:ldp-betting}, we have that
\[ C_t := \{ \mu \in [0, 1] : \Kcal_{t,n}(\mu) < 1/\alpha \} \]
forms a $(1-\alpha)$-\cs{} for $\mu^\star$. In particular, define
\[ \bar L_{t,n} := \inf \{ \mu \in [0, 1] : \Kcal_{t,n}(\mu) < 1/\alpha \}.\]
Then, $[\bar L_{t,n}, 1]$ forms a $(1-\alpha)$-\cs{} for $\mu^\star$, meaning $\PP(\forall t,\ \mu^\star \geq L_{t,n} ) \geq 1-\alpha$, and hence
\[\PP\left (\mu^\star \geq \max_{1\leq t \leq n} L_{t,n}\right ) = \PP\left (\mu^\star \geq \dot L_n \right ) \geq 1-\alpha.\]
This completes the proof.
\end{proof}

\subsection{Proof of Theorem~\ref{theorem:ldp-dkelly}}
\label{proof:ldp-dkelly}

\LDPDKelly*

\begin{proof}
  The proof will proceed in two steps. First, we will invoke Lemma~\ref{lemma:ldp-betting} to justify that $\bar C_t^\GK$ indeed forms a \cs{}. Second and finally, we prove that $\bar C_t^\GK$ forms an interval almost surely for each $t \in \{1,2,\dots\}$ by showing that $\Kcal_t^\GK(\mu)$ is a convex function.
\paragraph{Step 1. $\bar C_t^\GK$ forms a \cs{}.} Notice that by Lemma~\ref{lemma:ldp-betting}, we have that $\Kcal_t^+(\mu^\star)$ and $\Kcal_t^-(\mu^\star)$ defined in Theorem~\ref{theorem:ldp-dkelly} are both test martingales. Consequently, their convex combination
\[\Kcal_t^\GK (\mu^\star) := \theta\Kcal_t^+(\mu^\star) + (1-\theta)\Kcal_t^-(\mu^\star)\]
is also a test martingale. Therefore, $\bar C_t^\GK := \left \{ \mu\in [0, 1] : \Kcal_t^\GK(\mu) < 1/\alpha \right \}$ indeed forms a $(1-\alpha)$-\cs{}.

\paragraph{Step 2. $\bar C_t^\GK$ is an interval almost surely.}
We will now justify that $\bar C_t^\GK$ forms an interval by proving that $\Kcal_t^\GK(\mu)$ is a convex function of $\mu \in [0, 1]$ and noting that the sublevel sets of convex functions are themselves convex.

To ease notation, define the multiplicands $g_i^+(\mu) := 1 + \lambda_{i, d}^+ \cdot (Z_i - \zeta_i(\mu))$ so that
\[ \Kcal_t^+ (\mu) \equiv \prod_{i=1}^t g_i(\mu). \]
Rewriting $g_i(\mu)$, we have that
\begin{equation*}
  \label{eq:rewrite-multiplicand}
  1 + \lambda_{i, d}^+ \cdot (Z_i - \zeta_i(\mu)) = 1 + \frac{d}{D+1} \cdot \left ( \frac{Z_i}{r_i\mu + (1-r_i)/2} - 1 \right ),
\end{equation*}
from which it is clear that each $g_i(\mu)$ is (a) nonnegative, (b) nonincreasing, and (c) convex in $\mu \in [0, 1]$. Now, note that properties (a)--(c) are preserved under products \citep[Section A.7]{waudby2020estimating}, meaning
\[ \Kcal_t^+(\mu) \equiv \prod_{i=1}^t g_i(\mu) \]
also satisfies (a)--(c).

A similar argument goes through for $\Kcal_t^-(\mu)$, except that this function is nonincreasing rather than nondecreasing, but it is nevertheless nonnegative and convex. Since convexity of functions is preserved under convex combinations, we have that
\[   \Kcal_t^\GK(\mu) := \theta \Kcal_t^+(\mu) + (1-\theta)\Kcal_t^-(\mu) \]
is a convex function of $\mu \in [0, 1]$.

Finally, observe that $\bar C_t^\GK$ is the $(1/\alpha)$-sublevel set of $\Kcal_t^\GK(\mu)$ by definition, and the sublevel sets of convex functions are convex. Therefore, $\bar C_t^\GK$ is an interval almost surely. This completes the proof of Theorem~\ref{theorem:ldp-dkelly}.

\end{proof}

\subsection{Proof of \cref{proposition:ldp-eb-cs}}
\label{proof:ldp-eb-cs}

\NprrEbCs*

\begin{proof}
The proof proceeds in two steps. First, we derive a sub-exponential NSM\@. Second and finally, we apply Ville's inequality to the NSM and invert it to obtain $\infseqt{\dot L_t^\EB}$.

\paragraph{Step 1: Deriving a sub-exponential nonnegative supermartingale.}

Consider the process $\infseqt{M_t^\EB(\mu^\star)}$ given by
\begin{equation}
 M_t^\EB(\mu^\star) := \prod_{i=1}^t \exp \left \{ \lambda_i \cdot (Z_i -\zeta_i(\mu^\star)) - 4(Z_i- \widehat \zeta_{i-1}(\mu^\star))^2\psi_E(\lambda_i) \right \},
\end{equation}
and defined as $M_0^\EB(\mu^\star) \equiv 1$. Clearly, $M_t^\EB > 0$, and hence in order to show that $\infseqt{M_t^\EB(\mu^\star)}$ is an NSM, it suffices to show that $\EE(M_t^\EB(\mu^\star) \mid \Zcal_{t-1}) = M_{t-1}^\EB(\mu^\star)$ for each $t \geq 1$.
To this end, we have that
\begin{align}
  \EE(M_t^\EB(\mu^\star) \mid \Zcal_{t-1}) &= \EE \left ( \prod_{i=1}^t \exp \left \{ \lambda_i \cdot (Z_i -\zeta_i(\mu^\star)) - 4(Z_i- \widehat \zeta_{i-1}(\mu^\star))^2\psi_E(\lambda_i) \right \} \mid \Zcal_{t-1} \right ) \\
  &= M_{t-1}^\EB(\mu^\star) \underbrace{\EE \left ( \exp \left \{ \lambda_t \cdot (Z_t -\zeta_t(\mu^\star)) - 4(Z_t- \widehat \zeta_{t-1}(\mu^\star))^2\psi_E(\lambda_t) \right \} \mid \Zcal_{t-1} \right )}_{(\star)},
\end{align}
and hence it suffices to show that $(\star) \leq 1$. Following the proof of \citet[Theorem 2]{waudby2020estimating}, denote 
\begin{equation}
  Y_{t} := Z_{t} - \zeta_t(\mu^\star)\quad\text{and}\quad\delta_t := \widehat \zeta_t(\mu^\star) - \zeta_t(\mu^\star).
\end{equation}
Note that $\EE(Y_t \mid \Zcal_{t-1}) = 0$.
and thus it suffices to prove that for any $[0, 1)$-bounded, $\Zcal_{t-1}$-measurable $\lambda_{t}$, 
\[ \EE \left (\exp \Bigg \{ \lambda_{t} Y_{t} - 4(Y_{t} - \delta_{t-1})^2\psi_{E}(\lambda_{t}) \Bigg \}  \Bigm | \Fcal_{t-1} \right ) \leq 1.\]
Indeed, in the proof of~\citet[Proposition 4.1]{fan2015exponential}, $\exp\{\xi\lambda  - 4\xi^2 \psi_{E} (\lambda)\} \leq 1 + \xi \lambda$ for any $\lambda \in [0, 1)$ and $\xi \geq -1$. Setting $\xi := Y_{t} - \delta_{t-1} = Z_{t} - \widehat \zeta_{t-1}(\mu^\star)$,
\begin{align*}
	&\EE \left (\exp \Bigg \{ \lambda_{t} Y_{t} - 4(Y_{t} - \delta_{t-1})^2 \psi_{E}(\lambda_{t} ) \Bigg \}  \Bigm | \Zcal_{t-1} \right )\\
	=\ &\EE \left ( \exp \Big \{ \lambda_{t} (Y_{t}-\delta_{t-1}) - 4(Y_{t} - \delta_{t-1})^2 \psi_{E}(\lambda_{t}) \Big \}  \bigm | \Zcal_{t-1} \right ) \exp(\lambda_{t} \delta_{t-1})\\
	    \leq\ & \EE \left ( 1 + (Y_{t} - \delta_{t-1} )\lambda_{t} \mid \Zcal_{t-1} \right ) \exp(\lambda_{t} \delta_{t-1}) \overset{(i)}{=} \EE \left ( 1 - \delta_{t-1} \lambda_{t} \mid \Zcal_{t-1} \right ) \exp(\lambda_{t} \delta_{t-1}) \overset{(ii)}{\leq} 1,
\end{align*}
where $(i)$ follows from the fact that $Y_{t}$ is conditionally mean zero,
and $(ii)$ follows from the inequality $1-x \leq \exp(-x)$ for all $x \in \mathbb R$. This completes the proof of Step 1.

\paragraph{Step 2: Applying Ville's inequality and inverting.}

Now that we have established that $\infseqt{M_t^\EB(\mu^\star)}$ is an NSM, we have by Ville's inequality \citep{ville1939etude} that
\begin{equation}
  \PP(\exists t \geq 1 : M_t^\EB(\mu^\star) \geq 1/\alpha) \leq \alpha,
\end{equation}
or equivalently, $\PP(\forall t\geq 1,\ M_t^\EB(\mu^\star) < 1/\alpha) \geq 1-\alpha$. Consequently, we have that with probability at least $(1-\alpha)$,
\begin{align}
  & M_t^\EB(\mu^\star) > 1/\alpha \prod_{i=1}^t \exp \left \{ \lambda_i \cdot (Z_i -\zeta_i(\mu^\star)) - 4(Z_i- \widehat \zeta_{i-1}(\mu^\star))^2\psi_E(\lambda_i) \right \} < 1/\alpha \\
                                \iff\ &\sum_{i=1}^t \lambda_i (Z_i - \zeta_i(\mu^\star)) - \sum_{i=1}^t 4(Z_i - \widehat \zeta_{i-1}(\mu^\star))^2 \psi_E(\lambda_i) < \log (1/\alpha) \\
  \iff\ &  \sum_{i=1}^t \lambda_i\zeta_i(\mu^\star) > \sum_{i=1}^t \lambda_i Z_i - \sum_{i=1}^t 4(Z_i - \widehat \zeta_{i-1}(\mu^\star))^2 \psi_E(\lambda_i) - \log(1/\alpha) \\
  \iff\ & \sum_{i=1}^t \lambda_i \left (r_i \mu^\star + \frac{1-r_i}{2} \right ) > \sum_{i=1}^t \lambda_i Z_i - \sum_{i=1}^t 4(Z_i - \widehat \zeta_{i-1}(\mu^\star))^2 \psi_E(\lambda_i) - \log(1/\alpha) \\
  \iff\ & \mu^\star\sum_{i=1}^t \lambda_ir_i + \sum_{i=1}^t \lambda_i \frac{1-r_i}{2} > \sum_{i=1}^t \lambda_i Z_i - \sum_{i=1}^t 4(Z_i - \widehat \zeta_{i-1}(\mu^\star))^2 \psi_E(\lambda_i) - \log(1/\alpha) \\
  \iff\ & \mu^\star> \underbrace{\frac{\sum_{i=1}^t \lambda_i (Z_i - (1-r_i) / 2)}{\sum_{i=1}^t \lambda_i r_i }}_{\widehat \mu_t(\lambda_1^t)} - \underbrace{\frac{\sum_{i=1}^t 4(Z_i - \widehat \zeta_{i-1}(\mu^\star))^2 \psi_E(\lambda_i) + \log(1/\alpha)}{\sum_{i=1}^t \lambda_i r_i}}_{\widebar B_t(\lambda_1^t)}, \\
\end{align}
and hence $\widehat \mu_t(\lambda_1^t) - \widebar B_t(\lambda_1^t)$ forms a lower $(1-\alpha)$-\cs{}. This completes the proof of~\cref{proposition:ldp-eb-cs}.

\end{proof}

\subsection{Proof of Proposition~\ref{proposition:a/b-testing}}\label{proof:a/b-testing}
\ABTesting*
\begin{proof}
  In order to show that $\widetilde E^\Delta_t$ is an $e$-process, it suffices to find an NSM that almost surely upper bounds $\widetilde E^\Delta_t$ for each $t$ under the weak null $\widetilde \Hcal_0$: $\widetilde \Delta_t \leq 0$. As such, the proof proceeds in three steps. First, we justify why the one-sided NSM \eqref{eq:one-sided-mixture-NSM} given by Proposition~\ref{proposition:one-sided-cs-mean-so-far} is a nonincreasing function of $\widetilde \mu_t$. Second, we adapt the aforementioned NSM to the A/B testing setup to obtain $M_t^\Delta(\widetilde \Delta_t)$ and note that it is a nonincreasing function of $\widetilde \Delta_t$. Third and finally, we observe that $\widetilde E^\Delta_t := M_t^\Delta(0)$ is upper bounded by $M_t^\Delta(\widetilde \Delta_t)$ under the weak null, thus proving the desired result.

  \paragraph{Step 1: The one-sided NSM \eqref{eq:one-sided-mixture-NSM} is nonincreasing in $\widetilde \mu_t$.}
  Recall the $\lambda$-indexed process from Step 1 of the proof of Proposition~\ref{proposition:one-sided-cs-mean-so-far} given by
  \begin{equation*}
      M_t(\lambda) := \prod_{i=1}^t \exp \left \{ \lambda(Z_i - \zeta(\mu_i)) - \lambda^2 / 8 \right \},
  \end{equation*}
  which can be rewritten as
  \begin{equation*}
      M_t(\lambda) := \exp \left \{ S_t(\widetilde \mu_t) - \lambda^2 / 8 \right \},
  \end{equation*}
  where $S_t(\widetilde \mu_t) := \sum_{i=1}^t (Z_i - (1-r)/2) - tr\widetilde \mu_t$ and $\widetilde \mu_t := \frac{1}{t}\sum_{i=1}^t \mu_t$.
  In particular, notice that $M_t(\lambda)$ is a nonincreasing function of $\widetilde \mu_t$ for any $\lambda \geq 0$, and hence we also have that
  \begin{equation*}
    M_t(\lambda) f_{\rho^2}^+(\lambda)
  \end{equation*}
  is a nonincreasing function of $\widetilde \mu_t$ where $f_{\rho^2}^+ (\lambda)$ is the density of a folded Gaussian distribution given in \eqref{eq:folded-gaussian-pdf}, by virtue of $f^+_{\rho^2}(\lambda)$ being everywhere nonnegative, and 0 for all $\lambda < 0$. Finally, by Step 2 of the proof of Proposition~\ref{proposition:one-sided-cs-mean-so-far}, we have that
\begin{equation*}
\int_{\lambda} M_t(\lambda) f_{\rho^2}^+(\lambda) d\lambda \equiv \frac{2}{\sqrt{t\rho^2 / 4 + 1}} \exp \left \{ \frac{\rho^2 S_t(\widetilde \mu_t)^2}{2(t\rho^2 / 4 + 1)} \right \} \Phi\left( \frac{\rho S_t(\widetilde \mu_t)}{\sqrt{t\rho^2 / 4 + 1}} \right )
\end{equation*}
is nonincreasing in $\widetilde \mu_t$, and forms an NSM when evaluated at the true means $\infseq{\widetilde \mu^\star}t1$.

  \paragraph{Step 2: Applying Step 1 to the A/B testing setup to yield $M_t^\Delta(\widetilde \delta_t)$.}

  Adapting Step 1 to the setup described in Proposition~\ref{proposition:a/b-testing}, let $\delta_1, \delta_2, \dots \in \RR$ and let $\widetilde \delta_t := \sum_{i=1}^t\delta_i$. Define the partial sum process,
  \begin{equation*}
    S_t^\Delta(\widetilde \delta_t) := \sum_{i=1}^t (\psi_i - (1-r)/2) - rt \frac{\widetilde \delta_t +\frac{1}{1-\pi}}{\frac{1}{\pi} + \frac{1}{1-\pi}}
  \end{equation*}
  and the associated process,
  \begin{equation*}
    M_t^\Delta(\widetilde \delta_t) := \frac{2}{\sqrt{t\beta^2 + 1}} \exp \left \{ \frac{2 \beta^2 S_t^\Delta(\widetilde \delta_t)^2}{t\beta^2 + 1} \right \} \Phi\left( \frac{2\beta S_t^\Delta(\widetilde \delta_t)}{\sqrt{t\beta^2 + 1}} \right ),
  \end{equation*}
  where we have substituted $\rho:=2\beta > 0 $.
  Notice that by construction, $\psi_t$ is a [0, 1]-bounded random variable with mean $r\frac{\widetilde \Delta_t + 1/(1-\pi)}{1/\pi + 1/(1-\pi)} + (1-r)/2$, so $M_t^\Delta(\widetilde \Delta_t)$ forms an NSM\@. We are now ready to invoke the main part of the proof.

  \paragraph{Step 3: The process $\widetilde E^\Delta_t$ is upper-bounded by the NSM $M_t^\Delta(\widetilde \Delta_t)$.}
  Define the nonnegative process $\infseq {\widetilde E^\Delta}t0$ starting at one given by
  \begin{equation*}
    \widetilde E^\Delta_t := M_t^\Delta(0) \equiv \frac{2}{\sqrt{t\beta^2 + 1}} \exp \left \{ \frac{2 \beta^2 S_t^\Delta(0)^2}{t\beta^2 + 1} \right \} \Phi\left( \frac{2\beta S_t^\Delta(0)}{\sqrt{t\beta^2 + 1}} \right ).
  \end{equation*}
  By Steps 1 and 2, we have that $\widetilde E^\Delta_t \leq M_t^\Delta(\widetilde \Delta_t)$ for any $\widetilde \Delta_t \leq 0$, and since $M_t^\Delta(\widetilde \Delta_t)$ is an NSM, we have that
  $\infseq {\widetilde E^\Delta}t0$ forms an $e$-process for $\Hcal_0$: $\widetilde \Delta_t \leq 0$. This completes the proof.

\end{proof}

\subsection{Proof of Proposition~\ref{proposition:one-sided-cs-mean-so-far}}\label{proof:one-sided-cs-mean-so-far}
\OneSidedCSMeanSoFar*
\begin{proof}
The proof begins similar to that of Theorem~\ref{theorem:two-sided-cs-mean-so-far} but with a slightly modified mixing distribution, and proceeds in four steps. First, we derive a sub-Gaussian NSM indexed by a parameter $\lambda \in \RR$ identical to that of Theorem~\ref{theorem:two-sided-cs-mean-so-far}. Second, we mix this NSM over $\lambda$ using a folded Gaussian density, and justify why the resulting process is also an NSM\@. Third, we derive an implicit lower \cs{} for $\infseq{\widetilde \mu^\star}{t}{1}$. Fourth and finally, we compute a closed-form lower bound for the implicit \cs{}.

\paragraph{Step 1: Constructing the $\lambda$-indexed NSM.} This is exactly the same step as Step 1 in the proof of Theorem~\ref{theorem:two-sided-cs-mean-so-far}, found in Section~\ref{proof:two-sided-cs-mean-so-far}. In summary, we have that for any $\lambda \in \RR$,
\begin{equation}
    M_t(\lambda) := \prod_{i=1}^t \exp \left \{ \lambda(Z_i - \zeta(\mu_i^\star)) - \lambda^2 / 8 \right \},
\end{equation}
with $M_0(\lambda) \equiv 0$ forms an NSM with respect to the private filtration $\Zcal$.

\paragraph{Step 2: Mixing over $\lambda \in (0, \infty)$ to obtain a mixture NSM.} Let us now construct a one-sided sub-Gaussian mixture NSM\@. First, note that the mixture of an NSM with respect to a probability density is itself an NSM \citep{robbins1970statistical,howard2020time} and is a simple consequence of Fubini's theorem. For our purposes, we will consider the density of a \emph{folded Gaussian} distribution with location zero and scale $\rho^2$. In particular, if $\Lambda \sim N(0, \rho^2)$, let $\Lambda_+ := |\Lambda|$ be the folded Gaussian. Then $\Lambda_+$ has a probability density function $f_{\rho^2}^+(\lambda)$ given by

\begin{equation}
  \label{eq:folded-gaussian-pdf}
  f_{\rho^2}^+ (\lambda) := \1(\lambda > 0)\frac{2}{\sqrt{2\pi \rho^2}} \exp \left \{ \frac{-\lambda^2}{2\rho^2} \right \}.
\end{equation}
Note that $f_{\rho^2}^+$ is simply the density of a mean-zero Gaussian with variance $\rho^2$, but truncated from below by zero, and multiplied by two to ensure that $f_{\rho^2}^+(\lambda)$ integrates to one.

Then, since mixtures of NSMs are themselves NSMs, the process $\infseq Mt0$ given by
\begin{equation}
    M_t := \int_{\lambda} M_t(\lambda) f_{\rho^2}^+(\lambda) d\lambda
\end{equation}
is an NSM. We will now find a closed-form expression for $M_t$. Many of the techniques used to derive the expression for $M_t$ are identical to Step 2 of the proof of Theorem~\ref{theorem:two-sided-cs-mean-so-far}, but we repeat them here for completeness. To ease notation, define the partial sum ${S_t^\star := \sum_{i=1}^t (Z_i - \zeta(\mu_i^\star))}$. Writing out the definition of $M_t$, we have

\begin{align*}
    M_t &:= \int_{\lambda} \prod_{i=1}^t \exp\left \{ \lambda (Z_i - \zeta(\mu_i^\star)) - \lambda^2/8 \right \}f_{\rho^2}^+(\lambda) d\lambda \\
    &= \int_{\lambda} \exp\left \{ \lambda \underbrace{\sum_{i=1}^t (Z_i - \zeta(\mu_i^\star))}_{S_t^\star} - t\lambda^2/8 \right \}f_{\rho^2}^+(\lambda) d\lambda \\
    &= \int_{\lambda} \1(\lambda > 0)\exp\left \{ \lambda S_t^\star - t\lambda^2/8 \right \} \frac{2}{\sqrt{2\pi \rho^2}} \exp \left \{ \frac{-\lambda^2}{2\rho^2} \right \}d\lambda \\
    &= \frac{2}{\sqrt{2\pi \rho^2}}\int_{\lambda}\1(\lambda > 0) \exp\left \{ \lambda S_t^\star - t\lambda^2/8 \right \} \exp \left \{ \frac{-\lambda^2}{2\rho^2} \right \}d\lambda \\
    &= \frac{2}{\sqrt{2\pi \rho^2}} \int_\lambda \1(\lambda > 0)\exp \left \{ \lambda S_t^\star - \frac{\lambda^2 (t\rho^2/4 + 1)}{2 \rho^2}\right \} d\lambda \\
    &= \frac{2}{\sqrt{2\pi \rho^2}} \int_\lambda \1(\lambda > 0)\exp \left \{ \frac{-\lambda^2 (t\rho^2/4 + 1) + 2\lambda \rho^2 S_t^\star }{2\rho^2} \right \} d\lambda \\
    &= \frac{2}{\sqrt{2\pi \rho^2}} \int_\lambda\1(\lambda > 0) \underbrace{\exp \left \{ \frac{-a(\lambda^2 - \frac{b}{a} 2\lambda) }{2\rho^2} \right \}}_{(\star)} d\lambda,
\end{align*}
where we have set $a:= t\rho^2/4 + 1$ and $b := \rho^2 S_t^\star$. Completing the square in $(\star)$, we have that
\begin{align*}
  \exp \left \{ \frac{-a(\lambda^2 - \frac{b}{a} 2\lambda) }{2\rho^2} \right \} &= \exp \left \{ \frac{-\lambda^2 + 2\lambda \frac{b}{a} + \left ( \frac{b}{a} \right )^2 - \left ( \frac{b}{a} \right )^2 }{2 \rho^2 /a} \right \}\\
                                                                                &= \exp \left \{ \frac{-(\lambda - b/a)^2}{2\rho^2/a} + \frac{a \left ( b/a \right )^2}{2\rho^2} \right \} \\
                                                                                &= \exp \left \{ \frac{-(\lambda - b/a)^2}{2\rho^2/a} \right \} \exp \left \{  \frac{b^2}{2a\rho^2} \right \}.
\end{align*}
Plugging this back into our derivation of $M_t$ and multiplying the entire quantity by $a^{-1/2} / a^{-1/2}$, we have
\begin{align*}
  M_t &= \frac{2}{\sqrt{2\pi \rho^2}} \int_\lambda\1(\lambda > 0) \underbrace{\exp \left \{ \frac{-a(\lambda^2 + \frac{b}{a} 2\lambda) }{2\rho^2} \right \}}_{(\star)} d\lambda\\
  &= \frac{2}{\sqrt{2\pi \rho^2}} \int_\lambda\1(\lambda > 0) \exp \left \{ \frac{-(\lambda - b/a)^2}{2\rho^2/a} \right \} \exp \left \{  \frac{b^2}{2a\rho^2} \right \} d\lambda\\
  &= \frac{2}{\sqrt{a}}\exp \left \{  \frac{b^2}{2a\rho^2} \right \}  \underbrace{\int_\lambda\1(\lambda > 0) \frac{1}{\sqrt{2\pi \rho^2/a}}\exp \left \{ \frac{-(\lambda - b/a)^2}{2\rho^2/a} \right \} d\lambda}_{(\star \star)}.
\end{align*}
Now, notice that $(\star \star) = \PP(N(b/a, \rho^2 / a) \geq 0)$, which can be rewritten as $\Phi(b/\rho\sqrt{a})$, where $\Phi$ is the CDF of a standard Gaussian. Putting this all together and plugging in $a = t\rho^2 / 4 + 1$ and $b = \rho^2 S_t^\star$, we have the following expression for $M_t$,
\begin{align}
  M_t &= \frac{2}{\sqrt{a}} \exp \left \{ \frac{b^2}{2a\rho^2} \right \} \Phi\left( \frac{b}{\rho \sqrt{a}} \right ) \nonumber\\
  &= \frac{2}{\sqrt{t\rho^2 / 4 + 1}} \exp \left \{ \frac{\rho^4 (S_t^\star)^2}{2(t\rho^2 / 4 + 1)\rho^2} \right \}  \Phi\left( \frac{\rho^2 S_t^\star}{\rho \sqrt{t\rho^2 / 4 + 1}} \right ) \nonumber\\
  &= \frac{2}{\sqrt{t\rho^2 / 4 + 1}} \exp \left \{ \frac{\rho^2 (S_t^\star)^2}{2(t\rho^2 / 4 + 1)} \right \} \Phi\left( \frac{\rho S_t^\star}{\sqrt{t\rho^2 / 4 + 1}} \right ).\label{eq:one-sided-mixture-NSM}
\end{align}

\paragraph{Step 3: Deriving a $(1-\alpha)$-lower \cs{} $(L_t')_{t=1}^\infty$ for $(\widetilde \mu_t^\star)_{t=1}^\infty$.} Now that we have computed the mixture NSM $\infseq Mt0$, we apply Ville's inequality to it and ``invert'' a family of processes --- one of which is $M_t$ --- to obtain an \emph{implicit} lower \cs{} (we will further derive an \emph{explicit} lower \cs{} in Step 4).

First, let $\infseq \mu t1$ be an arbitrary real-valued process --- i.e. not necessarily equal to $\infseq {\mu^\star} t1$ --- and define their running average $\widetilde \mu_t := \frac{1}{t} \sum_{i=1}^t \mu_i$. Define the partial sum process in terms of $\infseq {\widetilde \mu}t1$,
\[ S_t(\widetilde \mu_t) := \sum_{i=1}^t Z_i - tr\widetilde \mu_t - t(1-r)/2, \]
and the resulting nonnegative process,
\begin{equation}
  \label{eq:e-process-mean-so-far}
M_t(\widetilde \mu_t) := \frac{2}{\sqrt{t\rho^2 / 4 + 1}} \exp \left \{ \frac{\rho^2 S_t(\widetilde \mu_t)^2}{2(t\rho^2 / 4 + 1)} \right \} \Phi\left( \frac{\rho S_t(\widetilde \mu_t)}{\sqrt{t\rho^2 / 4 + 1}} \right ).
\end{equation}
Notice that if $\infseq {\widetilde \mu} t1 = \infseq {\widetilde \mu^\star} t1$, then $S_t(\widetilde \mu^\star_t) = S_t^\star$ and $M_t(\widetilde \mu_t^\star) = M_t$ from Step 2. Importantly, $(M_t(\widetilde \mu_t^\star))_{t=0}^\infty$ is an NSM\@. Indeed, by Ville's inequality, we have
\begin{equation}
  \label{eq:ville-lower-cs-mean-so-far}
  \PP(\exists t : M_t(\widetilde \mu_t^\star ) \geq 1/\alpha ) \leq \alpha.
\end{equation}
We will now ``invert'' this family of processes to obtain an implicit lower boundary given by
\begin{equation}
  \label{eq:lower-cs-mean-so-far-implicit}
  L_t' := \inf \{\widetilde \mu_t : M_t(\widetilde \mu_t) < 1/\alpha \},
\end{equation}
and justify that $(L_t')_{t=1}^\infty$ is indeed a $(1-\alpha)$-lower \cs{} for $\widetilde \mu_t^\star$. Writing out the probability of miscoverage at any time $t$, we have
\begin{align*}
  \PP(\exists t : \widetilde \mu_t^\star < L_t') &\equiv \PP\left (\exists t : \widetilde \mu_t^\star < \inf_{\widetilde \mu_t} \{ M_t(\widetilde \mu_t) < 1/\alpha \} \right) \\
  &= \PP\left (\exists t : M_t(\widetilde \mu_t^\star) \geq 1/\alpha \right) \\
  &\leq \alpha,
\end{align*}
where the last line follows from Ville's inequality applied to $(M_t(\widetilde \mu_t^\star))_{t=0}^\infty$. In particular, $L_t'$ forms a $(1-\alpha)$-lower \cs{}, meaning
\[ \PP(\forall t,\ \widetilde \mu_t \geq L_t') \geq 1-\alpha. \]

\paragraph{Step 4: Obtaining a closed-form lower bound $(\widetilde L_t)_{t=1}^\infty$ for $(L_t')_{t=1}^\infty$.}
The lower \cs{} of Step 3 is simple to evaluate via line- or grid-searching, but a closed-form expression may be desirable in practice, and for this we can compute a sharp lower bound on $L_t'$.

First, take notice of two key facts:
\begin{enumerate}[(a)]
  \item When $\widetilde \mu_t = \frac{1}{tr}\sum_{i=1}^t Z_i - (1-r)/2r$, we have that $S_t(\widetilde \mu_t) = 0$ and hence $M_t(\widetilde \mu_t) < 1$, and
  \item $S_t(\widetilde \mu_t)$ is a strictly decreasing function of $\widetilde \mu_t \leq \frac{1}{tr}\sum_{i=1}^t Z_i - (1-r)/2r$, and hence so is $M_t(\widetilde \mu_t)$.
\end{enumerate}
 Property (a) follows from the fact that $\Phi(0) = 1/2$, and that $\sqrt{t \rho^2 / 4 + 1} > 1$ for any $\rho > 0$. Property (b) follows from property (a) combined with
 the definitions of $S_t(\cdot)$,
\[ S_t(\widetilde \mu_t) := \sum_{i=1}^t Z_i - tr\widetilde \mu_t - t(1-r)/2, \]
and of $M_t(\cdot)$,
\[ M_t(\widetilde \mu_t) := \frac{2}{\sqrt{t\rho^2 / 4 + 1}} \exp \left \{ \frac{\rho^2 S_t(\widetilde \mu_t)^2}{2(t\rho^2 / 4 + 1)} \right \} \Phi\left( \frac{\rho S_t(\widetilde \mu_t)}{\sqrt{t\rho^2 / 4 + 1}} \right ), \]
In particular, by facts (a) and (b), the infimum in \eqref{eq:lower-cs-mean-so-far-implicit} must be attained when $S_t(\cdot) \geq 0$. That is, \begin{equation}
  \label{eq:sum_positive_at_arginf}
  S_t(L_t') \geq 0.
\end{equation}
Using \eqref{eq:sum_positive_at_arginf} combined with the inequality $1 - \Phi(x) \leq \exp \{ -x^2 / 2\}$ (a straightforward consequence of the Cram\'er-Chernoff technique), we have the following lower bound on $M_t(L_t')$:
\begin{align*}
  M_t(L_t') &= \frac{2}{\sqrt{t\rho^2 / 4 + 1}} \exp \left \{ \frac{\rho^2 S_t(L_t')^2}{2(t\rho^2 / 4 + 1)} \right \} \Phi\left( \frac{\rho S_t(L_t')}{\sqrt{t\rho^2 / 4 + 1}} \right ) \\
            &\geq \frac{2}{\sqrt{t\rho^2 / 4 + 1}} \exp \left \{ \frac{\rho^2 S_t(L_t')^2}{2(t\rho^2 / 4 + 1)} \right \} \left ( 1 - \exp\left \{ -\frac{\rho^2 S_t(L_t')^2}{2(t\rho^2 / 4 + 1)} \right \} \right ) \\
            &= \frac{2}{\sqrt{t\rho^2 / 4 + 1}} \left ( \exp \left \{ \frac{\rho^2 S_t(L_t')^2}{2(t\rho^2 / 4 + 1)} \right \}  - 1 \right )\\
            &=: \widecheck M_t(L_t').
\end{align*}
Finally, the above lower bound on $M_t(L_t')$ implies that $1/\alpha \geq M_t(L_t') \geq \widecheck M_t(L_t')$ which yields the following lower bound on $L_t'$:
\begin{align*}
  \widecheck M_t(L_t') \leq 1/\alpha &\iff \frac{2}{\sqrt{t\rho^2 / 4 + 1}} \left ( \exp \left \{ \frac{\rho^2 S_t(L_t')^2}{2(t\rho^2 / 4 + 1)} \right \}  - 1 \right ) \leq 1/\alpha \\
  &\iff \exp \left \{ \frac{\rho^2 S_t(L_t')^2}{2(t\rho^2 / 4 + 1)} \right \} \leq 1 + \frac{\sqrt{t\rho^2 / 4 + 1}}{2\alpha}\\
  &\iff \frac{\rho^2 S_t(L_t')^2}{2(t\rho^2 / 4 + 1)} \leq \log \left (1 + \frac{\sqrt{t\rho^2 / 4 + 1}}{2\alpha} \right ) \\
  &\iff S_t(L_t') \leq \sqrt{\frac{2(t\rho^2 / 4 + 1)}{\rho^2} \log \left (1 + \frac{\sqrt{t\rho^2 / 4 + 1}}{2\alpha} \right ) }\\
  &\iff \sum_{i=1}^t Z_i - tr L_t' - t(1-r)/2 \leq \sqrt{\frac{2(t\rho^2 / 4 + 1)}{\rho^2} \log \left (1 + \frac{\sqrt{t\rho^2 / 4 + 1}}{2\alpha} \right ) }\\
  &\iff tr L_t'  \geq \sum_{i=1}^t Z_i- t(1-r)/2 -\sqrt{\frac{2(t\rho^2 / 4 + 1)}{\rho^2} \log \left (1 + \frac{\sqrt{t\rho^2 / 4 + 1}}{2\alpha} \right ) }\\
  &\iff L_t' \geq \frac{\sum_{i=1}^t (Z_i - (1-r)/2)}{tr} -\sqrt{\frac{2(t\rho^2 / 4 + 1)}{(tr\rho)^2} \log \left (1 + \frac{\sqrt{t\rho^2 / 4 + 1}}{2\alpha} \right ) }\\
  &\iff L_t' \geq \underbrace{\frac{\sum_{i=1}^t (Z_i - (1-r)/2)}{tr} -\sqrt{\frac{t\beta^2 + 1}{2(tr\beta)^2} \log \left (1 + \frac{\sqrt{t\beta^2 + 1}}{2\alpha} \right ) }}_{\widetilde L_t},
\end{align*}
where we set $\rho = 2\beta$ in the right-hand side of the final inequality. This precisely yields $\widetilde L_t$ as given in Proposition~\ref{proposition:one-sided-cs-mean-so-far}, completing the proof.

\end{proof}

\subsection{Proof of \cref{theorem:a/b-test}}\label{proof:variance-adaptive-a/b-test}
\AdaptiveAbTest*
\begin{proof}

  The proof proceeds in three steps and follows a similar form to the proof of \citet[Theorem 2]{waudby2022anytime}. First, we show that a collection of processes (indexed by $\lambda \in (0, 1)$) each form $\QcalNPRRstarinf$-NSMs with respect to the private filtration $\Zcal$. Second, we mix over $\lambda \in (0, 1)$ using the truncated gamma density to obtain the NSM obtained in \cref{theorem:a/b-test}. Third and finally, we ``invert'' the aforementioned NSM to obtain the \lpcs{} of \cref{theorem:a/b-test}.

  \paragraph{Step 1: Showing that $M_t^\lambda$ forms an NSM.}
  Consider the process $\infseqt{M_t^\lambda}$ given by
  \begin{equation}
    M_t^\lambda := \prod_{i=1}^t \exp \left \{ \lambda (\theta_i - \Delta_i) - \psi_E(\lambda) (\theta_i - \widehat \theta_{i-1})^2/4\right \}.
  \end{equation}
  We will show that $\infseqt{M_t^\lambda}$ forms an NSM\@.
First, note that $M_0^\lambda \equiv 1$ by construction, and $M_t^\lambda$ is always positive. It remains to show that $M_t^\lambda$ forms a supermartingale. Writing out the conditional expectation of $M_t^\lambda$ given $\Zcal_{t-1}$, we have that
\begin{align}
    \EE(M_t^\lambda \mid \Zcal_{t-1}) = M_{t-1}^\lambda \underbrace{\EE \left [ \exp \left \{ \lambda (\theta_t - \Delta_t) - \psi_E(\lambda) (\theta_t - \widehat \theta_{t-1})^2/4\right \} \mid \Zcal_{t-1} \right ]}_{(\dagger)},
\end{align}
and hence it suffices to prove that $(\dagger) \leq 1$.
Denote for the sake of succinctness,
\[\xi_{t} := \theta_t - \Delta_t ~~~\text{ and }~~~ \eta_t := \widehat \theta_{t-1} - \Delta_t, \]
and note that $\EE(\xi_t \mid \Zcal_{t-1}) = 0$.
Using the proof of \citet[Proposition 4.1]{fan2015exponential}, we have that $\exp\{b\lambda  - b^2 \psi_{E} (\lambda)\} \leq 1 + b \lambda$ for any $\lambda \in [0, 1)$ and $b \geq -1$. Noticing that $(\theta_t - \widehat \theta_{t-1})/2 \geq -1$ and setting $b := (\xi_{t} - \eta_{t})/2 = (\theta_t- \widehat \theta_{t-1})/2$, we have that
\begin{align*}
	&\EE \left [\exp \left \{ \lambda \xi_{t} - \psi_{E}(\lambda )(\xi_{t} - \eta_{t})^2 /4 \right \}  \Bigm | \Zcal_{t-1} \right ]\\
	=\ &\EE \left [\exp \left \{ \lambda (\xi_{t}-\eta_{t}) - \psi_{E}(\lambda)(\xi_{t} - \eta_{t})^2 /4 \right \}  \bigm | \Zcal_{t-1} \right ] \exp(\lambda \eta_{t})\\
  \leq \ &\EE\left  [1 + (\xi - \eta_{t} )\lambda \mid \Zcal_{t-1} \right ]\exp(\lambda \eta_{t}) \\
  = \ &\EE\left [1 - \eta_{t} \lambda \mid \Zcal_{t-1} \right ]\exp(\lambda \eta_{t}) \leq 1,
\end{align*}
where the last line follows from the fact that $\xi_{t}$ is conditionally mean zero
and the inequality $1-x \leq \exp(-x)$ for all $x \in \mathbb R$. This completes Step 1 of the proof.

\paragraph{Step 2: Mixing over $\lambda$ using the truncated gamma density.}
For any distribution $F$ on $(0, 1)$,
\begin{align}
    \widetilde M_t^\EB := \int_{\lambda \in (0, 1)} M_t^\lambda \dd F(\lambda)
\end{align}
forms a test supermartingale by Fubini's theorem. In particular, we will use the truncated gamma density $f(\lambda)$ given by
\begin{equation}
f(\lambda) = \frac{\rho^\rho e^{-\rho \left(1 - \lambda\right)} \left(1 - \lambda\right)^{\rho - 1}}{\Gamma(\rho) - \Gamma(\rho, \rho)}, 
\end{equation}
as the mixing density. Writing out $\widetilde M_t^\EB \equiv \widetilde M_t^\EB(\widetilde \Delta_t)$ using $dF(\lambda) := f(\lambda)d\lambda$, and using the shorthand $S_t \equiv S_t(\widetilde \Delta_t)$, we have 
\begin{align*}
\widetilde M_t^\EB &:= \int_0^1 \exp \left \{ \lambda S_t - V_t \psi_E(\lambda) \right \} f(\lambda) \dd\lambda \\
&=\int_0^1 \exp \left \{ \lambda S_t - V_t \psi_E(\lambda) \right \} \frac{\rho^\rho e^{-\rho \left(1 - \lambda\right)} \left(1 - \lambda\right)^{\rho - 1}}{\Gamma(\rho) - \Gamma(\rho, \rho)} \dd\lambda \\
&= \frac{\rho^\rho e^{-\rho}}{\Gamma(\rho) - \Gamma(\rho, \rho)} \int_0^1 \exp \{{\lambda \left(\rho + S_t + V_t\right)}\} \left(1 - \lambda\right)^{V_t + \rho - 1} \dd\lambda \\
&= \left(\frac{\rho^\rho e^{-\rho}}{\Gamma(\rho) - \Gamma(\rho, \rho)}\right) \left(\frac{1}{V_t + \rho}\right) \left.\left(  \frac{\Gamma(b)}{\Gamma(a) \Gamma(b-a)} \int_0^1 e^{z u} u^{a-1} (1 - u)^{b - a - 1} \dd u \right)\right|_{\substack{a = 1 \\ b = V_t + \rho + 1 \\ z = S_t + V_t + \rho }} \\
&= \left(\frac{\rho^\rho e^{-\rho}}{\Gamma(\rho) - \Gamma(\rho, \rho)}\right) \left(\frac{1}{V_t + \rho}\right) \onefone(1, V_t + \rho + 1, S_t + V_t + \rho),
\end{align*}
which completes this step.

\paragraph{Step 3: Applying Ville's inequality and inverting the mixture NSM.}
Notice that $\widetilde \Delta_t < \widetilde L_t^\Delta$ if and only if $\widetilde M_t(\widetilde \Delta_t) \geq 1/\alpha$, and hence by Ville's inequality for nonnegative supermartingales \citep{ville1939etude}, we have that
\begin{equation*}
    \PP(\exists t : \widetilde \Delta_t < \widetilde L_t^\Delta) = \PP(\exists t : \widetilde M_t^\EB \geq 1/\alpha) \leq \alpha,
\end{equation*}
and hence $\widetilde L_t^\Delta$ forms a lower $(1-\alpha, \eps)$-\lpcs{} for $\widetilde \Delta_t$. This completes the proof.
\end{proof}
  


\section{A more detailed survey of related work}
\label{section:related-detailed}
There is a rich literature exploring the intersection of statistics and differential privacy. \citet{wasserman2010statistical} studied DP estimation rates under various metrics for several privacy mechanisms. \citet{duchi2013local-FOCS,duchi2013local-NeurIPS,duchi2018minimax} articulated a new ``locally private minimax rate'' --- the fastest worst-case estimation rate with respect to any estimator \emph{and} LDP mechanism together --- and studied them in several estimation problems. To accomplish this they provide locally private analogues of the famous Le Cam, Fano, and Assouad bounds that are common in the nonparametric minimax estimation literature. As an example application, \citet{duchi2013local-FOCS,duchi2013local-NeurIPS,duchi2018minimax} derived minimax rates for nonparametric density estimation in Sobolev spaces, and showed that a naive application of the Laplace mechanism cannot achieve said rates, but a different carefully designed mechanism can. This study of density estimation was extended to Besov spaces by \citet{butucea2020local}. \citet{butucea2021locally} employed this minimax framework to study the fundamental limits of private estimation of nonlinear functionals.
\citet{acharya2021differentially} extended the locally private Le Cam, Fano, and Assouad bounds to central DP setting. \citet{duchi2018right} developed a framework akin to \citet{duchi2013local-FOCS,duchi2013local-NeurIPS,duchi2018minimax} but from the \emph{local} minimax point of view (here, ``local'' refers to the type of minimax rate considered, not ``local DP''). \citet{barnes2020fisher} studied the locally private Fisher information for parametric models. All of the aforementioned works are focused on estimation rates, rather than inference --- i.e. confidence sets, $p$-values, and so on (though some asymptotic inference is implicitly possible in \citep{duchi2018right,barnes2020fisher}).

The first work to explicitly study inference under DP constraints was \citet{vu2009differential}, who developed asymptotically valid private hypothesis tests for some parametric problems, including Bernoulli proportion estimation, and independence testing. Several works have furthered the study of differentially private goodness-of-fit and independence testing~\citep{wang2015revisiting,gaboardi2016differentially,berrett2020locally,amin2020pan,acharya2020domain,acharya2020inference,acharya2021inference,acharya2022interactive}. \citet{couch2019differentially} develop nonparametric tests of independence between categorical and continuous variables. \citet{awan2018differentially} derive private uniformly most powerful nonasymptotic hypothesis tests in the binomial case. \citet{karwa2018finite}, \citet{gaboardi2019locally}, and \citet{joseph2019locally} study nonasymptotic \ci{}s for the mean of Gaussian random variables. \citet{canonne2019structure} study optimal private tests for simple nulls against simple alternatives. \citet{covington2021unbiased} derive nonasymptotic \ci{}s for parameters of location-scale families. \citet{ferrando2020general} introduces a parametric bootstrap method for deriving asymptotically valid \ci{}s.

All of the previously mentioned works either consider goodness-of-fit testing, independence testing, or parametric problems where distributions are known up to some finite-dimensional parameter. \citet{drechsler2021non} study nonparametric \ci{}s for medians. To the best of our knowledge, no prior work derives private nonasymptotic \ci{}s (nor \cs{}s) for means of bounded random variables.

Moreover, like most of the statistics literature, the prior work on private statistical inference is non-sequential, with the exception of \citet{wang2020differential} who study private analogues of Wald's sequential probability ratio test \citep{wald1945sequential} for simple hypotheses, and \citet{berrett2021locally} who study locally private online changepoint detection. Another interesting paper is that of \citet[Sections 7.1 \& 7.2]{jun2019parameter} --- the authors study online convex optimization with sub-exponential noise, but also consider applications to martingale concentration inequalities (and thus \cs{}s) as well as locally private stochastic subgradient descent.


\end{document}